\documentclass[
 amsmath,amssymb,
]{revtex4-1}

\usepackage{graphicx}
\usepackage{dcolumn}
\usepackage{bm}

\usepackage[utf8]{inputenc}
\usepackage[T1]{fontenc}
\usepackage{mathptmx}
\usepackage{mathtools,bbm,amsthm}

\newcommand{\bE}{\mathbb{E}}

\newcommand{\cH}{\mathcal{H}}

\newcommand{\ket}[1]{| #1 \rangle}
\newcommand{\bra}[1]{\langle #1|}

\newcommand{\bracket}[3]{\langle #1|#2|#3 \rangle}

\newcommand{\avr}[1]{\left \langle#1 \right \rangle}

\newcommand{\1}{\mathbbm{1}}

\newcommand{\Tr}{\operatorname{Tr}}

\newcommand{\Real}{{\textrm{Re}}}
\newcommand{\Imag}{{\textrm{Im}}}
\newcommand{\mywedge}{{\textrm{{\large $\wedge$}}\!}^2}

\newcommand{\be}{\begin{equation}}
\newcommand{\ee}{\end{equation}}
\newcommand{\bea}{\begin{eqnarray}}
\newcommand{\eea}{\end{eqnarray}}
\newcommand{\bes}{\begin{equation*}}
\newcommand{\ees}{\end{equation*}}
\newcommand{\beas}{\begin{eqnarray*}}
\newcommand{\eeas}{\end{eqnarray*}}

\newtheorem{thm}{Theorem}
\newtheorem{corollary}[thm]{Corollary}

\newtheorem{lemma}[thm]{Lemma}

\newtheorem{prop}[thm]{Proposition}
\newtheorem{definition}[thm]{Definition}

\begin{document}

\title[Classical restrictions of generic matrix product states are quasi-locally Gibbsian]{Classical restrictions of generic matrix product states are quasi-locally Gibbsian}

\author{Y. Aragonés-Soria}
 \email{yaiza.aragonessoria@gmail.com}
 \altaffiliation{Institute for Theoretical Physics, University of Cologne, Zülpicher Str. 77, 50937 Köln, Germany.}
 
\author{J. \AA berg}%
\affiliation{Institute for Theoretical Physics, University of Cologne, Zülpicher Str. 77, 50937 Köln, Germany.
}%

\author{C-Y. Park}
\affiliation{%
Institute for Theoretical Physics, University of Cologne, Zülpicher Str. 77, 50937 Köln, Germany.
}%

\author{M. J. Kastoryano}
\affiliation{Institute for Theoretical Physics, University of Cologne, Zülpicher Str. 77, 50937 Köln, Germany.
}%
\affiliation{Amazon Quantum Solutions Lab, Seattle, Washington 98170, USA}
\affiliation{AWS Center for Quantum Computing, Pasadena, California 91125, USA}

\date{\today}

\begin{abstract}
We show that the norm squared amplitudes with respect to a local orthonormal basis (the classical restriction) of finite quantum systems on  one-dimensional lattices can be exponentially well approximated by Gibbs states of local Hamiltonians (i.e., are quasi-locally Gibbsian) if the classical conditional mutual information (CMI) of any connected tripartition of the lattice is rapidly decaying in the width of the middle region. For injective matrix product states, we moreover show that the classical CMI decays exponentially, whenever the collection of matrix product operators satisfies a `purity condition'; a notion previously  established in the theory of random matrix products. We furthermore show that violations of the purity condition enables a generalized notion of error correction on the virtual space, thus indicating the non-generic nature of such violations. We make this intuition more concrete by constructing a probabilistic model where purity is a typical property. The proof of our main result makes extensive use of the theory of random matrix products, and may find applications elsewhere. 
\end{abstract}

\maketitle

\section{Introduction}

Considerable effort has been devoted to understanding the entanglement properties of many-body quantum states. For finite  one-dimensional-lattice systems, the theory of Matrix Product States (MPSs) provides a complete framework for describing entanglement of gapped many-body systems \cite{HastingsMPSGap}, and allows for efficient high precision simulations via the DMRG algorithm \cite{DMRG0,DMRG1,DMRG2}. Similarly impressive degrees of numerical precision can be reached in other settings, such as disorder \cite{DMRGDisorder}, open systems \cite{OpenMPS}, time evolution \cite{MPSTev}, or critical systems \cite{DMRGCrit}. The success of these simulation methods can be traced back to the accurate parametrization of entanglement in MPSs.   There exist  extensions to  lattices of higher dimensions (projected entangled pair states) but these have been far less useful for simulations, due to their extensive entanglement growth. 

In contrast, quantum Monte-Carlo simulations are largely based on heuristic assumptions on the weights and phases of the underlying state. Indeed, if the system under study can be cast in a form with only positive weights, then Monte-Carlo methods often work well, though  convergence guarantees are only known in very special cases \cite{Bravyi}. This in turn is believed to be due to the local Gibbsian nature of the classical restriction of the state. Classical Monte-Carlo sampling is known to converge rapidly for Ising type problems \cite{Jerrum,Martinelli}, while quantum variational Monte-Carlo is often successful when using a locally restricted Gibbs Ansatz, such as the Jastrow-Ansatz. Further evidence of the importance of locality in the Ansatz wavefunction has been observed for more expressive Ans\"{a}tze, such as the complex Restricted Boltzmann machine \cite{Carleo}, where the activations naturally preserve locality in many cases.  Hence, whereas tensor network states explicitly encode the local entanglement structure in their construction,  quantum (variational) Monte Carlo implicitly invokes locality through the pervasive Gibbsian nature of probability distributions. 

Here, we connect these two pictures by showing that generic injective MPSs~\cite{injectivity} have classical restrictions that are quasi-locally Gibbsian. More precisely, we here refer to a probability distribution as locally Gibbsian if it can be written as the equilibrium distribution of a local Hamiltonian, i.e., as a sum of terms that each spans at most $\ell$ adjacent sites. Well known examples include the Ising and Potts models. We similarly say that a distribution is quasi-locally Gibbsian, if it can be approximated by  Gibbs distributions corresponding to local Hamiltonians, $h^{\ell}$, where the error of the approximation in some sense decays exponentially with increasing $\ell$. Such notions appear in various guises in the literature, e.g., Ref.~\onlinecite{Kozlov}, which  requires that  the coefficients in the cluster expansion of $\log(p)$ are rapidly decaying with the order of the cluster.

As the first step towards proving the generic quasi-local Gibbs property of injective MPSs, we show (in Section \ref{sec:CMI}) that probability distributions on a one-dimensional lattice with open boundary conditions are quasi-locally Gibbsian if the Conditional Mutual Information (CMI) between any tripartition of the lattice is decaying rapidly in the width of the middle region. The stronger the decay of the CMI, the more local the Gibbs distribution. In the case of zero correlation length, the distribution is (strictly) locally Gibbs  \cite{Poulin}. 
A number of recent studies in quantum information theory have revealed connections between the CMI and the Gibbsian nature of density matrices. In Ref.~\onlinecite{BrandaoKato}, the authors show that the quantum CMI  of a full rank density matrix on a one-dimensional lattice is small if and only if the state is Gibbsian. The Gibbsian nature of states has important implications for the nature of  edge states of topologically ordered systems \cite{KatoEdge,MePEPS}. Our results show that similar equivalences hold for classical restrictions of quantum states. Similar conclusions can be reached using perturbative methods for classical restrictions of high temperature quantum Gibbs states of gaped spin chains \cite{maes}.

The second step towards establishing the  quasi-locality is also our main result; that the classical restriction of injective MPSs have an exponentially decaying CMI if the matrices associated to the MPS satisfy a condition referred to as \textit{purity} (see Def.~\ref{DefPur}). This condition  has previously been shown \cite{Benoist,Maassen} to imply the `purification' of  quantum trajectories resulting from the applications of sequences of random matrices on an initial state. In our setting,  the classical CMI can be bounded by a corresponding quantum CMI. The latter can, in turn, be rewritten in terms of the expected entanglement entropy after  measurements on the conditional subsystem.  A vanishing entanglement entropy is thus equivalent to the purification of the state-trajectory induced by the sequence of measurements  on the virtual system. The purification of trajectories implies that the system asymptotically jumps between pure states of a specific stationary measure, irrespective of what (mixed) state the system started in. We are currently not aware of a meaningful operational interpretation of the stationary stochastic process, and believe it to be quite hard to evaluate in practice \cite{Benoist2}. Furthermore, and perhaps counter-intuitively, we observe that the rate of decay towards the stationary measure is unrelated to the gap of the transfer operator of the MPS. We moreover do not know of a closed functional form for the decay rate, in terms of the matrices associated to the MPS. 

One may note that our setting, which focuses on the degree of conditional post-measurement entanglement, is closely related to the notion of localizable entanglement \cite{locent0,locent1,locent2}. The latter is obtained by optimizing the measurements over all possible local bases, while we consider a fixed basis. However, to the best of our knowledge, a general proof of the exponential decay of the localizable entanglement has not been shown previously.

As mentioned above, we show that the purity condition is sufficient for the exponential decay of the \emph{classical} CMI. However,  it is less clear if it also is a sufficient condition. In order to better pinpoint the significance of the purity condition, we show that it in essence is both necessary and sufficient for the exponential decay of the above mentioned \emph{quantum} CMI.

As a further attempt to gain a better understanding of the purity condition, we moreover investigate the conspicuous similarity between   (the violation of) the purity condition (see Def.~\ref{DefPur}) and the Knill-Laflamme error correction condition \cite{KLcondition}. Indeed, we find (in Section \ref{sec:errorcorr}) that the purity-condition can be regarded as the non-existence of a non-trivial correctable subspace that persists indefinitely throughout iterated applications of an error-model, in a somewhat unconventional error correction scenario. One may note that invariant subspaces are special cases of such correctable spaces. As an example, MPSs with Symmetry-Protected Topological (SPT) order are associated to invariant subspaces \cite{Else} and would thus violate the purity condition. 
The above results suggests that violations of the purity condition in some sense are `fragile'. In order to shed some further light on this question, we construct  a  probabilistic model (in Sec.~\ref{SecModelTypicality}), where the purity condition holds, apart for a subset of measure zero.
 
The proof of our main theorem relies heavily on the theory of random matrix products, and in particular on the work of Benoist et.~al.~\cite{Benoist} and Maassen and Kümmerer~\cite{Maassen}. Since these results involve notions from probability theory that likely are unfamiliar to most of the quantum information community, we reproduce in Appendixes \ref{app:A}-\ref{AppProofPropMain} many of the basic results in a language that should be more familiar to the quantum-information reader. We hope that this will facilitate the access to a rich and extensive body of work that should see many more applications in the fields of quantum information and many body physics. For instance, the theory of random matrix products  has recently been leveraged in a different setting, to show ergodicity for ensembles of quantum channels \cite{Ramis1,Ramis2}. 

Concerning the structure of the paper, we begin by introducing the notation in Section \ref{s:notation}, while Section \ref{s:CondMutInf} focuses on the central object in this investigation, namely the CMI with respect to classical restrictions of MPSs. Section \ref{sec:CMI} presents the first result of the paper: an exponentially decaying CMI implies quasi-local Gibbs distributions. Section \ref{s:keytheorem} is devoted to the main result, namely the exponentially decaying CMI for a broad class of MPSs. Section \ref{s:Discussion} provides  examples and observations, where we in Section \ref{ss:SPTphases} observe that MPSs corresponding to SPT phases violate the purity condition. In Section \ref{sec:errorcorr} we further investigate the purity condition and show that its violation can be regarded as a type or error-correction condition. Section \ref{PuritySufficient} is devoted to give a sufficient condition for purity to hold for a set of operators in terms of the span of the operators. We use this relation in Section \ref{SecModelTypicality}, where we construct a model of typicality of purity, to prove that purity is a generic property. Section \ref{ss:RateOfDecay} compares the convergence rate of the CMI with the rate of the converge to the fixed point of the transfer operator. Concrete examples are provided in Section \ref{ss:Exapmles}. We finish with an outlook in Section \ref{s:Outlook}.

\section{Notation}\label{s:notation}
We consider pure states defined on a finite one-dimensional lattice, $\Lambda$, and associate a finite dimensional Hilbert space of dimension $d$ to each site. We index the sites of the lattice according to a tripartition of the lattice $\Lambda=ABC$ as follows: we denote sites in region $A$ as $-|A|+1,-|A|+2,\dots,-1,0$; sites in region $B$ as $1,\dots,N$; and sites in region $C$ as $N+1,\dots,|BC|$ (see Fig. \ref{f:lattice}).  This peculiar indexing of sites will make sense later on when considering the CMI for MPSs. 

\begin{figure}[htb]
\centering
\includegraphics[scale=1]{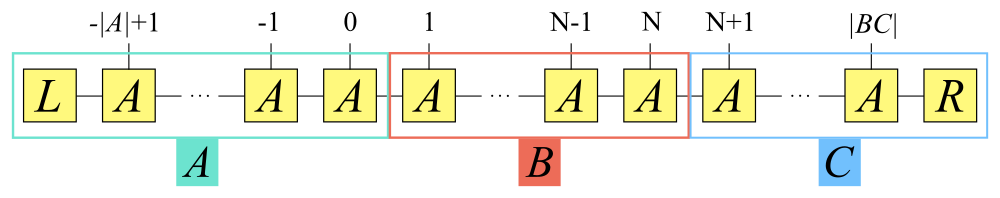}
\caption{We consider a MPS on a finite lattice, $\Lambda$, which is broken up into three contiguous regions such that $\Lambda=ABC$. We denote sites in region $A$ as $-|A|+1,-|A|+2,\dots,-1,0$; sites in region $B$ as $1,\dots,N$; and sites in region $C$ as $N+1,\dots,|BC|$.}
\label{f:lattice}
\end{figure}

Let $\ket{x_\Lambda}=\ket{x_{-|A|+1},\ldots,x_{0},x_{1},\ldots,x_N,\ldots,x_{|BC|}}$ be a local orthonormal basis, where $\{\ket{x_i}\}_{x_i=0}^{d-1}$ is the local basis at site $i$. Unless specified otherwise, we will be working with translationally invariant MPSs with open boundary conditions
\be \label{e:MPS}
\ket{\Psi}=\dfrac{1}{K}\sum_{x_{-|A|+1},\dots,x_{|BC|}=0}^{d-1}\bra{R}A_{x_{|BC|}}\cdots A_{x_{-|A|+1}}\ket{L}\ket{x_{-|A|+1}\cdots  x_{|BC|}},
\ee
where $K$ is a normalization factor.
Here, $A_{x_i}$ are $D\times D$ matrices encoding correlations in the system and $\ket{L}$ and $\ket{R}$ are normalized states on the $D$-dimensional virtual space specifying the boundary conditions, where $D$ is known as the bond dimension of the MPS. Without loss of generality, we consider (left-)normalized MPSs, which enforces that $\sum_{x_i=0}^{d-1} A^\dag_{x_i} A_{x_i}=\1$. Left normalization guarantees that the completely positive map
\begin{align}\label{e:TransferOperator}\bE(\cdot):=\sum_{x_i=0}^{d-1}A_{x_i}\cdot A_{x_i}^\dag
\end{align}
is trace preserving. The map $\bE$ is often referred to as the transfer operator and maps density matrices on the virtual space to density matrices from left to right. The adjoint map, $\bE^*$, maps operators from right to left along the chain. Our choice of boundary conditions serves mainly for notational simplicity. The results in the paper extend naturally to periodic or mixed boundary conditions. For periodic boundary conditions, the regions $ABC$ need to be chosen differently to ensure that $B$ separates $A$ from $C$. 

The normalization constant can be expressed concisely as $K^2=\Tr\left[\bE^{|\Lambda|}\left(L\right)R\right]$, where we use the shorthand notation $R=\ket{R}\bra{R}$ and $L=\ket{L}\bra{L}$.

\paragraph{Classical Restrictions}

For a given local basis $\{\ket{x_\Lambda}\}$, we define the quantum channel
\begin{align}\label{e:channel}
\Phi_\Lambda(\psi)&=\sum_ {x_\Lambda} \ket{x_\Lambda}\bra{x_\Lambda} \bra{x_\Lambda} \psi \ket{x_\Lambda}.
\end{align}
In other words, $\Phi_\Lambda$ generates a state that is diagonal with respect to  the basis $\{\ket{x_\Lambda}\}$, by deleting the off-diagonal elements of the input $\psi$. We refer to $\Phi_\Lambda$ as the classical restriction (also commonly referred to as a `dephasing map' or `pinching'). Since $\Phi_{\Lambda}(\psi)$ is diagonal, the map $\Phi_\Lambda$ effectively defines a classical probability distribution, $p_\psi( x_\Lambda)=\bra{x_\Lambda} \psi \ket{x_\Lambda}$, for any choice of basis $\{|x_\Lambda\rangle\}$.

We also consider the channel that measures a subset of systems  $B\subset\Lambda$ and we denote it as 
\begin{equation}\label{e:reducedchannel}
\begin{split}
    \Phi_B(\psi)&=\sum_{x_B} \ket{x_B}\bra{x_B}\bra{x_B}\psi\ket{x_B},\\
    &=\sum_{x_B}p_\psi(x_B)\psi(x_B),
\end{split}
\end{equation} 
with $ \ket{x_B}=\bigotimes_{i\in B} \ket{x_i}$.  Here, the channel $\Phi_B$ similarly defines a classical probability distribution on the sites in $B$ by $p_\psi(x_B)=\bra{x_B} \psi_B \ket{x_B}$, where $\psi_B:=\Tr_{\Lambda\setminus AC}\psi$ is the reduced state of $\psi$ on $B$. Note that 
\be
p_\psi(x_B)=\sum_{x_{AC}} p_\psi(x),
\ee
where recall that $\Lambda=ABC$.
Moreover, we refer to the post-measurement state after obtaining the measurement outcome $x_B$ as 

\be\label{e:post-measurementstate}
\psi(x_B)=\dfrac{1}{p_\psi(x_B)}\ket{x_B}\bra{x_B}\otimes\bra{x_B}\psi\ket{x_B}.
\ee

Consider now the MPS defined in  Eq.~(\ref{e:MPS}). The probability distribution on $B$ is
\be\label{e:prob}
p_\Psi(x_B)=\dfrac{1}{K^2}\Tr\left[A_{x_N}\cdots A_{x_1}\bE^{|A|}(L)A^\dag_{x_1}\cdots A^\dag_{x_N}\bE^{*|C|}(R)\right],
\ee
where $\bE^n$ is understood as convolution of the map and $x_B:= x_1,\dots,x_N$, with $x_i=0,\dots,d-1$.

In this investigation we primarily focus on translationally invariant injective MPSs \cite{injectivity}, which equivalently can be defined via the primitivity of the map $\mathbb{E}$~\cite{primitivity}. The latter means that $\mathbb{E}$ has a unique full-rank fixed point, i.e., there exists a unique full-rank density operator $\rho$ such that $\mathbb{E}(\rho)= \rho$. A consequence of the injectivity of the MPS is thus that  $\lim_{|A|\rightarrow\infty}\mathbb{E}^{|A|}(\chi) = \rho\Tr(\chi)$, $\lim_{|C|\rightarrow\infty}\mathbb{E}^{*|C|}(Q) = \1\Tr(Q\rho)$, and $\lim_{|A|\rightarrow\infty,|C|\rightarrow\infty}K^2 =  \Tr(R\rho)\neq 0$. Hence, if  region $B$ is kept fixed, while regions $A$ and $C$ both grow to infinity, the probability distribution (\ref{e:prob}) on $B$ reduces to 
\be
p_\Psi(x_B)=\Tr\left[A_{x_N}\cdots A_{x_1}\rho A^\dag_{x_1}\cdots A^\dag_{x_N}\right].
\ee

\section{\label{s:CondMutInf} The post-measurement conditional mutual information}
Throughout the paper, we use a number of entropic quantities, which we introduce in this section. In particular, we switch back and forth between classical and quantum systems. The quantum von Neumann entropy of a mixed state, $\chi$, is denoted as $S(\chi)=-\Tr{\chi\log\chi}$, while the classical entropy is referred to as $H(p)=-\sum_ x p(x)\log p(x)$ for a classical probability distribution $p(x)$. Here, $\log$ denotes the natural logarithm. In Sec. \ref{sec:CMI}, we use the classical relative entropy as a measure of distinguishability between probability distributions. The classical relative entropy of $p_1(x)$ with respect to $p_2(x)$ is defined as
\be\label{e:relativeentropy}
S(p_1||p_2)=\sum_x p_1(x)\log\left[\dfrac{p_1(x)}{p_2(x)}\right].
\ee
The (quantum) CMI between regions $A$ and $C$ conditioned on region $B$, is given by 
\be
I_\chi(A:C|B)=S(\chi_{AB})+S(\chi_{BC})-S(\chi_{B})-S(\chi_{ABC}).
\ee
After applying the classical conditioning map, $\Phi_\Lambda$ in Eq.~(\ref{e:channel}), on a quantum state, $\chi$, we get the classical CMI
\begin{align}
    I_{\Phi_{\Lambda}(\chi)}(A:C|B)&=I_{p_\chi}(A:C|B)=H(p_{\chi,AB})+H(p_{\chi,BC})-H(p_{\chi,B})-H(p_{\chi,ABC}),
\end{align}
where $p_{\chi,A}:= p_\chi(x_A)=\bra{x_A}\chi_A\ket{x_A}$.

We now point out an important observation on the CMI \cite{Hastings}. Suppose that we have a pure state, $\psi=\ket{\psi}\bra{\psi}$, and we measure all spins in region $B$. Then, the quantum CMI of the post-measurement state satisfies
\begin{equation}
\label{e:Ipartials}
\begin{split}
    I_{p_{\psi}}(A:C|B)&\leq
    I_{\Phi_B(\psi)}(A:C|B),\\
    &=\langle S\left[\psi_A(x_B)\right]\rangle_{p_\psi(x_B)}+\langle S\left[\psi_C(x_B)\right]\rangle_{p_\psi(x_B)},\\
    &=2\langle S\left[\psi_C(x_B)\right]\rangle_{p_\psi(x_B)},
\end{split}
\end{equation}
where the distribution $p_{\psi}$ is defined by $p_{\psi}(x_{\Lambda}) = \langle x_{\Lambda}|\psi| x_{\Lambda}\rangle$, and where the state $\psi_{X}(x_B) = \langle x_B|\rho_{XB}|x_B\rangle$ is the reduced state in region $X$ of the post-measurement state, $\psi(x_B)$, and $\langle S\left[\psi(x)\right]\rangle_{p_{\psi(x)}}$ is the average von Neumann entropy of $\psi(x)$ over $p_\psi(x)$, i.e.,
\be\label{e:averageS}
\langle S\left[\psi(x)\right]\rangle_{p_\psi(x)}:=\sum_{x}p_\psi(x)S\left[\psi(x)\right].
\ee

The inequality in Eq.~(\ref{e:Ipartials}) comes from monotonicity of the relative entropy. Note that $S\left[\psi_A(x_B)\right]=S\left[\psi_C(x_B)\right]$ since $\bracket{x_B}{\psi}{x_B}/p_\psi(x_B)$ is a pure state on the bipartition $AC$. Eq.~(\ref{e:Ipartials}) allows us to  characterise the states that have a small post-measurement CMI  by finding the states that have a small average entropy of $\psi_C(x_B)$. \\

Let us now go back to the MPS described in Sec. \ref{s:notation}. With the injective MPS in the canonical form of Eq.~(\ref{e:MPS}), it can be shown that the reduced state of the post-measurement state, $\Psi_C(x_B)$, is (up to zero eigenvalues) isospectral to
\begin{equation}\label{e:isoreducedstate}
\dfrac{1}{p_\Psi(x_B)K^2}\sqrt{\bE^{*|C|}(R)}A_{x_N}\cdots A_{x_1}\bE^{|A|}(L)A_{x_1}^\dag\cdots A_{x_N}^\dag\sqrt{\bE^{*|C|}(R)}.
\end{equation}
The average von Neumann entropy of the reduced state of a post-measurement translationally invariant injective MPS is then
\begin{align}\label{e:entropyreducedstate}
\langle S\left[\Psi_C(x_B)\right]\rangle_{p_\Psi(x_B)}=\sum_{x_B}p_\Psi(x_B)S\Bigg(\dfrac{1}{p_\Psi(x_B)K^2}FA_{x_N}\cdots A_{x_1}\sigma A_{x_1}^\dag\cdots A_{x_N}^\dag F^\dag\Bigg),
\end{align}
where  $\sigma:=\bE^{|A|}(L)$, and 
\begin{equation}
\label{mainFdef}
F :=\sum_{x_{|BC|},\ldots,x_{N+1}}|x_{x_{|BC|},\ldots,x_{N+1}}\rangle\langle R|A_{x_{|BC|},\ldots,x_{N+1}},
\end{equation}
 and thus $F^\dag F =\bE^{*|C|}(R)$. Eq.~(\ref{e:entropyreducedstate}) will be the main object of study throughout this paper. As mentioned earlier, a translationally invariant injective MPS results in a primitive channel $\mathbb{E}$. On a finite-dimensional  space, this implies that for sufficiently large $|A|$ and $|C|$, it follows that both $\sigma$ and $F^{\dagger}F$ are full-rank operators. We also recall that $\lim_{|A|\rightarrow\infty}\mathbb{E}^{|A|}(\chi) = \rho\Tr(\chi)$, $\lim_{|C|\rightarrow\infty}\mathbb{E}^{*|C|}(Q) = \1\Tr(Q\rho)$, and $\lim_{|A|\rightarrow\infty,|C|\rightarrow\infty}K^2 =  \Tr(R\rho)\neq 0$, and consequently (\ref{e:entropyreducedstate}) reduces to
\be\label{e:Sinfinite}
\langle S\left[\Psi_C(x_B)\right]\rangle_{p_\Psi(x_B)}= \sum_{x_B}p_\Psi(x_B)S\left(\frac{A_{x_N}\cdots A_{x_1}\rho A^\dag_{x_1}\cdots A_{x_N}^\dag}{p_\Psi(x_B)}\right),
\ee
for infinite chains.

\section{Distributions with small CMI are quasi-locally Gibbsian}
\label{sec:CMI}

In this section we consider probability distributions, $p_{1,\ldots,|\Lambda|}$, on finite one-dimensional lattices, $\Lambda$, and discuss conditions for when these can be well  approximated by Gibbs distributions of local Hamiltonians. We say that a Hamiltonian is $\ell$-local if it can be written as a sum of terms that each span at most $\ell$ consecutive sites. A distribution is $\ell$-local if it is the Gibbs distribution of some $\ell$-local Hamiltonian. In a similar spirit, we say that a distribution is quasi-locally Gibbsian if it can be approximated by $\ell$-local distributions, where the error of this approximation in some sense decays fast with respect to increasing $\ell$. In this section we show that, if the CMI $I_{p_{1,\ldots,|\Lambda|}}(A:C|B)$ of the distribution $p_{1,\ldots,|\Lambda|}$ decays sufficiently fast with increasing size $|B|$ of the bridging region in a contiguous tripartition $\Lambda = ABC$ of the lattice, then $p_{1,\ldots,|\Lambda|}$ is quasi-locally Gibbsian. (For convenience we change the notation in this section and enumerate the sites of the entire lattice as $1,\dots,|\Lambda|$.) This result is similar in spirit to Kozlov's theorem \cite{Kozlov} (see also Ref.~\onlinecite{Hastings}). Although this section exclusively focuses on probability distributions, the application to quantum states becomes apparent in Section \ref{s:keytheorem}, where we consider classical restrictions of underlying injective MPSs and show that these are quasi-locally Gibbsian under broad conditions. 

Let $p_{1,\ldots,|\Lambda|}$ be a probability distribution over a finite sub-chain $\Lambda$ of a one-dimensional lattice. We let $p_j$ denote the marginal distribution at site  $j$. For $1\leq j \leq k\leq |\Lambda|$ we let $p_{j,\ldots,k}$ denote the marginal distribution of the chain $j,\ldots,k$. 
 In the following, we assume that 
\begin{equation}
p_{1,\ldots,|\Lambda|}(x_1,\ldots,x_{|\Lambda|}) >0,\quad \forall x_1,\ldots,x_{|\Lambda|},
\end{equation}
which consequently leads to $p_{j,\ldots,k}(x_j,\ldots,x_k) >0$. 
With these assumptions, we can define 
\begin{equation}
\label{fbsfgnsfgn1}
h_{j,\ldots,k} := -\log p_{j,\ldots,k},\quad 1\leq j\leq k\leq |\Lambda|,
\end{equation}
and thus $p_{j,\ldots,k} = e^{-h_{j,\ldots,k}}$, where we for notational convenience assume that $h_{j,\ldots,j} := h_{j}$ and $h_{j,\ldots,j+1} := h_{j,j+1}$.

Hence, we have constructed $h_{j,\ldots,k}$ such that $p_{j,\ldots,k}$ is Gibbs distributed with respect to $h_{j,\ldots,k}$, with $\beta = 1$ in $e^{-\beta h_{j,\ldots,k}}/Z(h_{j,\ldots,k})$, where one may note that $Z(h_{j,\ldots,k}) := \sum_{x_j,\ldots,x_k}e^{-h_{j,\ldots,k}(x_j,\ldots,x_k)}  = 1$. 

For $1\leq \ell \leq |\Lambda|-2$, we define
\begin{equation}
\label{hlevelelldef}
h^{\ell}_{1,\ldots,|\Lambda|}   :=     \sum_{j=1}^{|\Lambda|-\ell}h_{j,\ldots,j+\ell} -\sum_{j=1}^{|\Lambda|-\ell -1} h_{j+1,\ldots,j+\ell}.
\end{equation}

Hence, $h^{\ell}_{1,\ldots,|\Lambda|}$ only includes the terms for which the range does not exceed $\ell$.
More precisely, $h^{\ell}_{1,\ldots,|\Lambda|}$ is a $(\ell+1)$-local Hamiltonian.
 The associated $(\ell+1)$-local Gibbs distribution is
\begin{equation}
\label{dvadfv1}
p^{\ell}_{1,\ldots,|\Lambda|}(x_1,\ldots,x_{|\Lambda|}) :=  \frac{e^{- h^{\ell}_{1,\ldots,|\Lambda|}(x_1,\ldots,x_{|\Lambda|}) }  }{
Z(h^{\ell}_{1,\ldots,|\Lambda|})},
\end{equation}
with
$$
Z(h^{\ell}_{1,\ldots,|\Lambda|}) :=  \sum_{x'_1,\ldots,x'_{|\Lambda|}}  e^{- h^{\ell}_{1,\ldots,|\Lambda|}(x'_1,\ldots,x'_{|\Lambda|})}.
$$
The following proposition expresses the classical relative entropy (see Eq.~(\ref{e:relativeentropy})) between the Gibbs distribution $p_{1,\dots,|\Lambda|}$ associated to the full Hamiltonian, $h_{1,\dots,|\Lambda|}$, and the Gibbs distribution $p^\ell_{1,\dots,|\Lambda|}$ associated to the $(\ell+1)$-local Hamiltonian, $h^\ell_{1,\dots,|\Lambda|}$, in terms of the CMIs between suitable regions of the chain. Hence, if the latter are sufficiently small, then the approximating $(\ell+1)$-local distribution $p^\ell_{1,\dots,|\Lambda|}$ is close to the original distribution  $p_{1,\dots,|\Lambda|}$. 
\begin{prop}
\label{PropDecomp}
For $p_{1,\ldots,|\Lambda|}(x_1,\ldots,x_{|\Lambda|}) >0$, let $p^{\ell}_{1,\ldots,|\Lambda|}$ be as defined in Eq.~(\ref{fbsfgnsfgn1}-\ref{dvadfv1}). For $1\leq \ell \leq |\Lambda|-2$ it is the case that

\begin{equation}
\label{fdfbdfb}
S(p_{1,\ldots,|\Lambda|}\Vert p^{\ell}_{1,\ldots,|\Lambda|})  =  \sum_{k=1}^{|\Lambda|-\ell-1}I(1,\ldots,k:k+\ell+1|k+1,\ldots,k+\ell).
\end{equation}

\end{prop}

\begin{proof}
We first note that
\begin{equation}
\label{nfgnfgnfghn}
\begin{split}
\langle h_{j,\ldots,k}\rangle_{p_{1,\ldots,|\Lambda|}} =  H(p_{j,\ldots,k}).
\end{split}
\end{equation}
A somewhat lengthy but straightforward calculation moreover yields
\begin{equation}
\label{dvadfbdfb}
Z(h^{\ell}_{1,\ldots,|\Lambda|}) = 1.
\end{equation}
Next we observe that 
\begin{equation}
\label{fgnsfgngf}
\begin{split}
 & \sum_{k=1}^{|\Lambda|-\ell-1}I(1,\ldots,k:k+\ell+1|k+1,\ldots,k+\ell)\\
 = & H(p_{1,\ldots,\ell+1})-H(p_{1,\ldots, |\Lambda|})  \\
   & + \sum_{k=1}^{|\Lambda|-\ell-1} \bigg[H(p_{k+1,\ldots, k+\ell+1}) -H(p_{k+1,\ldots,k+\ell})  \bigg],\\
 & [\textrm{By (\ref{nfgnfgnfghn})}]\\
 = & \langle h_{1,\ldots,\ell+1}\rangle-\langle h_{1,\ldots, |\Lambda|}\rangle  \\
   & + \sum_{j=1}^{|\Lambda|-\ell-1} \langle h_{j+1,\ldots, j+\ell+1}\rangle -\sum_{j=1}^{|\Lambda|-\ell-1}\langle h_{j+1,\ldots,j+\ell}\rangle,  \\
 = & -\langle h_{1,\ldots, |\Lambda|}\rangle   + \sum_{j=1}^{|\Lambda|-\ell} \langle h_{j,\ldots, j+\ell}\rangle -\sum_{j=1}^{|\Lambda|-\ell-1}\langle h_{j+1,\ldots,j+\ell}\rangle.
\end{split}
\end{equation}
Next we note that
\begin{equation}
\begin{split}
S(p_{1,\ldots,|\Lambda|}\Vert p^{\ell}_{1,\ldots,|\Lambda|})   =  &- \sum_{x_1,\ldots,x_{|\Lambda|}} p_{1,\ldots,|\Lambda|}(x_1,\ldots,x_{|\Lambda|}) h_{1,\ldots,|\Lambda|}(x_1,\ldots,x_{|\Lambda|})\\
   & +\sum_{x_1,\ldots,x_{|\Lambda|}} p_{1,\ldots,|\Lambda|}(x_1,\ldots,x_{|\Lambda|}) h^{\ell}_{1,\ldots,|\Lambda|}(x_1,\ldots,x_{|\Lambda|})\\
     & +\log Z(h^{\ell}_{1,\ldots,|\Lambda|}),\\
         =  &- \langle  h_{1,\ldots,|\Lambda|}\rangle +\langle h^{\ell}_{1,\ldots,|\Lambda|}\rangle +\log Z(h^{\ell}_{1,\ldots,|\Lambda|}),\\
        & [\quad \textrm{By (\ref{dvadfbdfb})} \quad]\\
            =  &- \langle  h_{1,\ldots,|\Lambda|}\rangle +\langle h^{\ell}_{1,\ldots,|\Lambda|}\rangle, \\
             =  &- \langle  h_{1,\ldots,|\Lambda|}\rangle    +\sum_{j=1}^{|\Lambda|-\ell}\langle h_{j,\ldots,j+\ell}\rangle -\sum_{j=1}^{|\Lambda|-\ell -1} \langle h_{j+1,\ldots,j+\ell}\rangle.
\end{split}
\end{equation}
A comparison with with (\ref{fgnsfgngf}) yields (\ref{fdfbdfb}).
\end{proof}

Loosely speaking, the above proposition tells us that, if the CMIs $I(1,\ldots,k:k+\ell+1|k+1,\ldots,k+\ell)$ in some sense decrease sufficiently fast with increasing $\ell$, then the $(\ell+1)$-local Gibbs distribution $p^{\ell}_{1,\ldots,|\Lambda|}$ approaches the true distribution $p_{1,\ldots,|\Lambda|}$. The following lemma formalizes this intuition.

\begin{lemma}
\label{fbsfgbsfg}
Suppose that the probability distribution $p_{1,\ldots,|\Lambda|}(x_1,\ldots,x_{|\Lambda|}) >0$ is such that there exists a function $\xi:\mathbb{N}\rightarrow\mathbb{R}$, such that   
for every contiguous partition $\Lambda = ABC$, it is the case that
\begin{equation}
\label{sfgnsfgnsf}
I_{p}(A:C|B)\leq \xi(|B|),
\end{equation} 
where $\xi$ is independent of $|A|$ and $|C|$.
Let $p^{\ell}_{1,\ldots,|\Lambda|}$ be as defined in (\ref{dvadfv1}), via  (\ref{hlevelelldef}) and  (\ref{fbsfgnsfgn1}). 
Then, for $1\leq \ell \leq |\Lambda|-2$, we have

\begin{equation}
\label{fgfgnf}
S(p_{1,\ldots,|\Lambda|}\Vert p^{\ell}_{1,\ldots,|\Lambda|})  \leq  (|\Lambda|-\ell-1) \xi(\ell) \leq |\Lambda|\xi(\ell).
\end{equation}

\end{lemma}

\begin{proof}
With the general observation that $I(A:C_{1}|B)\leq I(A:C_{1} C_{2}|B)$, we can use  $A =  \{1,\ldots,k\}$, $B =  \{ k+1,\ldots,k+\ell\}$, $C_{1} =  \{k +\ell+1\}$ and $C_{2} =  \{k+\ell+2,\ldots,|\Lambda|\}$ in (\ref{fdfbdfb}) and assumption (\ref{sfgnsfgnsf}), which yields
\begin{equation}
\begin{split}
S(p_{1,\ldots,|\Lambda|}\Vert p^{\ell}_{1,\ldots,|\Lambda|}) 
\leq &    \sum_{k=1}^{|\Lambda|-\ell-1} \xi(\ell) = (|\Lambda|-\ell-1) \xi(\ell).
\end{split}
\end{equation}
\end{proof}

Equation (\ref{fgfgnf}) estimates the contribution of the tails of the  distribution $p_{1,\ldots,|\Lambda|}$: the smaller $|\Lambda|\xi(\ell)$ is, the smaller the contribution of these tails. In suitable joint limits of $\ell$ and $|\Lambda|$, where $|\Lambda|\xi(\ell)$ vanishes exponentially, we say that  $p_{1,\dots,|\Lambda|}$ is a quasi-local Gibbs distribution.
At first sight it might not be clear whether there exists an exponentially decreasing bound $\xi$ with the necessary properties. However, in Section \ref{s:keytheorem}, we establish such a bound, when $p_{1,\ldots,|\Lambda|}$ is the classical restriction of a large class of injective MPS, thus showing that those classical restrictions are quasi-locally Gibbsian.

\section{The Main Theorem}\label{s:keytheorem}

In Sec. \ref{sec:CMI}, we showed that  probability distributions  on finite one-dimensional lattices are quasi-locally Gibbsian if the relevant CMI decays sufficiently rapidly. Here, we apply this result to  classical restrictions of injective MPSs, i.e., to the  the square amplitudes in a given local basis. We express the relevant CMI in terms of the average post-measurement entropy (as discussed in Section \ref{s:CondMutInf}) and find sufficient conditions for when this average entropy decays exponentially. This approach thus yields sufficient conditions for injective MPSs to have classical restrictions that are quasi-locally Gibbsian.

Our result builds extensively on the theory of products of random matrices \cite{RMBook}, and its application to quantum trajectories \cite{Benoist,Maassen}. We particularly follow the approach of Ref.~\onlinecite{Benoist} and formulate the condition for the exponential decay of the average post-measurement entropy in terms of the following condition (referred to as {\bf Pur} in Ref.~\onlinecite{Benoist}) on the matrices $A_x$ associated to the MPS.

\begin{definition}[Purity \cite{Benoist}]
\label{DefPur}
Let $\{A_x\}_{x=0}^{d-1}$ be linear operators on a complex finite-dimensional Hilbert space, $\mathcal{H}$. We say that $\{A_x\}_{x=0}^{d-1}$ satisfies the {\rm purity condition} if the following implication holds:
\begin{equation}
\label{dsfbsfg}
\begin{split}
& \textrm{If $P$ is an orthogonal projector on $\mathcal{H}$ such that}\\
& P A_{x_1}^{\dagger}\cdots A_{x_N}^{\dagger}A_{x_N}\cdots A_{x_1}P \propto P,\quad\forall N\in\mathbb{N},\quad \forall (x_1,\ldots,x_N) \in \{0,\ldots,d-1\}^{\times N},\\
& \textrm{then $\mathrm{rank}(P) = 1$}.
\end{split}
\end{equation}
\end{definition}

Note that the condition $ P A_{x_1}^{\dagger}\cdots A_{x_N}^{\dagger}A_{x_N}\cdots A_{x_1}P \propto P$ is trivially true whenever $P$ is a rank-one projector. Hence, the purity condition means that $ P A_{x_1}^{\dagger}\cdots A_{x_N}^{\dagger}A_{x_N}\cdots A_{x_1}P \propto P$ \emph{only} holds for rank-one projectors. 
The purity condition bears some resemblance to the Knill-Laflamme condition \cite{KLcondition}. We discuss the relationship between the purity condition and error correction/detection in Section \ref{sec:errorcorr}.

An immediate question is if the purity condition is commonly satisfied, or if these cases are rare. One can argue that the error-correction perspective in Section \ref{sec:errorcorr}  suggests that violations of the purity condition are `brittle', and thus provides evidence for the purity condition being `generic' or `typical'.  To shed some further light on this question, we do in Section \ref{PuritySufficient}  present a somewhat simpler condition that implies purity, and where the nature of this simplified condition suggests that the purity condition in some sense is `easily' satisfied. As a concrete application and illustration of this simplified condition, Section \ref{SecModelTypicality} considers a specific probabilistic setting, where all $\{A_x\}_{x=0}^{d-1}$  satisfy the purity condition, apart from a subset of measure zero. This construction thus  formalizes  the notion that purity is a typical or generic property.

Even if  the purity condition is generic, another pertinent question is whether it is easy or not to \emph{check} if a given MPS, in terms of the operators $\{A_x\}_{x=0}^{d-1}$, satisfies the purity condition. Although an interesting question, we leave this as an open problem.

The primary focus of this investigation is the classical CMI  $I_{p_{\Psi}}(A:C|B)$ of the distribution $p_{\Psi} = \langle x_{\Lambda}|\Psi|x_{\Lambda}\rangle$. Theorem \ref{thm:main}, below, shows that purity is a sufficient condition for an exponential decay of  $I_{p_{\Psi}}(A:C|B)$ with increasing $|B|$. However, in order to facilitate a better understanding of the role of the purity condition, Theorem \ref{thm:main} also includes the closely related quantity $I_{\Phi_B(\Psi)}(A:C|B)$, which we recall is the quantum CMI of the post-measurement state $\Phi_B(\Psi)$ as defined in (\ref{e:reducedchannel}). Theorem \ref{thm:main} in essence shows that purity is both necessary and  sufficient condition for the exponential decay of $I_{\Phi_B(\Psi)}(A:C|B)$. Since $I_{\Phi_B(\Psi)}(A:C|B)$ can be viewed as the average entanglement entropy of the post-measurement states $\Psi(x_B)$ (which are pure), this loosely speaking means that the latter typically approach pure product states with respect to the bipartition $A$ and $C$. Whether purity also is a necessary condition for the exponential decay of $I_{p_{\Psi}}(A:C|B)$ is less clear, although one may note that one can find pure states for which $I_{p_{\Psi}}(A:C|B) = 0$, while $I_{\Phi_B(\Psi)}(A:C|B)\neq 0$. (For further details, see the end of Section \ref{app:SecondHalf}.) With this observation in mind, it is conceivable that there may exist a weaker condition than purity that would yield an exponential decay of $I_{p_{\Psi}}(A:C|B)$. Although an interesting question, we leave this as an open problem for future investigations.

\begin{thm}
\label{thm:main}
Let $\Psi$ be an injective MPS on a finite one-dimensional lattice, $\Lambda$, with finite bond dimension, $D$, and open boundary conditions.  If the purity condition holds for the matrices associated to the MPS corresponding to a specific local basis, $\{\ket{x}\}$, then there exist constants $1>\kappa\geq 0$ and  $c\geq 0$, such that for any three contiguous regions $\Lambda=ABC$ as in Fig. \ref{f:lattice}, we have
\begin{equation}
\label{gnffgnnfg}
I_{p_{\Psi}}(A:C|B) \leq I_{\Phi_B(\Psi)}(A:C|B) \leq c \kappa^{|B|}.
\end{equation}
The constants $c$ and $\kappa$ are independent of $|A|$, $|B|$, $|C|$, $|L\rangle$, and $|R\rangle$.

Conversely, suppose that there exist some $|R\rangle$, $|L\rangle$, $|A|$, and $|C|$ such that 
$\sigma:=\mathbb{E}^{|A|}(L)$, and $F^\dag F =\mathbb{E}^{*|C|}(R)$ are full rank operators. Moreover, suppose that there  exist constants $c\geq 0$ and $1 > \kappa \geq 0$, such that 
\begin{equation}
\label{gdhmghm}
  I_{\Phi_B(\Psi)}(A:C|B)\leq  c\kappa^{|B|},
\end{equation}
then   $\{A_x\}_{x=0}^{d-1}$ satisfies the purity condition in Definition \ref{DefPur}.
\end{thm}

The following provides an overview of the essential steps of the proof of Theorem \ref{thm:main}. For a more detailed account of the first half of Theorem \ref{thm:main}, i.e., the purity condition as a sufficient condition for  (\ref{gnffgnnfg}), see the proof of Theorem \ref{hndghngdh} in Appendix \ref{app:FirstHalf}. For the second half, with  (\ref{gdhmghm}) implying the purity condition, see the proof of Theorem \ref{TheoremSecondPart} in Appendix \ref{app:SecondHalf}.

\begin{proof}
For the first part of Theorem \ref{thm:main}, the first step is to bound $I_{\Phi_B(\Psi)}(A:C|B)$ in terms of the quantity $f(N)$, defined below in Eq.~(\ref{fsgbsfgb}). Because of the inequality in (\ref{e:Ipartials}), we consequently also bound the post-measurement CMI $I_{p_{\Psi}}(A:C|B)$. The second step is to show that $f(N)$ decays exponentially if  the purity condition is satisfied; this step is shown independently in Prop.~\ref{PropMain}.  We relegate much of the technical details of the proof to Appendixes \ref{app:A}-\ref{AppProofPropMain} to allow for a clearer presentation of the main ideas. 

To start with, we bound the average entropy (Eq.~(\ref{e:averageS})) in terms of a quantity that can be interpreted as the average purity and we get
\be
\langle S\left[\Psi_C(x_B)\right]\rangle_{p_\Psi(x_B)}\leq -Q\log Q+Q\left[1+\log\left(D-1\right)\right],
\ee
with $Q:=1-\sum_{x_B}p_\Psi(x_B)\Vert\Psi_C(x_B)\Vert$. The proof, which is deferred to Lemma \ref{lem:1} in Appendix \ref{app:A}, follows from concavity of the entropy functional. It is clear that exponential decay of $Q$ implies exponential decay of $I_{\Phi_B(\Psi)}(A:C|B)$ by Eq.~(\ref{e:Ipartials}).

Next, we show that $Q$ can be bounded above by a function of the ordered singular values of the matrix product defining the classical post-measurement MPS. First, Lemma \ref{lem2:AppA} in Appendix \ref{app:A} establishes an upper bound on $Q$ in terms of the average second eigenvalue of the matrix product in Eq.~(\ref{e:entropyreducedstate}) as
\be
Q\leq\dfrac{D-1}{K^2}\sum_{x_B}\lambda^{\downarrow}_2\left(FA_{x_N}\cdots A_{x_1}\sigma A_{x_1}^\dag\cdots A_{x_N}^\dag F^\dag\right),
\ee
where $\{\lambda^{\downarrow}_j(O)\}$ and $\{\nu^{\downarrow}_j(O)\}$ denote the eigenvalues and singular values of an operator $O$ in decreasing order, i.e., $\lambda^{\downarrow}_1(O)\geq\cdots\geq\lambda^{\downarrow}_D(O)$ and $\nu^{\downarrow}_1(O)\geq\cdots\geq\nu^{\downarrow}_D(O)$, respectively. 

Then, recalling that for any operator $O$, we have $\lambda_j(OO^\dag)=\nu_j(O)^2$, we get 
\begin{flalign}
Q&\leq\dfrac{D-1}{K^2}\sum_{x_B}\lambda^{\downarrow}_2\left(FA_{x_N}\cdots A_{x_1}\sigma A_{x_1}^\dag\cdots A_{x_N}^\dag F^\dag\right), \nonumber\\
&\leq\dfrac{D-1}{K^2}\sum_{x_B}\sqrt{\lambda^{\downarrow}_1(FA_{x_N}\cdots A_{x_1}\sigma A_{x_1}^\dag\cdots A_{x_N}^\dag F^\dag)\lambda^{\downarrow}_2(FA_{x_N}\cdots A_{x_1}\sigma A_{x_1}^\dag\cdots A_{x_N}^\dag F^\dag)}, \label{bound:subopt}\\
&=\dfrac{D-1}{K^2}\sum_{x_B} \nu^{\downarrow}_1(FA_{x_N}\cdots A_{x_1}\sqrt{\sigma})\nu^{\downarrow}_2(FA_{x_N}\cdots A_{x_1}\sqrt{\sigma})=:\dfrac{D-1}{K^2}f(N),\nonumber
\end{flalign}
where recall that $|B|=N$.

Next, we need to take into account the fact that $K$ depends on the size of the regions $A$, $B$ and $C$, and in principle $K$ could approach zero. However, the assumption that the MPS is injective, implies that $\mathbb{E}(\cdot) = \sum_x A_x \cdot A_x^{\dagger}$ is primitive, which means that $\mathbb{E}$ has a unique full-rank fixed point. The latter is  used in  Lemma \ref{lemmaadflbdlfk} in Appendix \ref{app:A} to show that for all sufficiently large $|B|$ there exists a number $r>0$ such that
\begin{equation}
    K^2= \langle R|\mathbb{E}^{|\Lambda|}(|L\rangle\langle L|)|R\rangle = \langle R|\mathbb{E}^{|A|+|B|+|C|}(|L\rangle\langle L|)|R\rangle\geq r,
\end{equation}
where $r$ is independent of $|A|$, $|C|$, $|L\rangle$ and $|R\rangle$. We use this to obtain an upper bound on $Q$ that only depends on $N$ via $f(N)$.

Finally, in Proposition \ref{PropMain} below, the function $f(N)$  is shown to decay exponentially if  the purity condition holds. Moreover, the constants $\overline{c}$ and $\gamma$ in the bound (\ref{hgmghmg}) can be chosen to be independent of $|A|$, $|B|$, $|C|$, which follows from the fact that $\overline{c}$ and $\gamma$ are independent of $\sigma$ and $F$.

The proof of the first part of Theorem \ref{thm:main}, requires us to find an upper bound of the average entropy $\langle S\left[\Psi_C(x_B)\right]\rangle_{p_\Psi(x_B)}$ in terms of the quantity $f(N)$. For the second part of Theorem \ref{thm:main},
 i.e., that (\ref{gdhmghm}) implies the purity condition, we instead need to find an upper bound to $f(N)$ in terms of $\langle S\left[\Psi_C(x_B)\right]\rangle_{p_\Psi(x_B)}$. We obtain this via a chain of inequalities 
 \begin{equation}
4\log(2)\lambda^{\downarrow}_2(\rho)\lambda^{\downarrow}_1(\rho) \leq 4\log(2)\lambda^{\downarrow}_1(\rho)\big(1-\lambda^{\downarrow}_1(\rho)\big) \leq  H_B\big(\lambda^{\downarrow}_1(\rho)\big) \leq S(\rho),
 \end{equation}
 where $H_B$ is the binary entropy, i.e., $H_B(\lambda):= -\lambda\log\lambda -(1-\lambda)\log(1-\lambda)$ with $H_B(0):=0$ and $H_B(1):=0$.
 These observations are utilized to show that
\begin{equation}
f(N) \leq \frac{1}{2\sqrt{\log(2)}}\sqrt{\langle S\left[\Psi_C(x_B)\right]\rangle_{p_\Psi(x_B)}},
\end{equation}
with the consequence that an exponential decay of $\langle S\left[\Psi_C(x_B)\right]\rangle_{p_\Psi(x_B)}$ with increasing $N$, implies an exponential decay of $f(N)$. By  Prop.~\ref{PropMain}, the exponential decay of $f(N)$ implies purity of  $\{A_x\}_{x=0}^{d-1}$, if $\sigma:=\mathbb{E}^{|A|}(L)$, and $F^\dag F =\mathbb{E}^{*|C|}(R)$ are full rank operators.

\end{proof}

Note that the bound in Eq.~(\ref{bound:subopt}) is likely quite sub-optimal. It is an interesting open question whether there exists a more direct bound of the average purity that does not rely on bounding the function $f(N)$. The main reason to work with $f(N)$ rather than the average purity is because $f(N)$ is explicitly submultiplicative. 

We now state the key proposition adapted from Ref.~\onlinecite{Benoist}, and references therein.

\begin{prop}[Ref.~\onlinecite{Benoist}]
\label{PropMain}
Let $\{A_x\}_{x=0}^{d-1}$ be operators on a finite-dimensional complex Hilbert space, $\mathcal{H}$, such that $\sum_{x=0}^{d-1}A_x^{\dagger}A_x=\1$. 
For an operator $\sigma$ and an operator $F:\mathcal{H}\rightarrow\mathcal{H}'$ for a finite-dimensional complex Hilbert space $\mathcal{H}'$, define
\begin{equation}
\label{fsgbsfgb}
f(N) :=\sum_{x_1,\ldots,x_N=0}^{d-1}\nu_1^{\downarrow}(FA_{x_N}\cdots A_{x_1}\sqrt{\sigma})\nu_2^{\downarrow}(FA_{x_N}\cdots A_{x_1}\sqrt{\sigma}).
\end{equation}
 If $\{A_x\}_{x=1}^{d-1}$ satisfies the purity condition in Definition \ref{DefPur}, then there exist real constants, $0\leq \overline{c}$ and $0<\gamma <1$, such that for all density operators $\sigma$, and all $F$ such that $F^{\dagger}F\leq \1$, it is the case that 
\begin{equation}
\label{hgmghmg}
f(N)\leq \overline{c}\gamma^N,\quad \forall N\in\mathbb{N}.
\end{equation}
Conversely, if there exists constants $0\leq \overline{c}$ and $0<\gamma <1$ such that (\ref{hgmghmg}) holds for some $\sigma$ and $F^{\dagger}F$ that both are full-rank operators, then $\{A_x\}_{x=1}^{d-1}$ satisfies the purity condition.
\end{prop}

In the application of this proposition, the Hilbert space $\mathcal{H}$ is the virtual space, while $\mathcal{H}'$ is the Hilbert space corresponding to sub-chain $C$, c.f., the definition of $F$ in (\ref{mainFdef}). Theorem \ref{thm:main} provides the necessary bound $\xi(|B|) = c\kappa^{|B|}$ in Lemma \ref{fbsfgbsfg} for showing the quasi-locality of the classical restriction  $p_{1,\ldots,|\Lambda|}(x_1,\ldots,x_{|\Lambda|})= \langle x_{\Lambda}|\Psi| x_{\Lambda}\rangle$. Theorem \ref{thm:main} and Lemma \ref{fbsfgbsfg} thus yield as a corollary  (for a more exact formulation, see Corollary \ref{nzrgnthmtm} in Appendix \ref{app:A})
\begin{equation}
\begin{split}
S(p_{1,\ldots,|\Lambda|}\Vert p^{\ell}_{1,\ldots,|\Lambda|})  \leq & c|\Lambda|\kappa^{\ell},\quad 1\leq \ell \leq |\Lambda|-2.
\end{split}
\end{equation}
A simple example that leads to an exponential decay of the relative entropy is if  $\ell$ is a constant fraction of $|\Lambda|$, i.e.,
\begin{equation}
\ell = \alpha |\Lambda|,\quad 0<\alpha <1.
\end{equation}
The result is that the relative entropy decays exponentially, and the family of $\ell$-local distributions thus approaches $p_{1,\ldots,|\Lambda|}$ exponentially fast. The classical restriction of typical injective MPSs is thus in this sense quasi-locally Gibbsian.

\subsection{Proof overview for Proposition \ref{PropMain}}

Here, we give a brief overview of the general structure and ideas behind the proof of Proposition \ref{PropMain}, i.e., that $f(N)$ decays exponentially if  $\{A_x\}_x$ satisfies the purity condition. Although we do not always follow the exact same tracks, the essence of the proof is due to Refs. \onlinecite{Benoist,Maassen}, which we have adapted to our particular setting and cast in a language that is hopefully more accessible to the quantum information theory community. The proof is essentially self-contained, and is presented in Appendices \ref{SecTechnicalReview}, \ref{TheProcess}
and \ref{AppProofPropMain},  only referencing some standard results from the theory of Martingales (see Appendix \ref{SecTechnicalReview}), that can be found in a number of classic textbook on the subject. 

As one may note from Eq.~(\ref{fsgbsfgb}), the sequence $f(N)$ not only depends on the operators $A_x$, but also on the operators $\sigma$ and $F$. It turns out to be convenient to first focus on  the function
\begin{equation}
w(N) = \sum_{x_1,\ldots,x_N= 0}^{d-1}\nu_1^{\downarrow}(A_{x_N}\cdots A_{x_1})\nu_2^{\downarrow}(A_{x_N}\cdots A_{x_1}).
\end{equation}
Once we have established the purity condition as a necessary and sufficient condition for exponential decay of $w(N)$, we extend (Proposition \ref{bofboneabona} in Section \ref{SecTransition}) this result to $f(N)$, which thus yields the statement of Proposition \ref{PropMain}.

The proof of the exponential convergence of $w(N)$ is essentially done in two steps. First, it is shown that $w(N)$ converges to zero. Next, it is shown that $w(N)$ is submultiplicative, in the sense that $w(N+M)\leq w(N)w(M)$, and thus $\log w(N)$ is subadditive. This observation is used, together with Fekete's subadditive lemma, to show that $w(N)$ goes to zero exponentially fast. These steps are incorporated into the proof of Proposition \ref{wMain}.

The essential approach for proving that $w(N)$ converges to zero is to interpret $w(N)$ as the average over a stochastic process. This process can be viewed as the random measurement outcomes $\boldsymbol{x}_1,\ldots,\boldsymbol{x}_N$ due to a  repeated sequential measurement of the positive operator-valued measure (POVM) $\{A_x^{\dagger}A_x\}_{x=0}^{d-1}$. (This process is described more precisely in Appendix \ref{TheProcess}.) For the proof, it is useful to introduce the operator 
\begin{equation}
\boldsymbol{M}_N = \frac{A_{\boldsymbol{x}_1}^{\dagger}\cdots A_{\boldsymbol{x}_N}^{\dagger}A_{\boldsymbol{x}_N}\cdots A_{\boldsymbol{x}_1}}{\Tr(A^{\dagger}_{\boldsymbol{x}_1}\cdots A_{\boldsymbol{x}_N}^{\dagger}A_{\boldsymbol{x}_N}\cdots A_{\boldsymbol{x}_1})},
\end{equation}
which thus depends on the sequence of  random measurement outcomes  $\boldsymbol{x}_1,\ldots,\boldsymbol{x}_N$. It turns out that one can express $w(N)$ in terms of  $\boldsymbol{M}_N$  via the relation $w(N) =  E\big( \sqrt{\lambda_1^{\downarrow}(\boldsymbol{M}_N) \lambda_2^{\downarrow}(\boldsymbol{M}_N)} \big)D$,
where $\lambda_1^{\downarrow}(\boldsymbol{M}_N)$ and  $\lambda_2^{\downarrow}(\boldsymbol{M}_N)$ denote the largest and the second largest eigenvalue of $\boldsymbol{M}_N$, respectively, and $D$ the dimension of the underlying Hilbert space.  Moreover, $E$ denotes the expectation value over all  possible measurement outcomes.
One can realize that $\boldsymbol{M}_N$ is positive semi-definite, has trace $1$, and can thus be interpreted as a density operator. The main point is that if $\boldsymbol{M}_N$ would be a rank-one operator, and thus correspond to a pure state, then it follows that $\lambda_2^{\downarrow}(\boldsymbol{M}_N)$ is zero. Intuitively, it thus seems reasonable that $w(N)$ converges to zero if it is `sufficiently likely' that $\boldsymbol{M}_N$ converges to a rank-one operator.

The starting point for demonstrating that  $\boldsymbol{M}_N$ converges to a rank-one operator is to show (Lemma \ref{PropMartingale}) that the sequence $(\boldsymbol{M}_N)_{N\in\mathbb{N}}$ is a martingale relative to the sequence of measurement outcomes $(\boldsymbol{x}_N)_{N\in\mathbb{N}}$. This enables us to show (Lemma \ref{fnjjnfddanj}) that $(\boldsymbol{M}_N)_{N\in\mathbb{N}}$ almost surely converges to a positive operator $\boldsymbol{M}_{\infty}$. (All these notions are reviewed in Appendix \ref{SecTechnicalReview}.) Once this is established, the bulk of the proof is focused on showing that $\boldsymbol{M}_{\infty}$ (almost surely) is a rank-one operator if and only if $\{A_x\}_{x=0}^{d-1}$ satisfies the purity condition.  

The arguable least transparent part of the proof is how to show that the purity condition is  sufficient for $\boldsymbol{M}_{\infty}$ to be a rank-one operator. 
 The first part of the proof (Lemma \ref{ghdgguk}) shows that $\boldsymbol{M}_{N+p}$ and $\boldsymbol{M}_N$ in some sense `approach' each other, even when conditioned on $\boldsymbol{x}_1,\ldots,\boldsymbol{x}_N$.  The second part (Lemma \ref{adfbafdba})  
loosely speaking shows that  $\boldsymbol{M}_{N+p}$ gives rise to a term of the form $\sqrt{\boldsymbol{M}_N}\boldsymbol{U}_N^{\dagger}A_{x_{1}}^{\dagger}\cdots A_{x_{p}}^{\dagger}  A_{x_{p}} \cdots A_{x_{1}}\boldsymbol{U}_N\sqrt{\boldsymbol{M}_N}$ for a unitary operator, $\boldsymbol{U}_N$, while $\boldsymbol{M}_N$ gives rise to a term that is proportional to $\boldsymbol{M}_N$. As these operators approach each other when $N$ approaches infinity, one can use this to show that 
\begin{equation}\label{dgndghm}
\boldsymbol{M}_{\infty}\boldsymbol{U}_{\infty}^{\dagger}A_{x'_{1}}^{\dagger}\cdots A_{x'_{p}}^{\dagger}  A_{x'_{p}} \cdots A_{x'_{1}}\boldsymbol{U}_{\infty}\boldsymbol{M}_{\infty}\propto  \boldsymbol{M}_{\infty}\boldsymbol{U}_{\infty}^{\dagger}\boldsymbol{U}_{\infty}\boldsymbol{M}_{\infty}.
\end{equation} 
In a reformulation  (Lemma \ref{netrswethn})  of the purity condition, the projector, $P$, is replaced by a general operator, $O$, again with the conclusion that $O$ must be a rank-one operator. With $O = \boldsymbol{U}_{\infty}\boldsymbol{M}_{\infty}$ it follows that $\boldsymbol{M}_{\infty}$ is a  rank-one operator. 

To conversely show (Lemma \ref{ghdnghnhn}) that the purity condition is a necessary condition is somewhat less involved. By assuming that a projector $P$ satisfies the proportionality in (\ref{dsfbsfg}) while having a rank larger than one, then it follows that the only way in which $\boldsymbol{M}_{\infty}$ can be a rank-one operator, is if $P\boldsymbol{M}_{\infty}P = 0$. However, this leads to a contradiction with $A_x$ being such that $\sum_{k=1}^{L}A_x^{\dagger}A_x = \1$.

\underline{Remark}. Using the same tools as above, Benoist et.~al.~show in Ref.~\onlinecite{Benoist} that the stochastic process defined in Appendix \ref{TheProcess} equilibrates exponentially. It is worth noting that the average purity can converge to zero much faster than the stochastic process. For instance, if $\sigma$ is a rank-one operator, then it trivially follows that  $f(N)$ is identically zero for all $N$, irrespective of whether $\{A_x\}_{x=0}^{d-1}$ satisfies the purity condition or not.

\section{\label{s:Discussion} Discussions  and examples}

In this section, we discuss the purity condition and the decay of the CMI in the context of quantum information theory. In Sec.~\ref{ss:SPTphases} we specifically study the behaviour of the CMI for SPT phases and obtain that it remains constant. The purity condition is discussed from the point of view of quantum error correction in Sec.~\ref{sec:errorcorr}.
In Sec.~\ref{PuritySufficient} we find a simpler condition that implies the purity condition, and based on this simplified condition we discuss the typicality of the purity condition in Sec.~\ref{SecModelTypicality}.
In Sec.~\ref{ss:RateOfDecay}  we show, by constructing two simple examples, that the decay rate of the CMI is unrelated to the decay of the transfer operator of the corresponding MPS. Finally, in Sec.~\ref{ss:Exapmles} we discuss some concrete examples.

\subsection{Symmetry-protected phases}\label{ss:SPTphases}

Here we briefly discuss  systems that do not satisfy the purity condition, and comment on the relation to SPT phases in one dimension.

Consider an MPS, $\ket{\Psi}$, of the form of Eq.~(\ref{e:MPS}) with matrices $A_{x_i}$ that have a tensor product decomposition into two subsystems such that
\be \label{e:decompositionQ}
A_{x_i}=U_{x_i}\otimes T_{x_i},
\ee
where $U_{x_i}$ is a unitary matrix and $T_{x_i}$ is any matrix. Then, the reduced post-measurement state, $\Psi_C( x_B)$, in the infinite chain case (see Eq.~(\ref{e:Sinfinite})) is isospectral to
$$
\Psi_C(x_B)\simeq\frac{1}{p_\Psi(x_B)K^2}\left(U_{x_N}\cdots U_{x_1}\otimes T_{x_N}\cdots T_{x_1}\right)\rho\left(U^\dag_{x_1}\cdots U^\dag_{x_N}\otimes T^\dag_{x_1}\cdots T_{x_N}^\dag\right).
$$

For simplicity, let us further consider the case where the unique fixed point of the transfer operator is proportional to the identity, i.e., $\rho=\1/(D_1D_2)$, where $D_1$ and $D_2$ are the dimensions of the two sub-systems respectively. We obtain
$$
\Psi_C(x_B)\simeq\frac{1}{D_1D_2p_\Psi(x_B)K^2}\left(\1\otimes T_{x_N}\cdots T_{x_1}T^\dag_{x_1}\cdots T_{x_N}^\dag\right).
$$
The von Neumann entropy of this state has two independent contributions coming from each factor of the tensor product, namely
\begin{align*}
S\left[\Psi_C( x_B)\right]&=\log D_1+S\left(\dfrac{1}{p_\Psi( x_B)D_2K^2}T_{x_N}\cdots T_{x_1}T^\dag_{x_1}\cdots T_{x_N}^\dag\right).
\end{align*}
Consequently, the average entropy of entanglement of $\Psi_C( x_B)$, and thus the post-measurement CMI, $I_{\Phi_B(\Psi)}( A: C|B)$ (see Eq.~(\ref{e:Ipartials})), always has a constant contribution independent of the length of the middle region $B$. More generally, the CMI is non-vanishing for MPSs in a basis where the matrices $A_{x_i}$ can be isometrically mapped to a form as in Eq.~(\ref{e:decompositionQ}) \cite{locent2}. 

It was shown in Ref.~\onlinecite{Else} that, for a SPT phase in the MPS framework, there always exists a local basis in which the matrices have the form of Eq.~(\ref{e:decompositionQ}), with the additional property that the unitary matrices form a representation of the symmetry group. The AKLT model (see Sec.~\ref{AKLT}) is such an example.

\subsection{The purity condition: Relation to error correction}\label{sec:errorcorr}

In this section we will explore the purity condition (see Def.~\ref{DefPur}) in more detail. The purity condition states that the only projectors $P$ that satisfy 
\begin{equation}
\label{fbsfgnbfgn}
 P A_{x_1}^{\dagger}\cdots A_{x_N}^{\dagger}A_{x_N}\cdots A_{x_1}P \propto P,
 \end{equation}
for all $ N\in\mathbb{N}$, and all $(x_1,\ldots,x_N) \in \{0,\ldots,d-1\}^{\times N}$, are those that have rank one. 
Here we investigate the relation between this condition (or rather the violation of it) and the Knill-Laflamme error correction condition \cite{KLcondition}.  

Suppose that there exits a projector, $P$, onto a subspace, $\mathcal{C}$, with $\dim\mathcal{C}\geq 2$ such that 
\begin{equation}
\label{AlternativeErrorCorrection}
PA^{\dagger}_{x} A_{x}P = \lambda_{x}P,\quad \forall x.
\end{equation}
This looks suspiciously similar to the Knill-Laflamme error correction condition, which is
\begin{equation}
\label{StandardErrorCorrection}
 P A^{\dagger}_{x} A_{y}P =c_{xy}P,\quad \forall x,y.
\end{equation}
The question is how one can understand the apparent similarity between Eq.~(\ref{AlternativeErrorCorrection}) and (\ref{StandardErrorCorrection}). 
To this end, let us first recall the  error correction scenario.
If $A_x$ are operators on a Hilbert space, $\mathcal{H}$, with $\sum_xA_x^{\dagger}A_x=\1$, we define the corresponding noise channel 
 \begin{equation}
 \label{StandardScenario}
 \bE(\chi) := \sum_x A_x\chi A_x^{\dagger}.
 \end{equation}
For any state, $\chi$, with support on the subspace $\mathcal{C}\subseteq \mathcal{H}$, it is the case that $\bE(\chi)$ can be restored to $\chi$ if and only if (\ref{StandardErrorCorrection}) is true. More precisely, there exists a recovery operation, $\mathcal{R}$, (that does not depend on $\chi$) such that $\mathcal{R}\circ\bE(\chi) = \chi$ for all density operators $\chi$ with support on $\mathcal{C}$.

It turns out that Eq.~(\ref{AlternativeErrorCorrection}) is also a necessary and sufficient condition for error correction, but for a different type of error-model. The channel $\bE$, in the standard error-correction scenario, is the effect of a unitary evolution that acts on $\cH$ and on an environment, $\cH_E$, where the latter is inaccessible to us. In the alternative scenario, we assume that there exists an ancillary system, $A$, which we do have access to, and which we can use in order to help us restore the initial state on $\cH$. More precisely, we assume an error model of the form
\begin{equation}
\label{fsdfbsdb}
\tilde{\bE}(\chi) = \sum_x|x\rangle_A\langle x|\otimes A_x \chi A_x^{\dagger},
\end{equation}
where $\{|x\rangle_A\}_l$ is an orthonormal basis of the Hilbert space associated to the ancillary system, $\cH_A$. We can interpret this as having access to additional classical information about the error in the register, $A$. We use this additional information in order to restore the state on $\mathcal{C}$. One may note that if we have no access to $A$, then we are back to the standard scenario, where the channel on $\cH$ is $\bE = \Tr_A\tilde{\bE}$. It turns out that  (\ref{AlternativeErrorCorrection}) is a necessary and sufficient condition for the existence of a recovery channel $\tilde{\mathcal{R}}:\mathcal{L}(\cH\otimes\mathcal{H}_A)\rightarrow\mathcal{L}(\cH)$, such that $\tilde{\mathcal{R}}\circ\tilde{\bE}(\chi) = \chi$ for all density operators $\chi$ on $\mathcal{C}$. The proof of this statement is nearly identical to that of the original Knill-Laflamme theorem and is omitted here. 

If one finds a non-trivial projector $P$ (i.e. if $\Tr(P) = \dim\mathcal{C}\geq 2$) such that Eq.~(\ref{AlternativeErrorCorrection}) holds, then  one can explicitly construct a collection of unitary operators, $U_x$, such that 
\begin{equation}
\sum_x U_xA_x\chi A_x^{\dagger}U_x^{\dagger} = \chi,
\end{equation}
for all density operators $\chi$ on $\mathcal{C}$. In other words, the operators $U_x$ perform the error correction on  subspace $\mathcal{C}$. More precisely, if we have a set $\{A_x\}$ with $\sum_xA_x^{\dagger}A_x = \1$, for which there exists a non-trivial projector, $P$, that satisfies Eq.~(\ref{AlternativeErrorCorrection}), then we can construct a new `error-corrected' set, $\{\overline{A}_x\}$, with $\overline{A}_x := U_xA_x$ (and $\sum_x\overline{A}_x^{\dagger}\overline{A}_x = \1$). For this new set we will thus not get a decay to zero of the average entropy (Eq.~(\ref{e:entropyreducedstate})), no matter how long a chain $\overline{A}_{x_N}\cdots \overline{A}_{x_1}$  we construct. 

Nothing prevents us from repeating the above reasoning for products $\{A_{x_2}A_{x_1}\}_{x_2,x_1}$, i.e., we can try to find the largest subspace $\mathcal{C}_2$ with corresponding  projector, $P$,  such that 
\begin{equation}
PA_{x_1}^{\dagger}A_{x_2}^{\dagger}A_{x_2}A_{x_1}P = \lambda_{x_2,x_1}P.
\end{equation}
We can similarly ask for the largest subspace $\mathcal{C}_3$  that is correctable for $\{A_{x_3}A_{x_2}A_{x_1}\}_{x_3,x_2,x_1}$. One can realize that we always have $\mathcal{C}_{n}\subseteq \mathcal{C}_{n-1}$.

The purity condition is violated if and only if there exists a non-trivial projector $P$ such that (\ref{fbsfgnbfgn}) holds for all $N$. By the above reasoning we can thus conclude that the purity condition fails if and only if there for all $N$ exists a fixed non-trivial correctable subspace $\mathcal{C}$. Loosely speaking, we can alternatively phrase the purity condition as the non-existence of a non-trivial correctable subspace that persists indefinitely throughout iterated applications of the error channel. Intuitively, this observation suggests that the violation of the purity condition is a rather `brittle' and non-generic phenomenon.

\subsection{\label{PuritySufficient}A sufficient condition for purity}

It is maybe not entirely clear what is the deeper meaning of the purity-condition, or how easy or difficult it is to satisfy. In order to shed some light on the latter question, we here show that if there exists some $N$ for which the set of operators $A^{\dagger}_{x_1}\cdots A^{\dagger}_{x_N}A_{x_N}\cdots A_{x_1}$ span the space of linear operators $\mathcal{L}(\mathcal{H})$ on the underlying (finite-dimensional) Hilbert space $\mathcal{H}$, then $\{A_x\}_{x=0}^{d-1}$ satisfies the purity condition. 
Another way of phrasing this is to say that if for some $N$, the POVM $\{A^{\dagger}_{x_1}\cdots A^{\dagger}_{x_N}A_{x_N}\cdots A_{x_1}\}_{x_1,\ldots,x_N}$ is informationally complete, then $\{A_x\}_{x=0}^{d-1}$ satisfies purity.

In the general case, it seems intuitively reasonable to expect that the set of products $A^{\dagger}_{x_1}\cdots A^{\dagger}_{x_N}A_{x_N}\cdots A_{x_1}$ eventually spans the whole of $\mathcal{L}(\mathcal{H})$, for sufficiently large $N$ (assuming linear combinations with complex coefficients). The exception would be if there exists some particular algebraic relation between the operators $A_x$, which so to speak `trap' the products $A^{\dagger}_{x_1}\cdots A^{\dagger}_{x_N}A_{x_N}\cdots A_{x_1}$ within a nontrivial subspace of $\mathcal{L}(\mathcal{H})$. This argument suggests that the purity condition in some sense would be `easily' satisfied. We investigate this question further in Section \ref{SecModelTypicality}.

Let us first note that a set of operators  $\{A_x\}_{x=0}^{d-1}$ does \emph{not} satisfy the purity condition if there exists a projector $P$ onto an at least two-dimensional subspace of $\mathcal{H}$, and there exist numbers $r_{x_1,\ldots,x_N}$ such that
\begin{equation}
\label{dfnfdbfbd}
PA^{\dagger}_{x_1}\cdots A^{\dagger}_{x_N}A_{x_N}\cdots A_{x_1}P = r_{x_1,\ldots,x_N}P,
\end{equation}
for all $N\in\mathbb{N}$ and all $x_1,\ldots,x_N$.

\begin{prop}
\label{SufficientForPurity}
Let  $\{A_x\}_{x=0}^{d-1}$ be operators on the finite-dimensional complex Hilbert space $\mathcal{H}$, such that $\sum_{x=0}^{d-1}A_x^{\dagger}A_x = \1$. 
If there  exists an $N\in \mathbb{N}$ such that 
\begin{equation}
\mathrm{Sp}\Big(\{ A^{\dagger}_{x_1}\cdots A^{\dagger}_{x_N}A_{x_N}\cdots A_{x_1}  \}_{x_1,\ldots,x_N = 0}^{d-1}\Big) = \mathcal{L}(\mathcal{H}),
\end{equation}
then  $\{A_x\}_{x=0}^{d-1}$ satisfies the purity condition.
\end{prop}
\begin{proof}
It turns out to be convenient to prove that if   $\{A_x\}_{x=0}^{d-1}$ does \emph{not} satisfy the purity condition, then $\{ A^{\dagger}_{x_1}\cdots A^{\dagger}_{x_N}A_{x_N}\cdots A_{x_1}  \}_{x_1,\ldots,x_N = 0}^{d-1}$ does not span $\mathcal{L}(\mathcal{H})$ for any $N$. We thus assume that there exists a projector onto an at least two-dimensional subspace, such that (\ref{dfnfdbfbd}) is satisfied for all $N$, and all $x_1,\ldots,x_N$.
Since $P$ projects onto an at least two-dimensional subspace, there exists an operator $Q$ such that $PQP = Q$, and where $Q\neq cP$ for all $c\in\mathbb{C}$. Assume that $\{A^{\dagger}_{x_1}\cdots A^{\dagger}_{x_N}A_{x_N}\cdots A_{x_1}\}_{x_1,\ldots,x_N}$ would span the whole of $\mathcal{L}(\mathcal{H})$. Then, there exist $c_{x_1,\ldots,x_N}\in\mathbb{C}$ such that
\begin{equation}
\sum_{x_1,\ldots,x_N}c_{x_1,\ldots,x_N}A^{\dagger}_{x_1}\cdots A^{\dagger}_{x_N}A_{x_N}\cdots A_{x_1} = Q.
\end{equation}
This in turn implies that
\begin{equation}
\label{bfbfgn}
\sum_{x_1,\ldots,x_N}c_{x_1,\ldots,x_N}PA^{\dagger}_{x_1}\cdots A^{\dagger}_{x_N}A_{x_N}\cdots A_{x_1}P = PQP = Q.
\end{equation}
However, by (\ref{dfnfdbfbd}) we know that 
\begin{equation}
\sum_{x_1,\ldots,x_N}c_{x_1,\ldots,x_N}PA^{\dagger}_{x_1}\cdots A^{\dagger}_{x_N}A_{x_N}\cdots A_{x_1}P = \sum_{x_1,\ldots,x_N}c_{x_1,\ldots,x_N}r_{x_1,\ldots,x_N}P.
\end{equation}
This combined with (\ref{bfbfgn}) yields
\begin{equation}
 \sum_{x_1,\ldots,x_N}c_{x_1,\ldots,x_N}r_{x_1,\ldots,x_N}P = Q.
\end{equation}
However, this is in contradiction with $Q\neq cP$ for all $c\in\mathbb{C}$. Hence, we can conclude that $\{A^{\dagger}_{x_1}\cdots A^{\dagger}_{x_N}A_{x_N}\cdots A_{x_1}\}_{x_1,\ldots,x_N}$ cannot span the whole of $\mathcal{L}(\mathcal{H})$.

To conclude, if  $\{A_x\}_{x=0}^{d-1}$ does not satisfy the purity condition, then $\{A^{\dagger}_{x_1}\cdots A^{\dagger}_{x_N}A_{x_N}\cdots A_{x_1}\}_{x_1,\ldots,x_N}$ cannot span the whole of $\mathcal{L}(\mathcal{H})$ for any $N$. This yields the statement of the Lemma.
\end{proof}

\subsection{\label{SecModelTypicality}A model for typicality of purity}

In this section,  we consider a concrete model for formalizing the notion of typicality of the purity condition. A common method is to assign a probability measure over the set under consideration, and say that a property is typical, or generic, if it holds for all elements in that set, apart from a subset of measure zero.
This approach thus requires us to construct a probability measure over the objects $\{A_x\}_{x=0}^{d-1}$. Within this construction, we will use Proposition \ref{SufficientForPurity}, and a result from the previous literature (Lemma \ref{ghmgdhmgdh} below), to show that the purity condition is satisfied generically. 
As the reader will note, we here only present a construction for Hilbert spaces of odd dimensions. The reason for why we impose this restriction is to avoid the additional technical complications  that arise in the even-dimensional case (briefly explained  below). It seems likely that these complications are due to the particular proof-technique that we use, rather than some genuine limitations.

In order to construct a probability measure on the sets $\{A_x\}_{x=0}^{d-1}$, we consider, apart from the Hilbert space $\mathcal{H}$, also an ancillary Hilbert space, $\mathcal{H}_A$, of dimension $d$. On $\mathcal{H}_A$, we fix an orthonormal basis, $\{|a_x\rangle\}_{x=0}^{d-1}$, and a normalized element, $|a\rangle\in\mathcal{H}_{A}$. For each unitary operator, $U$, on $\mathcal{H}\otimes\mathcal{H}_A$, we let 
\begin{equation}
\label{sgnsfgnfgm}
A_x :=\langle a_x|U|a\rangle,
\end{equation}
where we note that since $U$ is a mapping on $\mathcal{H}\otimes\mathcal{H}_A$, it follows that $A_x$ is a mapping on $\mathcal{H}$. We can regard this as the result of a procedure where we append an ancillary state, $|a\rangle\langle a|$, to an input state, $\rho$, evolve the system unitarily with $U$, and then perform the projective measurement $\{|a_x\rangle\langle a_x|\}_{x=0}^{d-1}$ on the ancillary system. The conditional (unnormalized) post-measurement state resulting from this procedure is $A_x \rho A_{x}^{\dagger}$.
 If we consider the Haar measure over the set of $Dd\times Dd$ unitary matrices $U$, the construction in (\ref{sgnsfgnfgm}) thus induces a probability measure on the class of sets $\{A_x\}_{x=0}^{d-1}$. In the following, we shall argue that, with respect to the Haar measure, the set of $\{A_x\}_{x=0}^{d-1}$ that satisfy the purity condition is typical, in the sense that the set that violates the purity condition has measure zero. To reach this conclusion, we make use of the following result, which we have taken from  \cite{Nechita_2010}.
Consider some polynomial $P$ with real coefficients, over the real and imaginary parts of the elements of complex $K\times K$ matrices. The following lemma  says that there are only two possibilities: either $P$ is zero on the whole set of unitary $K\times K$ matrices, $\mathbb{U}(K)$, or $P$ is non-zero on almost all of $\mathbb{U}(K)$. 
\begin{lemma}[Lemma 4.3 in \cite{Nechita_2010}]
\label{ghmgdhmgdh}
Given a polynomial $P\in\mathbb{R}[X_1,\ldots,X_{2K^2}]$, the set 
\begin{equation}
 \Big\{[U_{i,j}]_{i,j=1}^{K}\in\mathbb{U}(K): P\big(\Real(U_{i,j}),\Imag(U_{i,j})\big) = 0\Big\},
\end{equation}
is either equal to the whole of $\mathbb{U}(K)$, or it has Haar measure $0$.
\end{lemma}
This means that if we can find a single unitary $U$ for which the polynomial is non-zero, then we know that the polynomial is non-zero for the whole set $\mathbb{U}(K)$, except possibly for a subset of  measure zero.

Regarding the space $\mathcal{L}(\mathcal{H})$ as an inner product space with respect to the Hilbert-Schmidt inner product $\langle B,C\rangle := \Tr(B^{\dagger}C)$, we note that a finite collection of  operators $\mathcal{Q} := \{Q_x\}_{x = 0}^{K-1}$ is linearly independent if and only if the Gram matrix
\begin{equation}
\label{adfdfh}
\boldsymbol{M}(\mathcal{Q}) = [M_{x,x'}]_{x,x' = 0}^{K-1},\quad M_{x,x'} := \langle Q_x, Q_{x'} \rangle = \Tr(Q_{x}^{\dagger}Q_{x'}),
\end{equation}
is positive definite. Since $\boldsymbol{M}(\mathcal{Q})$ in general is positive semi-definite, we thus know that  $\mathcal{Q}$ is linearly independent if and only if all the eigenvalues of $\boldsymbol{M}(\mathcal{Q})$ are non-zero, and consequently, if and only if $\det \boldsymbol{M}(\mathcal{Q}) \neq 0$.

In the following we shall prove that $\{A_x\}_{x=0}^{d-1}$, as constructed via (\ref{sgnsfgnfgm}), satisfies the purity condition for all $U$, except for a subset of Haar measure zero. The general idea of the proof is as follows. For a sufficiently large $N$, we consider a specific subset 
$\mathcal{Q} \subset \{ A^{\dagger}_{x_1}\cdots A^{\dagger}_{x_N}A_{x_N}\cdots A_{x_1}  \}_{x_1,\ldots,x_N = 0}^{d-1}$, 
and we note that $\det\boldsymbol{M}(\mathcal{Q})$ is a polynomial in the matrix elements of $U$. By Lemma \ref{ghmgdhmgdh}, we can thus conclude that either $\det\boldsymbol{M}(\mathcal{Q})=0$ on the whole of $\mathbb{U}(D)$, or $\det\boldsymbol{M}(\mathcal{Q}) \neq 0$ for all $U$ except for a subset of measure zero. If we moreover let $\mathcal{Q}$ contain precisely $D^2$ elements, $|\mathcal{Q}| = D^2$, then this would mean that either $\mathcal{Q}$ does not span $\mathcal{L}(\mathcal{H})$ for any $U$, or $\mathcal{Q}$ spans $\mathcal{L}(\mathcal{H})$ for almost all $U$.
To show the latter, it thus suffices to find one single unitary $U$ such that $\det\boldsymbol{M}(\mathcal{Q}) \neq 0$. The bulk of the proof below is focused on determining such a  unitary, and subset $\mathcal{Q}$, yielding $\det \boldsymbol{M}(\mathcal{Q}) \neq 0$.

An important building block for the construction of the particular unitary operator is the set of generalized Pauli-operators, or shift and clock operators \cite{Sylvester, Weyl, Weyl2, Schwinger}. 
On a Hilbert space $\mathcal{H}$ with dimension $D$, and an orthonormal basis $\{|n\rangle\}_{n=0}^{D-1}$, we define the operators
\begin{equation}
\label{DefShiftClockOperators}
\begin{split}
\Lambda_1 := & \sum_{n=0}^{D-1} |(n+1)\mathrm{mod} D\rangle\langle n|,\quad 
\Lambda_3 :=  \sum_{n=0}^{D-1}\omega^{n}|n\rangle\langle n|,\quad \omega := e^{i2\pi/D},\quad U_{jk} :=   \Lambda_1^{j}\Lambda_3^{k}.
\end{split}
\end{equation}
(The reason for the, at first sight maybe odd-looking, numbering in the subscripts is that $\Lambda_1$ can be regarded as the counterpart to the Pauli-operator $\sigma_1$, and $\Lambda_3$ the counterpart to $\sigma_3$.)
We recall that the set of unitary operators $\{U_{jk}\}_{j,k = 0}^{d-1}$ forms a basis for $\mathcal{L}(\mathcal{H})$, and that  $\langle U_{jk},U_{j'k'}\rangle = \Tr(U_{jk}^{\dagger}U_{j'k'}) = D\delta_{jj'}\delta_{kk'}$, 
and moreover that $\Lambda_3\Lambda_1 = \omega\Lambda_1\Lambda_3$, which in turn leads to 
\begin{equation}
\label{sfgnsfgnfgn}
U_{j'k'}U_{jk} =  \omega^{k'j-j'k}U_{jk} U_{j'k'}.
\end{equation}
We moreover note that 
\begin{equation}
\label{sfgnsfgmhgmsm}
\begin{split}
U_{jk}^{\dagger}=  \omega^{jk}U_{(D-j)\mathrm{mod}D,(D-k)\mathrm{mod}D}.
\end{split}
\end{equation}

Let us here recall that we want a subset $\mathcal{Q}$ made of positive semi-definite operators, $Q_x$, such that $\det M(\mathcal{Q})\neq0$, and thus $\mathcal{Q}$ is linearly independent. For this purpose, we will in the following consider a decomposition (Lemma \ref{Expansion}) of Hermitian operators, $R$, in terms of the operators $\{U_{jk}\}_{j,k=0}^{D-1}$.  Next, we pick an operator $R$ (Lemma \ref{sfgnfnnenene}) with particular properties, which we will use for the construction of  $\mathcal{Q}$. More precisely,  in the following we wish to find a positive semi-definite operator, $R$, that is bounded by the identity, and has non-zero overlaps with all the operators $U_{jk}$, i.e., all the expansion coefficients in the $\{U_{jk}\}_{j,k=0}^{D-1}$ basis should be non-zero. 
By the virtue of being positive semi-definite, the hypothetical operator $R$ has to be Hermitian, i.e., $R^{\dagger} = R$. For this reason, it is useful to take a closer look on the effect, as described by (\ref{sfgnsfgmhgmsm}), of the Hermitian conjugation on the basis elements $U_{jk}$.  
 Suppose now that $D\geq 3$ is an odd number. This means that we can partition the index set $\{0,\ldots, D-1\}$ into the three subsets $\{0\}$, $\{1,\ldots, (D-1)/2\}$, $\{(D+1)/2,\ldots,D-1\}$. 
Under the mapping $n\mapsto (D-n)\, \mathrm{mod}\, D$ the set $\{0\}$ is mapped to itself, while the two sets  $\{1,\ldots, (D-1)/2\}$ and  $\{(D+1)/2,\ldots,D-1\}$ are mapped into each other. In view of  (\ref{sfgnsfgmhgmsm}) one can thus conclude that every Hermitian operator is uniquely determined by the expansion coefficients corresponding to, e.g., the basis elements
\begin{equation}
U_{00},\quad \{U_{n,0}\}_{n=1}^{(D-1)/2},\quad \{U_{0,m}\}_{m=1}^{(D-1)/2},\quad \{U_{n,m}\}_{n = 1,m = 1}^{n=(D-1)/2,m = (D-1)/2},\quad \{U_{n,m}\}_{n = 1,m = (D+1)/2}^{n=(D-1)/2,m = D-1},
\end{equation}
while the expansion coefficients of the remaining basis-elements are fixed by the Hermiticity and the resulting map (\ref{sfgnsfgmhgmsm}).
By the mapping (\ref{sfgnsfgmhgmsm}) it also follows that for a Hermitian operator, the expansion coefficient corresponding to $U_{00}$ has to be real. By the above consideration, we can conclude the following lemma.
\begin{lemma}
\label{Expansion}
For a finite-dimensional complex Hilbert space with odd dimension $D\geq 3$, every Hermitian operator $R$ can be uniquely expanded as
\begin{equation}
\begin{split}
R = & rU_{00}  + \sum_{n=1}^{(D-1)/2}[a_n U_{n,0} + a_n^{*}U_{D-n,0}]   + \sum_{m=1}^{(D-1)/2}[b_m U_{0,m} + b_{m}^{*} U_{0,D-m}]\\
& +\sum_{n=1}^{(D-1)/2} \sum_{m=1}^{(D-1)/2}[A_{n,m}U_{n,m} + A_{n,m}^{*} \omega^{nm}U_{D-n,D-m}]\\
& + \sum_{n=1}^{(D-1)/2} \sum_{m=(D+1)/2}^{D-1}[B_{n,m}U_{n,m} + B_{n,m}^{*} \omega^{nm}U_{D-n,D-m}],\\
\end{split}
\end{equation}
where $r\in\mathbb{R}$ and  $(a_n)_{n=1}^{(D-1)/2},  (b_m)_{m=1}^{(D-1)/2}\in\mathbb{C}^{(D-2)/2}$ and $(A_{n,m})_{n = 1,m = 1}^{n=(D-1)/2,m = (D-1)/2},(B_{n,m})_{n = 1, m=(D+1)/2}^{n=(D+1)/2,m=D-1}\in\mathbb{C}^{(D-2)/2\times (D-1)/2}$. Moreover, by choosing the above coefficients to be non-zero, it follows that all the expansion coefficients of $R$ in the basis $\{U_{jk}\}_{j,k=0}^{D-1}$ are non-zero.
\end{lemma}
Remark: The even-dimensional case is more involved since in this case, not only $0$ is invariant under the map $n\mapsto (D-n)\mathrm{mod} D$, but also $D/2$. For this reason we here focus on the more straightforward odd-dimensional case.

The next step, in Lemma \ref{sfgnfnnenene} below, consists of choosing coefficients $r$, $a_n, b_m$, $A_{n,m}$, and $B_{n,m}$ to obtain an operator $R$ with the desired properties to construct $\mathcal{Q}$. However, in order to prove  this lemma we first need the following observation, which we state without proof.

\begin{lemma}
\label{UnitaryToPositive}
Let $U$ be a unitary operator on a finite-dimensional Hilbert space, and $\theta$ a real number, then
\begin{equation}
0\leq \frac{1}{2}\1 +\frac{1}{4}e^{i\theta}U + \frac{1}{4}e^{-i\theta}U^{\dagger} \leq \1.
\end{equation}
\end{lemma}

\begin{lemma}
\label{sfgnfnnenene}
On every finite-dimensional Hilbert space of odd dimension $D\geq 3$, there exists an operator $R$, such that
 $\1\geq R\geq 0$, and $\Tr(U_{j,k}^{\dagger}R) \neq 0$ for all $j,k\in \{0,\ldots,D-1\}$, 
where $U_{jk}$ are as defined in (\ref{DefShiftClockOperators}).
\end{lemma}
\begin{proof}
By combining Lemma \ref{Expansion}  with Lemma \ref{UnitaryToPositive}, and the observation that 
$U_{D-n,0} = U_{n,0}^{\dagger}$, $U_{0,D-m} = U_{0,m}^{\dagger}$ and $\omega^{nm}U_{D-n,D-m} = U_{n,m}$, 
we find that we can obtain an $R$ with the desired properties, if we choose $(a_n)_{n=1}^{(D-1)/2},  (b_m)_{m=1}^{(D-1)/2}\in\mathbb{C}^{(D-2)/2}$ and $(A_{n,m})_{n = 1,m = 1}^{n=(D-1)/2,m = (D-1)/2},(B_{n,m})_{n = 1, m=(D+1)/2}^{n=(D+1)/2,m=D-1}\in\mathbb{C}^{(D-2)/2\times (D-1)/2}$  and $r\in\mathbb{R}$ in Lemma \ref{Expansion} as non-zero numbers, such that 
\begin{equation}
\begin{split}
r:=2\sum_{n=1}^{(D-1)/2}|a_n|  + 2\sum_{m=1}^{(D-1)/2}|b_m|+2\sum_{n=1}^{(D-1)/2} \sum_{m=1}^{(D-1)/2}|A_{nm}|+ 2\sum_{n=1}^{(D-1)/2} \sum_{m=(D+1)/2}^{D-1}|B_{n,m}|\leq \frac{1}{2}.
\end{split}
\end{equation}
\end{proof}

Having established the existence of an operator $R$ with non-zero overlaps with all the operators $U_{jk}$, which at the same time satisfies that $\1\geq R\geq0$, we are in the position to construct the subset $\mathcal{Q} \subset \{ A^{\dagger}_{x_1}\cdots A^{\dagger}_{x_N}A_{x_N}\cdots A_{x_1}  \}_{x_1,\ldots,x_N = 0}^{d-1}$.

\begin{lemma}
\label{LemmaOneUnitary}
Let $\mathcal{H}$ be a finite-dimensional complex Hilbert space with odd dimension $D\geq 3$. Let $\mathcal{H}_A$ be a finite-dimensional complex Hilbert space with dimension $d\geq 5$. Let $\{|a_x\rangle\}_{x=0}^{d-1}$ be an orthonormal basis of $\mathcal{H}_A$, and $|a\rangle\in \mathcal{H}_A$ normalized.  
Then, there exists a unitary operator $U$ on $\mathcal{H}\otimes\mathcal{H}_A$, and a subset  $\mathcal{Q} \subset \{ A^{\dagger}_{x_1}\cdots A^{\dagger}_{x_{2D-1}}A_{x_{2D-1}}\cdots A_{x_1}  \}_{x_1,\ldots,x_{2D-1} = 0}^{d-1}$,
with $|\mathcal{Q}| = D^2$, such that $\det\boldsymbol{M}(\mathcal{Q}) \neq 0$,
where $\boldsymbol{M}(\mathcal{Q})$ is as defined in (\ref{adfdfh}), and where $A_x :=\langle a_x|U|a\rangle$ for $x = 0,\ldots,d-1$.
\end{lemma}

\begin{proof}
Since $D\geq 3$ is odd, we begin this proof by using Lemma \ref{sfgnfnnenene} in order to construct the subset $\mathcal{Q}$. We know that  there exists an operator $R$ on $\mathcal{H}$ such that $\1\geq R\geq 0$
and $\Tr(U_{j,k}^{\dagger}R) \neq 0$ for all $j,k = 0,\ldots,D-1$,
where the latter implies that
\begin{equation}
\label{knfndfbnfd}
R =\sum_{j,k=0}^{D-1}\xi_{j,k}U_{j,k},\quad \xi_{j,k} = \frac{1}{D}\Tr(U_{j,k}^{\dagger}R) \neq 0.
\end{equation}
Because of  $\1\geq R\geq 0$ both, $\sqrt{R}$ and $\sqrt{1-R}$, are well defined, and since $d\geq 5$, we can define 
\begin{equation}
V := \frac{1}{2}|a_0\rangle\langle a|\otimes\sqrt{R} + \frac{1}{2}|a_1\rangle\langle a|\otimes \Lambda_1 
+ \frac{1}{2}|a_2\rangle\langle a|\otimes \sqrt{1-R} 
+ \frac{1}{2} |a_3\rangle\langle a|\otimes \Lambda_3\\
+ \frac{1}{2}|a_4\rangle\langle a|\otimes\1. 
\end{equation}
We note that $V^{\dagger}V= |a\rangle\langle a|\otimes\1$, and thus $V$ is a partial isometry.

 By virtue of being a partial isometry, $V$ can be extended to a unitary operator $U$ on $\mathcal{H}\otimes\mathcal{H}_A$,  that  satisfies
\begin{equation}
\begin{split}
A_0 := \langle a_0|U|a\rangle = & \langle a_0|V|a\rangle = \sqrt{R},\quad A_1 := \langle a_1|U|a\rangle =  \langle a_1|V|a\rangle = \Lambda_1,\\
A_3 := \langle a_3|U|a\rangle = & \langle a_3|V|a\rangle = \Lambda_3,\quad A_4 := \langle a_4|U|a\rangle =  \langle a_4|V|a\rangle = \1.
\end{split}
\end{equation}
(The above puts no particular restrictions on $A_x := \langle a_x|U|a\rangle$ for   $x\geq 5$.)
In the following, we consider a particular index subset $I\subset \{0,\ldots,d-1\}^{\times(2D-1)}$ of the form
\begin{equation}
(x_1,\ldots,x_{2D-1}) = (\underbrace{4,\ldots,4}_{2D-2-j-k},\underbrace{3,\ldots,3}_k,\underbrace{1,\ldots,1}_{j},0),\quad j,k = 0,\ldots, D-1,
\end{equation}
thus leaving us with operators on the form
\begin{equation}
\begin{split}
 A^{\dagger}_{x_1}\cdots A^{\dagger}_{x_{2D-1}}A_{x_{2D-1}}\cdots A_{x_1} = & {A^{\dagger}_{4}}^{2D-2-k-j}{A_3^{\dagger}}^k{A_1^{\dagger}}^jA_0^{\dagger}A_0A_1^jA_3^kA_{4}^{2D-2-k-j}
  =  U_{jk}^{\dagger}RU_{jk},
\end{split}
\end{equation}
for $j,k= 0,\ldots, D-1$. We let
\begin{equation}
\begin{split}
\mathcal{Q}:= & \{A^{\dagger}_{x_1}\cdots A^{\dagger}_{x_{2D-1}}A_{x_{2D-1}}\cdots A_{x_1} \}_{(x_1,\ldots x_{2D-1})\in I},\\
= &  \{U_{jk}^{\dagger}RU_{jk}\}_{j,k =0}^{D-1}.
\end{split}
\end{equation}
By these constructions, it is clear that $|\mathcal{Q}| = D^2$. Next, we wish to show that $\mathcal{Q}$ is a linearly independent set. We have
\begin{equation}
\begin{split}
Q_{j,k} := &U_{jk}^{\dagger}RU_{jk},  \\
& [\textrm{By (\ref{knfndfbnfd})}]\\
= & \sum_{j',k'=0}^{D-1}\xi_{j',k'}U_{jk}^{\dagger}U_{j'k'}U_{jk}, \\
& [\textrm{By (\ref{sfgnsfgnfgn})}]\\
= & \sum_{j',k'=0}^{D-1}\xi_{j',k'}\omega^{k'j-j'k}U_{j'k'}.
\end{split}
\end{equation}
Hence, if we let $(c_{j,k})_{j,k=0}^{D-1}\in \mathbb{C}^{D\times D}$, then 
\begin{equation}\label{e:liQ}
\begin{split}
 \sum_{j,k=0}^{D-1}c_{j,k}Q_{j,k} = & \sum_{j',k'=0}^{D-1}\xi_{j',k'}\Big(\sum_{j,k=0}^{D-1}c_{j,k}\omega^{k'j-j'k}\Big)U_{j'k'}.
\end{split}
\end{equation}
In order to show that $\mathcal{Q}=\{Q_{j,k}\}_{j,k=0}^{D-1}$ is a linearly independent set, we want to check that Eq.~(\ref{e:liQ}) vanishes only when all $c_{j,k}$ are zero. Since $\{U_{j'k'}\}_{j',k'=0}^{D-1}$ is a basis, we can conclude from Eq.~(\ref{e:liQ}) that  $\sum_{j,k=0}^{D-1}c_{j,k}Q_{j,k}=0$ implies that
$\xi_{j',k'}\sum_{j,k=0}^{D-1}c_{j,k}\omega^{k'j-j'k} = 0$, for all  $j',k' =  0,\ldots, D-1$.
From (\ref{knfndfbnfd}), we know that  $\xi_{j',k'}\neq 0$, and thus it follows that
\begin{equation}
\label{fgnsfgmfhm}
\begin{split}
\boldsymbol{\omega}_{k'}^{\dagger}C\boldsymbol{\omega}_{j'} =   \sum_{j,k=0}^{D-1}c_{j,k}\omega^{k'j-j'k} = 0,\quad j',k' = 0,\ldots, D-1,
\end{split}
\end{equation}
where $C := [c_{j,k}]_{j,k=0}^{D-1}$ and $\boldsymbol{\omega}_{j'} := (1,\omega^{j'},\omega^{2j'},\ldots,\omega^{(D-1)j'})^{t}$.
We note that $\{\boldsymbol{\omega}_{j'}\}_{j'=0}^{D-1}$ forms a basis of $\mathbb{C}^{D}$. Hence, $[\boldsymbol{\omega}_{k'}^{\dagger}C\boldsymbol{\omega}_{j'}]_{k',j'= 0}^{D-1}$ is nothing but the matrix-representation of $C$ in the basis $\{\boldsymbol{\omega}_{j'}\}_{j'=0}^{D-1}$, and thus (\ref{fgnsfgmfhm}) implies that $C =0$, and consequently $c_{j,k} = 0$.
Hence, we can conclude that $\{Q_{j,k}\}_{j,k=0}^{D-1}$ is a linearly independent set. As in (\ref{adfdfh}), we can define $\boldsymbol{M}(\mathcal{Q}) = [M_{jk,j'k'}]_{j,k,j',k' = 0}^{D-1}$ with $M_{jk,j'k'}:= \Tr(Q_{j,k}^{\dagger}Q_{j',k'})$. From the linear independence of $\mathcal{Q}$, it follows that $\det\boldsymbol{M}(\mathcal{Q}) \neq 0$.
\end{proof}

Finally, we assemble the above results in order to prove that purity is a typical property of a set of operators $\{A_x\}_{x=0}^{d-1}$ of a finite-dimensional Hilbert space with odd dimension.

\begin{prop}
Let $\mathcal{H}$ be a finite-dimensional complex Hilbert space with odd dimension $D\geq 3$. Let $\mathcal{H}_A$ be a finite-dimensional complex Hilbert space with dimension $d\geq 5$. Let $\{|a_x\rangle\}_{x=0}^{d-1}$ be an orthonormal basis of $\mathcal{H}_A$, and $|a\rangle\in \mathcal{H}_A$ normalized.  For each unitary operator $U$ on $\mathcal{H}\otimes\mathcal{H}_A$, let $\{A_x\}_{x=0}^{d-1}$ be defined by $A_x :=\langle a_x|U|a\rangle$ for  $x = 0,\ldots,d-1$. 
Then, $\{A_x\}_{x=0}^{d-1}$ satisfies the purity condition in Definition \ref{DefPur} for all unitary operators $U$ on $\mathcal{H}\otimes\mathcal{H}_A$, except for a subset of Haar-measure zero.
\end{prop}
\begin{proof}
For any fixed subset  $\mathcal{Q} \subset \{ A^{\dagger}_{x_1}\cdots A^{\dagger}_{x_{2D-1}}A_{x_{2D-1}}\cdots A_{x_1}  \}_{x_1,\ldots,x_{2D-1} = 0}^{d-1}$, with corresponding matrix $\boldsymbol{M}(\mathcal{Q})$ as defined by (\ref{adfdfh}), it is the case that $\det \boldsymbol{M}(\mathcal{Q})$ is a polynomial in the real and  imaginary parts of the matrix elements in a matrix representation of $U$ (where one may note that since $ \boldsymbol{M}(\mathcal{Q})$ is positive semi-definite, it follows that $\det \boldsymbol{M}(\mathcal{Q})$ is real-valued). Hence, by Lemma \ref{ghmgdhmgdh}, we know that either $\det \boldsymbol{M}(\mathcal{Q}) = 0$ for all $U$, or $\det \boldsymbol{M}(\mathcal{Q}) \neq 0$ for all $U$ except for a subset of Haar measure $0$. By Lemma \ref{LemmaOneUnitary}, there exists a unitary $U$ on $\mathcal{H}\otimes\mathcal{H}_A$, and a subset  $\mathcal{Q} \subset \{ A^{\dagger}_{x_1}\cdots A^{\dagger}_{x_{2D-1}}A_{x_{2D-1}}\cdots A_{x_1}  \}_{x_1,\ldots,x_{2D-1} = 0}^{d-1}$,
with $|\mathcal{Q}| = D^2$, such that $\det\boldsymbol{M}(\mathcal{Q}) \neq 0$. Hence, we can conclude that $\det \boldsymbol{M}(\mathcal{Q}) \neq 0$ for all $U$, except for a subset of measure zero. If $\det\boldsymbol{M}(\mathcal{Q}) \neq 0$, then this  implies that the  elements of $\mathcal{Q}$ are linearly independent. Since $|\mathcal{Q}|=D^2 = \dim\mathcal{L}(\mathcal{H})$, it moreover follows that $\mathrm{Sp}(\mathcal{Q}) = \mathcal{L}(\mathcal{H})$. Consequently, $\mathrm{Sp}\Big(\{ A^{\dagger}_{x_1}\cdots A^{\dagger}_{x_{2D-1}}A_{x_{2D-1}}\cdots A_{x_1}  \}_{x_1,\ldots,x_{2D-1} = 0}^{d-1}\Big) = \mathcal{L}(\mathcal{H})$. By Proposition  \ref{SufficientForPurity} it follows that $\{A_x\}_{x=0}^{d-1}$ satisfies the purity condition. We can conclude that for all unitary operators on $\mathcal{H}\otimes\mathcal{H}_A$, except for a subset of Haar measure zero, it is the case that $\{A_x\}_{x=0}^{d-1}$ satisfies the purity condition.

\end{proof}

\subsection{\label{ss:RateOfDecay} Rate of decay}

As we have seen in Section \ref{s:keytheorem}, the CMI of an MPS decays exponentially to zero when the matrices of the MPS satisfy the purity condition (see Def.~\ref{DefPur}). Moreover, the transfer operator of an injective MPS, $\bE$, decays to its unique fixed point exponentially fast at a rate lower bounded by the gap of the channel. The decay rate is often referred to as the correlation length. One could expect that there exists a relation between the correlation length and the decay rate of the classical  CMI of Theorem \ref{thm:main}. In this section, we consider two examples which give clear evidence of the absence of such relation. \\

Consider an MPS of the form of Eq.~(\ref{e:MPS}) such that all matrices $A_{x_i}$ have rank one. After a  measurement of system $B$, the reduced state $\Psi_C(x_B)$ becomes  pure for any measurement outcome (see Eq.~(\ref{e:isoreducedstate})). Therefore, the average von Neumann entropy of $\Psi_C(x_B)$, and thus the post-measurement CMI according to Eq.~(\ref{e:Ipartials}), is zero instantaneously even if region $B$ is a single site. On the other hand, the correlation length need not be zero. For example, given a collection of transition probabilities, $\{P(x_j|x_i)\}_{x_i,x_j=0}^{d-1}$, and an orthonormal basis, $\{\ket{x_i}\}_{x_i=0}^{d-1}$, the repeated application of the transfer operator of an MPS with $A_{x_{ij}}=\sqrt{P(x_j|x_i)}\ket{x_j}\bra{x_i}$ effectively implements a classical Markov process with transition probability $P(x_j|x_i)$ such that
$$
\bE^{\circ N}(\chi)=\sum_{x_1,\dots,x_N=0}^{d-1}P(x_2|x_1)\cdots P(x_N|x_{N-1})\bra{x_1}\chi\ket{x_1}\ket{x_N}\bra{x_N}.
$$
Nothing prevents this Markov chain to have a slow convergence to its equilibrium distribution.

Conversely,  consider an MPS with matrices $A_{x_i}$  proportional to unitary operators. As we have discussed in Section \ref{ss:SPTphases}, this implies that there is no decay of the von Neumann entropy, and thus the CMI remains constant. However, an injective MPS always has a finite correlation length. As a concrete example, consider
$$
A_{x_{ij}}:=\dfrac{1}{D}\sum_{k=0}^{D-1}e^{2\pi i\frac{kx_j}{D}}\ket{k}\bra{(k+x_i)\text{ mod }D},
$$
where $\{\ket{k}\}_{k=0}^{D-1}$ is an orthonormal basis. One can easily check that $A_{x_{ij}}$ are proportional to unitary operators, and hence the average von Neumann entropy, $\langle S\left[\Psi_C(x_B)\right]\rangle$, does not decay. The transfer operator of this MPS is the replacement map that replaces any input state, $\chi$, with the maximally mixed state, i.e.,
$$
\bE(\chi)=\sum_{x_i,x_j=0}^{D-1}A_{x_{ij}}\chi A_{x_{ij}}^\dag=\dfrac{\Tr\chi}{D}\1.
$$

In Refs.~\onlinecite{locent0, locent1}, the decay of classical and quantum correlations is also studied. There, the authors introduce an entanglement measure called Localizable Entanglement (LE). The LE is defined as the maximal amount of entanglement that can be created on average between two spins at positions $i$ and $j$ of a chain by performing local measurements on the other spins. It is easy to note that the LE is similar to the scenario that we are considering in this paper (see Eq.~(\ref{e:averageS})). Indeed, the difference is simply that the LE optimises over the basis of the measurement, while we pick a concrete basis. For the case when the measured spins are spin-1/2, it is shown in Refs.~\onlinecite{locent0, locent1} that the connected correlation function provides a lower bound on the LE.

\subsection{\label{ss:Exapmles} Examples}

In this section we consider examples that illustrate some features of the process under study. As a prototypical example, we look at  the AKLT model and obtain the exact convergence rate in a specific basis. Then, we consider MPSs with strictly contractive transfer operator and pure fixed point. This second example shows that primitivity of the transfer operator is not a necessary condition for the exponential convergence of the post-measurement CMI. In the last example that we construct, the purity condition is violated up to a fixed length $|B|=N$, but satisfied thereafter.

\subsubsection{AKLT state}\label{AKLT}
The first state we want to consider is the 1D AKLT model. The AKLT state defined on a chain has a well-known MPS description with bond dimension $D=2$ and physical dimension $d=3$. The matrices in the MPS picture are given by 
\be A_0=-\frac{1}{\sqrt{6}}\left(\begin{matrix} 
1 & 0 \\
0 & -1 
\end{matrix}\right), ~~~~ 
A_+=\sqrt{\frac{1}{3}}\left(\begin{matrix} 
0 & 0 \\
1& 0 
\end{matrix}\right), ~~~~ 
A_-=-\sqrt{\frac{1}{3}}\left(\begin{matrix} 
0 & 1 \\
0& 0 
\end{matrix}\right).\ee

We take the $\{x_i\}=\{0,+,-\}$ as the basis for our physical space. It can be seen by inspection that the transfer operator,  $\bE(\chi)=\sum_{x_i} A_{x_i}\chi A_{x_i}^\dag$, has a unique stationary state $\rho=\1/2$. In the infinite chain setting, the probability of a measurement outcome $x_B=x_1,\dots, x_N$ is given by 
\be
p_\Psi(x_B)=\dfrac{1}{2}\Tr\left[A_{x_N}\cdots A_{x_1}A_{x_1}^\dag\cdots A_{x_N}^\dag\right].
\ee
We want to calculate the average entropy on $C$ after measurement of system $B$, i.e., we need to estimate
\be\label{e:SofAKLT}
\avr{S\left[\Psi_C\left(x_B\right)\right]}=\sum_{x_B=0,+,-}p_\Psi\left(x_B\right)S\left[\dfrac{A_{x_N}\cdots A_{x_1}A_{x_1}^\dag\cdots A_{x_N}^\dag}{\Tr\left(A_{x_N}\cdots A_{x_1}A_{x_1}^\dag\cdots A_{x_N}^\dag\right)} \right].
\ee
We note two scenarios: $p_\Psi(x_B)=0$ and $p_\Psi(x_B)\neq0$. Since $A_+A_+=A_-A_-=0$, we get that whenever the string $x_B$ contains two (or more) successive $+$ (or $-$), then $p_\Psi(x_B)=0$. In other words, the only strings that give a $p_\Psi(x_B)\neq0$ are those with an alternating sequence (Ex: $+-+-$), possibly interspersed with $0$'s. However, the only string with non-zero entropy is the one with all $0$'s because any alternating sequence has rank one. This string will occur with probability $1/3^{N}$. We get that 
\be\label{e:SofAKLTdecaying}
\avr{S\left[\Psi_C(x_B)\right]} = \frac{1}{3^{N}}S\left[\Psi_C(x_0)\right],
\ee
where $x_0:=0,\dots,0$. Hence, the AKLT model in the standard basis has a post-measurement CMI that is exponentially decaying in the size of $B$ for large $A$ and $C$. The correlation length is coincidentally the same as the classical CMI decay in this basis. 

Let us now consider a change of basis. The 1D AKLT state is also given by the MPS representation with matrices
$$
\tilde{A}_0:=\sqrt{\dfrac{1}{3}}\hat{\sigma}_x=\sqrt{\dfrac{1}{3}}\left(
\begin{matrix}
0 & 1 \\
1 & 0
\end{matrix}\right),~~~~
\tilde{A}_1:=\sqrt{\dfrac{1}{3}}\hat{\sigma}_y=\sqrt{\dfrac{1}{3}}\left(
\begin{matrix}
0 & -i \\
i & 0
\end{matrix}\right),~~~~
\tilde{A}_2:=\sqrt{\dfrac{1}{3}}\hat{\sigma}_z=\sqrt{\dfrac{1}{3}}\left(
\begin{matrix}
1 & 0 \\
0 & -1
\end{matrix}\right),~~~~
$$
where $\hat{\sigma}_i$ are the Pauli matrices. The Pauli matrices are unitary, and thus the average von-Neumann entropy of $\Psi_C(x_B)$ (Eq.~(\ref{e:SofAKLT})) is constant, namely $\log 2$. In other words, the post-measurement CMI of the AKLT chain when measured in the basis corresponding to the Pauli matrices is not decaying. This is consistent with the discussion in Section \ref{ss:SPTphases} because the AKLT chain is in the Haldane phase, which is a SPT phase protected by the $Z_2\times Z_2$ symmetry generated by the $\pi$ rotations around three orthogonal axes.

\subsubsection{\label{s:PureFiXPoint} Strictly contractive map with a pure fixed point}

As mentioned in Section \ref{s:notation}, translationally invariant MPSs are injective if and only if the transfer operator $\bE$, defined in (\ref{e:TransferOperator}),  is primitive~\cite{primitivity}, where primitivity means that the channel possesses a unique full-rank fixed point. Primitivity in turn guarantees the existence of a gapped parent Hamiltonian and exponential decay of correlations~\cite{injectivity,FCS}.

In view of the essential role played by primitivity for the decay of correlations, one may ask how it relates to the purity condition for the exponential decay of the CMI, in the sense of Theorem \ref{thm:main}.  In this section, we present an example which shows that primitivity is \emph{not} a necessary condition for the exponential decay of the CMI. In other words, we consider an MPS with a non-primitive transfer operator, which nevertheless yields and exponentially decaying CMI due to purity.

Consider an MPS of the form of Eq.~(\ref{e:MPS}) which has a strictly contractive transfer operator, $\bE$, with a pure fixed point, denoted by $\ket{\phi}$. Let us recall that a channel, $\Phi$, is strictly contractive if there exists a number $0\leq \alpha <1$ such that $\Vert \Phi(\chi_1)-\Phi(\chi_2)\Vert_1\leq\alpha\Vert\chi_1-\chi_2\Vert_1$ for all density operators $\chi_1$ and $\chi_2$. Note that if a channel is strictly contractive, then the fixed point is unique.  Note further that, since the fixed point of $\mathbb{E}$ is pure, the transfer operator is, by definition, not primitive. Let us also assume that $F=\sqrt{\bE^{*|C|}(R)}$ is full rank. Under these assumptions, our aim is to find an exponentially decaying bound of the average von Neumann entropy of the reduced  post-measurement state (see Eq.~(\ref{e:entropyreducedstate})). We start by using the concavity of the entropy and obtain
\begin{align}
\nonumber\langle S\left[\Psi_C(x_B)\right]\rangle_{p_\Psi(x_B)}&\nonumber\leq S\left(\dfrac{1}{K^2}\sum_{x_B}FA_{x_N}\cdots A_{x_1}\sigma A_{x_1}^\dag\cdots A_{x_N}^\dag F^\dag\right)=S\left(\dfrac{F\bE^{N}\left(\sigma\right)F^\dag}{\Tr\left[F\bE^{N}\left(\sigma\right)F^\dag\right]}\right),
\end{align}
where $K^2=\Tr\left[F\bE^{N}\left(\sigma\right)F^\dag\right]$. The purity of the fixed point of $\bE$ allows us to transform the above inequality to
\begin{align*}
\langle S\left[\Psi_C(x_B)\right]\rangle_{p_\Psi(x_B)}&\leq S\left(\dfrac{F\bE^{N}\left(\sigma\right)F^\dag}{\Tr\left[F\bE^{N}\left(\sigma\right)F^\dag\right]}\right)=\left|S\left(\dfrac{F\bE^{N}\left(\sigma\right)F^\dag}{\Tr\left[F\bE^{N}\left(\sigma\right)F^\dag\right]}\right)-S\left(\dfrac{F\bE^{N}\left(\ket{\phi}\bra{\phi}\right)F^\dag}{\Tr\left[F\bE^{N}\left(\ket{\phi}\bra{\phi}\right)F^\dag\right]}\right)\right|,
\end{align*}

We can further bound the average von Neumann entropy using the Fannes-Audenaert inequality \cite{Fannes,Fannes-Au}. This yields 
\begin{align}\label{e:contractive_boundafterFannes}
\langle S\left[\Psi_C(x_B)\right]\rangle_{p_\Psi(x_B)}\leq t\log(D-1)+H_B(t),
\end{align}
where $H_B$ is the binary entropy, i.e., $H_B(t)=-t\log(t)-(1-t)\log(1-t)$, with $H_B(0):=0$ and $H_B(1):=0$, and where $t$ is defined as
\begin{equation}
\label{jdffjnd}
t:=\dfrac{1}{2}\left\Vert\dfrac{F\bE^{N}\left(\sigma\right)F^\dag}{\Tr\left[F\bE^{N}\left(\sigma\right)F^\dag\right]}-\dfrac{F\bE^{N}\left(\ket{\phi}\bra{\phi}\right)F^\dag}{\Tr\left[F\bE^{N}\left(\ket{\phi}\bra{\phi}\right)F^\dag\right]}\right\Vert_1,
\end{equation}
with $\Vert\cdot\Vert_1$ denoting the trace norm. 

For any full-rank $F$, and any pair of density operators $\chi_1$, $\chi_2$ on a finite-dimensional Hilbert space, one can show that
\begin{equation}
\label{nzmrzmr}
\begin{split}
\left\Vert \frac{F\chi_2 F^{\dagger}}{\Tr(F\chi_2 F^{\dagger})} - \frac{F\chi_1 F^{\dagger}}{\Tr(F\chi_1 F^{\dagger})}\right\Vert_1\leq   2\left(\frac{\nu_1(F)}{\nu_D(F)}\right)^4\Vert\chi_1-\chi_2\Vert_1,
\end{split}
\end{equation}
where $\nu_1(F)$ and $\nu_D(F)$ denote the largest and the smallest singular values of $F$, and where we note that $\nu_D(F)>0$ since $F$ is full-rank on a finite-dimensional space.

By combining (\ref{jdffjnd}) and (\ref{nzmrzmr}) with an iterative use of strict contractivity of $\bE$, we find an upper-bound on $t$ such that
\begin{align}\label{e:boundont}
\nonumber t&\leq\left(\dfrac{\nu_1(F)}{\nu_D(F)}\right)^4\Vert\bE^N\left(\sigma\right)-\bE^N\left(\ket{\phi}\bra{\phi}\right)\Vert_1, \\
&\leq\left(\dfrac{\nu_1(F)}{\nu_D(F)}\right)^4\alpha^N\Vert\sigma-\ket{\phi}\bra{\phi}\Vert_1.
\end{align}
If it would be possible to choose $\alpha = 0$, then $\Vert\bE(\chi_1)-\bE(\chi_2)\Vert_1 = 0$, and thus $\bE(\chi_1)=\bE(\chi_2)$, which implies that the convergence is not only exponential, but immediate. Hence, without loss of generality, we may in the following assume that $0<\alpha<1$.

The next step consists of bounding the binary entropy, $H_B(t)$. For that, we define the function $g(t):=t-t\log t$, with $g(0):=0$. One can show that $g$ is monotonically increasing on $t\in[0,1]$ and satisfies $H_B(t)\leq g(t)$ for $0\leq t\leq 1$. These two properties of $g$ together with inequality (\ref{e:boundont}) lead to
\begin{equation}\label{e:boundonHB}
\begin{split}
H_B(t) \leq &~\left(\dfrac{\nu_1(F)}{\nu_D(F)}\right)^4 \alpha^N\left\Vert\sigma-\ket{\phi}\bra{\phi}\right\Vert_1 \\
&\hspace{0.35cm} - \left(\dfrac{\nu_1(F)}{\nu_D(F)}\right)^4\alpha^N \left\Vert\sigma  - |\phi\rangle\langle\phi| \right\Vert_1\log \left[ \left(\dfrac{\nu_1(F)}{\nu_D(F)}\right)^4\left\Vert\sigma  - |\phi\rangle\langle\phi| \right\Vert_1 \right]\\
&\hspace{0.35cm}  - \left(\dfrac{\nu_1(F)}{\nu_D(F)}\right)^4\left\Vert\sigma  - |\phi\rangle\langle\phi| \right\Vert_1  N\alpha^N   \log  \alpha.
\end{split}
\end{equation}
By combining (\ref{e:contractive_boundafterFannes}) with (\ref{e:boundonHB}), and again using inequality (\ref{e:boundont}), we find that 
\be\label{e:contractive_boundonS}
\langle S\left[\Psi_C(x_B)\right]\rangle_{p_\Psi(x_B)}\leq c_1\alpha^N+c_2N\alpha^N,
\ee
where $c_1$ and $c_2$ are defined as
$$
c_1:=\left(\dfrac{\nu_1(F)}{\nu_D(F)}\right)^4 \left\Vert\sigma  - |\phi\rangle\langle\phi| \right\Vert_1\bigg[ \log(D-1) + 1  - \log \left( \left(\dfrac{\nu_1(F)}{\nu_D(F)}\right)^4\left\Vert\sigma  - |\phi\rangle\langle\phi| \right\Vert_1 \right)\bigg],
$$
$$
c_2:=\left(\dfrac{\nu_1(F)}{\nu_D(F)}\right)^4\left\Vert\sigma  - |\phi\rangle\langle\phi| \right\Vert_1 \left(- \log  \alpha\right).
$$
The right-hand side of Eq.~(\ref{e:contractive_boundonS}) decays exponentially to zero when $N$ grows to infinity because $0<\alpha < 1$. Hence, we have found an exponentially-decaying bound on the post-measurement CMI for an MPS with a transfer operator that is strictly contractive and has a pure fixed point. This shows that primitivity of the transfer operator is not a necessary condition for the CMI to converge.

\subsubsection{Jordan blocks}
Here, we construct a simple example of a two-element set $\{A_0,A_1\}$ that has a nontrivial correctable subspace (in the sense of Section \ref{sec:errorcorr}) for small enough $N$, but where the purity condition nevertheless holds. 
For the Hilbert space $\cH$, we let $\dim\mathcal{H} = D+1$ and let $|0\rangle,\ldots,|D-1\rangle,|D\rangle$ be an orthonormal basis of $\cH$. We define the projector $P:= \sum_{k=0}^{D-1}|k\rangle\langle k|$ and the operators
\begin{equation}
A_{0} :=  \sum_{k=0}^{D-1}|k+1\rangle\langle k|, ~~~~~~{\rm and}~~~~~~A_{1} :=  |D\rangle\langle D|.
\end{equation}
We note that $A_0$ is a Jordan block with zeros on the diagonal, and that
\begin{equation}
A_{0}^{\dagger}A_0 =  \sum_{k=0}^{D-1}|k\rangle\langle k| = P,
\end{equation}
while $A_{1}^{\dagger}A_1 = A_1 = |D\rangle\langle D|$.
Thus, $A_{0}^{\dagger}A_0 + A_{1}^{\dagger}A_1 = \1$.
Moreover, observe that 
\begin{equation}
\begin{split}
PA_{0}^{\dagger}A_0P = & \lambda_0P,\quad \lambda_0 = 1,\\
PA_{1}^{\dagger}A_1P = & \lambda_1P,\quad \lambda_1 = 0.
\end{split}
\end{equation}
Hence, for a single site, the correctable subspace is $\mathcal{C}_1 := {\rm span}\{|0\rangle,\ldots,|D-1\rangle\}$.

Now consider the case of several sites, where we construct the sequence $A_{x_N}\cdots A_{x_1}$. It is not difficult to see that 
\be 
A_0^{N} =  \sum_{k=0}^{D-N}|k+N\rangle\langle k|.
\ee
Hence, the correctable subspace decreases the  dimension with one step along the sequence, until it is exhausted. As a consequence, the CMI in this example will be exponentially decaying with a pre-factor that grows exponentially in $D$. Examples that do not have a block diagonal structure can also be constructed. We note that the example above is similar in spirit to a bosonic annihilation operator. Indeed, in the infinite system case where the operators $A_i$ are bosonic creation and annihilation operators, the purity condition no longer makes sense, and the theory breaks down.

\section{\label{s:Outlook} Outlook}

We have shown that the amplitudes of an injective MPS in a specific local basis follow a quasi-local Gibbs distribution with exponentially decaying tails if the matrices associated to the MPS satisfy a particular `purity-condition'. The purity condition reflects the fact that no information can be preserved in the virtual subspace on average, upon measurements. Our proof makes extensive use of the theory of random matrix products. 

A number of open questions remains. Perhaps the most obvious is whether the methods used in this paper can be applied in higher dimensions or in the context of matrix product operators, and whether this leads to new insights or algorithmic improvements. In the setting of matrix product operators, the purity condition would no longer be sufficient to prevent information transmission along the chain. There one would likely have to bound the stochastic process upon measurements from above and below. Some recent progress in this direction has been communicated to us \cite{BK}.

Another place where the present tools might be applied is in the rigorous analysis of the Wave Function Monte Carlo algorithm. A first attempt to achieve this has been made in Ref.~\onlinecite{Benoist2}, yet some work remains to be done in connecting these mathematical results to more realistic physical settings and particular examples. 
Yet another extension would be to continuous MPSs \cite{cMPS}.  

On a more technical level, it would be valuable to get a better handle on the decay rate of the stochastic process. In particular, whether there exists a closed form expression as is the case for the correlation length (as the spectral gap of the transfer operator).   

There are also open questions related to the purity condition, such as whether it can be easily checked or not for a given MPS. In this investigation we have shown that the purity condition in essence is necessary and sufficient for the exponential decay of the quantum conditional mutual information, $I_{\Phi_B(\Psi)}(A:C|B)$, while we only have proved that purity is a sufficient condition for the exponential decay of the classical conditional mutual information, $I_{p_{\Psi}}(A:C|B)$. An open question is thus whether purity in essence also is a necessary condition for the latter, or whether there exists another weaker condition that would be both  necessary and sufficient.

\begin{acknowledgments}

We thank T.~Benoist for clarifying some details in Ref.~\onlinecite{Benoist}. We thank David Gross for helpful discussions. Funded by the Deutsche Forschungsgemeinschaft (DFG, German Research Foundation) under Germany's Excellence Strategy - Cluster of Excellence Matter and Light for Quantum Computing (ML4Q) EXC 2004/1-390534769. This work was completed while MJK was at the University of Cologne. We also thank an anonymous referee for pointing out improvements of Proposition \ref{PropDecomp} and Lemma \ref{fbsfgbsfg}, which  strengthened the results and simplified the proofs. 

\end{acknowledgments}

\section*{Data availability}
Data sharing is not applicable to this article as no new data were created or analyzed in this
study.

\appendix

\section{Elements of the proof of Theorem \ref{thm:main}}\label{app:A}

Here we present a detailed account of the proof of Theorem \ref{thm:main}. For convenience of presentation, we have divided the statement of Theorem \ref{thm:main} into two parts. The first part, that purity implies (\ref{gnffgnnfg}), is proved in \ref{app:FirstHalf}, and stated in a more detailed version in Theorem  \ref{hndghngdh}. The second part of Theorem  \ref{thm:main}, that (\ref{gdhmghm}) implies purity, is the focus of section \ref{app:SecondHalf}, and is formalized in Theorem \ref{TheoremSecondPart}.

\subsection{\label{app:FirstHalf} Proof of the first part of Theorem \ref{thm:main}}

In order to show Theorem \ref{thm:main} in Section \ref{s:keytheorem}, we use two key bounds, one on the average von Neumann entropy and a second on the average purity. Here, we state these bounds in the form of two lemmas. 
In the following we let $\Vert Q\Vert :=\sup_{\Vert \psi\Vert =1}\Vert Q|\psi\rangle\Vert$ denote the standard operator norm.
\begin{lemma}\label{lem:1}
Let $\{\rho_{x}\}_{x=0}^{M-1}$ be a collection of density operators on a Hilbert space $\cH$, with $D = \dim\mathcal{H}$, and $\{p(x)\}_{x=0}^{M-1}$ be real numbers such that $p(x)\geq0$ and $\sum_{x=0}^{M-1}p(x)=1$. Then,
\begin{equation}\label{e:lemma1}
\sum_{x=0}^{M-1}p(x) S(\rho_x )\leq -Q\log Q+Q\left[\log\left(D-1\right)+1\right],
\end{equation}
where we refer to $Q$ as the the average purity and define it as
\begin{equation}\label{e:Q}
Q:=1-\sum_{x=0}^{M-1}p(x) \Vert\rho_x\Vert.
\end{equation}
\end{lemma}

\begin{proof}
To begin with, let us first consider a single density operator, $\rho_x $, and define the channel
$$
\Gamma(\rho_x):=\ket{\phi}\bracket{\phi}{\rho_x}{\phi}\bra{\phi}+\Phi^\perp\rho_x\Phi^\perp,
$$
where $\ket{\phi}$ is a pure state in $\cH$, and $\Phi^\perp:=\1-\ket{\phi}\bra{\phi}$. The channel $\Gamma$ is mixing-enhancing, i.e., $S(\rho_x)\leq S\left[\Gamma(\rho_x)\right]$. Moreover, $\Gamma$ transforms any input state into a block-diagonal state, which implies that for any function, $f$, and any input state, $\rho$, it holds that $
f\left[\Gamma(\rho)\right]=f\left(\ket{\phi}\bracket{\phi}{\rho}{\phi}\bra{\phi}\right)+f\left(\Phi^\perp\rho\Phi^\perp\right)$. Using these two properties of $\Gamma$, we obtain
\begin{align*}
S(\rho_x)&\leq H_B\left[q(x)\right]+q(x)S\left[\dfrac{\Phi^\perp\rho_x\Phi^\perp}{\Tr\left(\Phi^\perp\rho_x\right)}\right], \\
&\leq H_B\left[q(x)\right]+q(x)\log\left(\dim\cH-1\right),
\end{align*}
where we have defined $q(x):=\Tr(\Phi^\perp\rho_x)=1-\bracket{\phi}{\rho_x}{\phi}$ for $x=0,\dots,M$, and recall that $H_B(t) = -t\log t -(1-t)\log(1-t)$ is the binary entropy.  Note that we can choose $\ket{\phi}$ to be the normalized eigenvector corresponding to the largest eigenvalue of $\rho_x$, which we denote as $\lambda^{\downarrow}_1(\rho_x)$. Then, we have $q(x) =1-\lambda^{\downarrow}_1(\rho_x)=1-\Vert\rho_x\Vert$.

Considering now the whole set of density operators, $\{\rho_x\}_{x=0}^{M}$, we have by the concavity of the entropy that
\begin{align}
\label{e:dfbsfbg}
\sum_{x=0}^{M-1}p(x) S(\rho_x)&\leq H_B\left(Q\right)+Q\log\left(\dim\cH-1\right),
\end{align}
where $Q$ is defined in Eq.~(\ref{e:Q}).

One can next bound the binary entropy, as $H_B(t)\leq t-t\log t$ on $0\leq t\leq 1$. By combining this observation with (\ref{e:dfbsfbg}), we obtain (\ref{e:lemma1}).

\end{proof}

The average purity, $Q$, defined in Eq.~(\ref{e:Q}) can be bounded if one considers some structure on the density operators and the probabilities. In particular, taking $\rho_x=\Psi_C(x_B)$ and $p(x)=p_\Psi(x_B)$ (see Eq.~(\ref{e:MPS}) and Eq.~(\ref{e:prob})), the average purity is $Q=1-K^{-2}\sum_{x_N,\ldots,x_1=0}^{d-1}\Vert FA_{x_N}\cdots A_{x_1}\sigma A_{x_1}^\dag\cdots A_{x_N}^\dag F^\dag\Vert$, where recall that $\sigma=\bE^{|A|}(L)$ and $F^\dag F=\bE^{*|C|}(R)$. An upper and a lower bound on Q are stated and shown in the following lemma.

\begin{lemma}\label{lem2:AppA}
Let $\{A_{x}\}_{x=0}^{d-1}$ be a collection of operators on a Hilbert space, $\cH$, with $D:=\dim\cH\leq+\infty$, such that $\sum_{x=0}^{d-1}A_{x}^\dag A_{x}=\1$. Let $\sigma$ be a density operator on $\mathcal{H}$, and $F$ an operator on $\mathcal{H}$, such that $F^{\dagger}F\leq \1$. Then,
\begin{equation}\label{e:boundsPbar}
\tfrac{1}{K^2}\sum_{x_B}\lambda^{\downarrow}_2\left(FA_{x_N}\cdots A_{x_1}\sigma A_{x_1}^\dag\cdots A_{x_N}^\dag F^\dag\right)\leq Q\leq\tfrac{D-1}{K^2}\sum_{x_B}\lambda^{\downarrow}_2\left(FA_{x_N}\cdots A_{x_1}\sigma A_{x_1}^\dag\cdots A_{x_N}^\dag F^\dag\right),
\end{equation}
where $x_B:= x_1,\dots,x_N$, and where the sum over $x_B$ spans all of $\{0,\ldots,d-1\}^{N}$. Moreover, $K$ is a normalization constant such that
\begin{equation}\label{e:K^2}
K^2:=\sum_{x_B}\Tr\left(FA_{x_N}\cdots A_{x_1}\sigma A_{x_1}^\dag\cdots A_{x_N}^\dag F^\dag\right);
\end{equation}
and $Q$ is
\begin{equation}
Q=1-\dfrac{1}{K^2}\sum_{x_B}\left\Vert FA_{x_N}\cdots A_{x_1}\sigma A_{x_1}^\dag\cdots A_{x_N}^\dag F^\dag\right\Vert.
\end{equation}
\end{lemma}

\begin{proof}
Consider a positive semi-definite operator, $\rho\geq0$, and define a function, $L$, on $\rho$ such that
$$
L(\rho):=\sum_{j=2}^D\lambda^{\downarrow}_j(\rho).
$$
This function $L(\rho)$ can be upper and lower bounded as
\begin{equation}\label{e:boundsL}
\lambda^{\downarrow}_2(\rho)\leq L(\rho)\leq(D-1)\lambda^{\downarrow}_2(\rho).
\end{equation}
Moreover, it holds that
\begin{equation}\label{e:equalityL}
\lambda^{\downarrow}_1(\rho)+L(\rho)=\Tr(\rho).
\end{equation}
If we introduce in Eq.~(\ref{e:equalityL}) the positive operator $\rho:=K^{-2}FA_{x_N}\cdots A_{x_1}\sigma A_{x_1}^\dag\cdots A_{x_N}^\dag F^\dag$ and we sum over all possible values of $x_B:= x_1,\dots,x_N$, we obtain
\begin{align*}
\dfrac{1}{K^2}\sum_{x_B}\lambda^{\downarrow}_1\left(FA_{x_N}\cdots A_{x_1}\sigma A_{x_1}^\dag\cdots A_{x_N}^\dag F^\dag\right)+&\dfrac{1}{K^2}\sum_{x_B}L\left(FA_{x_N}\cdots A_{x_1}\sigma A_{x_1}^\dag\cdots A_{x_N}^\dag F^\dag\right) \\
&=\dfrac{1}{K^2}\sum_{x_B}\Tr\left(FA_{x_N}\cdots A_{x_1}\sigma A_{x_1}^\dag\cdots A_{x_N}^\dag F^\dag\right), \\
&=1,
\end{align*}
where the last equality holds due to the definition of $K$ in Eq.~(\ref{e:K^2}). This implies that
\begin{align*}
\dfrac{1}{K^2}\sum_{x_B}L\left(FA_{x_N}\cdots A_{x_1}\sigma A_{x_1}^\dag\cdots A_{x_N}^\dag F^\dag\right)&=1-\dfrac{1}{K^2}\sum_{x_B}\lambda^{\downarrow}_1\left(FA_{x_N}\cdots A_{x_1}\sigma A_{x_1}^\dag\cdots A_{x_N}^\dag F^\dag\right), \\
&=1-\dfrac{1}{K^2}\sum_{x_B}\left\Vert FA_{x_N}\cdots A_{x_1}\sigma A_{x_1}^\dag\cdots A_{x_N}^\dag F^\dag\right\Vert, \\
&=Q.
\end{align*}
Using the bounds of Eq.~(\ref{e:boundsL}), we finish the proof, since we obtain the bounds on $Q$ in Eq.~(\ref{e:boundsPbar}).

\end{proof}

Recall that we throughout this investigation assume that $\log$ denotes the natural logarithm. We state the following two lemmas without proof.

\begin{lemma}
\label{enhethnet}
Let 
\begin{equation}
H_B(t) := -t\log t -(1-t)\log(1-t),\quad 0 < t <1,
\end{equation}
and $H_B(0) := 0$ and $H_B(1) := 0$. Let 
\begin{equation}
g(t) :=  t -t\log t ,\quad 0 < t\leq 1, 
\end{equation}
and $g(0) := 0$.
Then, $g$ is monotonically increasing on $[0,1]$, and 
\begin{equation}
H_B(t)\leq g(t),\quad 0\leq t\leq 1.
\end{equation}
\end{lemma}
\begin{lemma}
\label{fgbfnfnhfg}
\begin{equation}
-t\log t \leq  \frac{1}{\epsilon}t^{1-\epsilon},\quad 0\leq t\leq 1,\quad 0 <\epsilon < 1.
\end{equation}
\end{lemma}

We recall that if the channel $\mathbb{E}$ is primitive, then it follows that $\mathbb{E}$ has a unique full-rank fixed point $\rho$ \cite{primitivity}. With the replacement-map $\mathcal{R}(\sigma) := \rho\Tr(\sigma)$, the fact that every initial state $\sigma$ converges to $\rho$ can be expressed as $\lim_{N\rightarrow\infty}\mathbb{E}^N = \mathcal{R}$. Since the underlying Hilbert space is finite-dimensional, we can express the convergence in terms of any norm. It is convenient to express the convergence in terms of the norm
\begin{equation}
\begin{split}
\Vert \mathcal{F} \Vert_{1:1} := \sup_{\Vert Q\Vert_1= 1} \Vert \mathcal{F}(Q)\Vert_1,
\end{split}
\end{equation}
and thus $\lim_{N\rightarrow \infty}\Vert \mathbb{E}^{N} - \mathcal{R}\Vert_{1:1} = 0$, where $\Vert Q\Vert_1 := \Tr\sqrt{Q^{\dagger}Q}$ is the trace norm.
\begin{lemma}
\label{lemmaadflbdlfk}
Let $\{A_x\}_{x=0}^{d-1}$ be operators on a Hilbert space $\mathcal{H}$, with $D := \dim\mathcal{H} <+\infty$, such that $\sum_{x=0}^{d-1}A_x^{\dagger}A_x = \1$, and $\mathbb{E}(\cdot) := \sum_{x=0}^{d-1}A_x \cdot A_x^{\dagger}$ is primitive. Then, there exists a real number $r>0$ and a natural number $N_0$ such that
\begin{equation}
 \langle R|\mathbb{E}^{|A|+|B|+|C|}(|L\rangle\langle L|)|R\rangle \geq r,\quad  \forall|A|,|C|,\,\,\forall \Vert R\Vert = 1,\Vert L\Vert = 1,\,\, \forall |B|\geq N_0.
\end{equation} 
\end{lemma}

\begin{proof}
For the map $\mathcal{R}(\sigma) := \rho\Tr(\sigma)$ with $\rho$ the unique fixed point $\rho$ of $\mathbb{E}$, we first observe that
\begin{align}\label{dvadfbadfb}
\nonumber\Big\vert\langle R|\mathbb{E}^{|A|+|B|+|C|}(|L\rangle\langle L|)|R\rangle  - \langle R|\rho|R\rangle\Big\vert = & \Big\vert \Tr\Big(|R\rangle\langle R| \big(\mathbb{E}^{|A|+|B|+|C|}(|L\rangle\langle L|)  - \rho\big)\Big)\Big\vert,\\
\nonumber\leq & \big\Vert |R\rangle\langle R|\big\Vert\big\Vert \mathbb{E}^{|A|+|B|+|C|}(|L\rangle\langle L|)  - \rho\big\Vert_1,\\
\nonumber = & \Big\Vert \mathbb{E}^{|B|}\big(\mathbb{E}^{|A|+|C|}(|L\rangle\langle L|)\big)  - \mathcal{R}\big(\mathbb{E}^{|A|+|C|}(|L\rangle\langle L|)\big)\Big\Vert_1,\\
\nonumber\leq & \big\Vert \mathbb{E}^{|B|} - \mathcal{R}\big\Vert_{1:1}\Vert  \mathbb{E}^{|A|+|C|}(|L\rangle\langle L|)\Vert_1,\\
=& \big\Vert \mathbb{E}^{|B|} - \mathcal{R}\big\Vert_{1:1}.
\end{align}
Since $\mathbb{E}$ is assumed to be primitive, it follows that $\mathbb{E}$ has a unique full rank fixed point $\rho$. Since $\mathcal{H}$ is assumed to be finite-dimensional, it follows that the minimal eigenvalue of $\rho$ is such that  $\lambda_{\mathrm{min}}(\rho)>0$. By (\ref{dvadfbadfb}), it follows that 
\begin{equation}
\begin{split}
\lambda_{\mathrm{min}}(\rho)  -\big\Vert \mathbb{E}^{|B|} - \mathcal{R}\big\Vert_{1:1} \leq  &  \langle R|\rho|R\rangle-\big\Vert \mathbb{E}^{|B|} - \mathcal{R}\big\Vert_{1:1},\\
 \leq  & \langle R|\mathbb{E}^{|A|+|B|+|C|}(|L\rangle\langle L|)|R\rangle.
\end{split}
\end{equation}
Since $\lim_{|B|\rightarrow\infty}\big\Vert \mathbb{E}^{|B|} - \mathcal{R}\big\Vert_{1:1}=0$ and $\lambda_{\mathrm{min}}(\rho)>0$, it follows that there exists an $r$ such that $\lambda_{\mathrm{min}}(\rho) > r>0$ and a $N_0$, such that 
 \begin{equation}
\begin{split}
 \langle R|\mathbb{E}^{|A|+|B|+|C|}(|L\rangle\langle L|)|R\rangle\geq r,\quad \forall |B|\geq N_0.
\end{split}
\end{equation}
One should note that $r$ and $N_0$ are independent of  $|A|$, $|C|$, and all normalized  $|R\rangle$ and $|L\rangle$.
\end{proof}

Theorem \ref{thm:main} in the main text follows as a direct corollary of Theorem \ref{hndghngdh} below with $\kappa := \gamma^{1-\epsilon}$  and $c := c_{\epsilon}$ for any fixed $0< \epsilon <1$. In essence, we use the bound $f(N)\leq \overline{c}\gamma^{N}$ in  Proposition \ref{PropMain} in order to prove the bound in Theorem \ref{hndghngdh}, and thus it is the same $\gamma$  that appears in both bounds. The reason for the transition from $\gamma$ to $\gamma^{1-\epsilon}$ is loosely speaking due to a leading order term proportional to $|B|\gamma^{|B|}$. This term appears in a bound on the CMI and can be accommodated by an arbitrarily small sacrifice of the rate in the exponential decay. However, since we here are not only interested in the asymptotics, but rather wish to achieve a general bound valid for all values of $|B|$, the construction in the proof becomes more elaborate. 
\begin{thm}
\label{hndghngdh}
For a set of operators $\{A_x\}_{x=0}^{d-1}$ on a Hilbert space $\mathcal{H}$ with $D := \dim\mathcal{H}\geq 2$, and normalized $|R\rangle,|L\rangle\in\mathcal{H}$, let  $\Psi$ be the MPS as defined in (\ref{e:MPS}) on a region $\Lambda = ABC$. The set  $\{A_x\}_{x=0}^{d-1}$ is such that $\sum_{x=0}^{d-1}A_x^{\dagger}A_x = \1$ satisfies the purity condition in Definition \ref{DefPur}, and is such that $\mathbb{E}(\cdot) :=\sum_{x=0}^{d-1}A_x \cdot A_x^{\dagger}$ is primitive. 
For the constant $\gamma$ as guaranteed by Proposition \ref{PropMain}, and for every $0<\epsilon <1$, there exists a constant $c_{\epsilon} \geq 0$ such that 
\begin{equation}
\begin{split}
   I_{p_{\Psi}}(A:C|B) \leq I_{\Phi_B(\Psi)}(A:C|B)\leq & c_{\epsilon} \gamma^{|B|(1-\epsilon)}, \quad |B| = 1,2,\ldots\hspace{0.1cm}.
\end{split}
\end{equation}
The constant $\gamma$ is independent of $|A|$, $|B|$, $|C|$, $|L\rangle$, $|R\rangle$ and $\epsilon$. The constant $c_{\epsilon}$ is independent of $|A|$, $|B|$, $|C|$, $|L\rangle$ and $|R\rangle$, but may depend on $\epsilon$.
\end{thm}

\begin{proof}
We first note that
\begin{equation}
\label{adfbadfbdf}
\begin{split}
 I_{\Phi_B(\Psi)}(A:C|B) &=\langle S\left[\Psi_A(x_B)\right]\rangle_{p_\Psi(x_B)}+\langle S\left[\Psi_C(x_B)\right]\rangle_{p_\Psi(x_B)},\\
&=2\langle S\left[\Psi_C(x_B)\right]\rangle_{p_\Psi(x_B)},
\end{split}
\end{equation}
where we recall that $p_{\Psi}(x_{\Lambda}) = \langle x_{\Lambda}|\Psi|x_{\Lambda}\rangle$, and
where the state $\Psi_{X}(x_B)$ is the reduced state in region $X$ of the post-measurement state, $\Psi(x_B)$, and $\langle S\left[\Psi(x)\right]\rangle_{p_{\Psi(x)}}$ is the average von Neumann entropy 
\begin{equation} \label{e:averageS1}
\langle S\left[\Psi_C(x_B)\right]\rangle_{p_\Psi(x_B)}=  \sum_{x_B}p_\Psi(x_B)S(\Psi_C(x_B)),
\end{equation}
with 

\begin{equation}
\label{ghmdghm}
\begin{split}
p_\Psi(x_B) := & \dfrac{1}{K^2}\Tr\left[A_{x_N}\cdots A_{x_1}\mathbb{E}^{|A|}(L)A^{\dagger}_{x_1}\cdots A^{\dagger}_{x_N}\mathbb{E}^{*|C|}(R)\right],\\
= & \dfrac{1}{K^2}\Tr\left[FA_{x_N}\cdots A_{x_1}\sigma A^{\dagger}_{x_1}\cdots A^{\dagger}_{x_N}F^{\dagger}\right],
\end{split}
\end{equation}

 and
\begin{equation}
\label{fsgnsfgnmgfhm}
\Psi_C(x_B) :=  \dfrac{1}{p_\Psi(x_B)K^2}FA_{x_N}\cdots A_{x_1}\sigma A_{x_1}^\dag\cdots A_{x_N}^\dag F^\dag,
\end{equation}
where
\begin{equation}
\label{fgnsrgmsmr}
\begin{split}
F := & \sum_{x_{|BC|},\ldots,x_{N+1}}|x_{x_{|BC|},\ldots,x_{N+1}}\rangle\langle R|A_{x_{|BC|},\ldots,x_{N+1}},\\
\sigma:= &  \mathbb{E}^{|A|}(L).
\end{split}
\end{equation}

The equality (\ref{adfbadfbdf}) follows from the fact that the post-measurement state is pure, and thus the reduced states on regions $A$ and $C$ are isospectral (up to zero eigenvalues). With $\rho_{x_B} := \Psi_C(x_B)$ in Lemma  \ref{lem:1}, we know that
\begin{equation}
\label{e:firststep}
\begin{split}
\langle S\left[\Psi_C(x_B)\right]\rangle_{p_\Psi(x_B)}\leq  & H_B\left(Q\right)+Q\log(D-1),
\end{split}
\end{equation}
with 
\begin{equation}
\label{fdadfbadfb}
Q:=1-\sum_{x_B}p_\Psi(x_B)\Vert \Psi_C(x_B)\Vert.
\end{equation}

By Lemma \ref{enhethnet} we know that the function $g(t) = t -t\log t$ is monotonically increasing on the interval $[0,1]$ and satisfies $H_B(t)\leq g(t)$. By combining this observation with (\ref{e:firststep}), we get
\begin{equation}\label{e:dfbfsbs}
\langle S\left[\Psi_C(x_B)\right]\rangle_{p_\Psi(x_B)} \leq g(Q) +Q\log(D-1) = -Q\log Q  + Q +Q\log(D-1).
\end{equation}
By Lemma \ref{lem2:AppA} we furthermore know that 
\begin{equation}
\begin{split}
Q\leq &\dfrac{D-1}{K^2}\sum_{x_B}\lambda^{\downarrow}_2\left(FA_{x_N}\cdots A_{x_1}\sigma A_{x_1}^\dag\cdots A_{x_N}^\dag F^\dag\right), \\
\end{split}
\end{equation}
where $\lambda^{\downarrow}_j(O)$ are the eigenvalues of an operator $O$ in non-increasing order, i.e., $\lambda^{\downarrow}_1(O)\geq\cdots\geq\lambda^{\downarrow}_D(O)$.
Similarly, we let in the following $\nu^{\downarrow}_j(O)$ denote the singular values of $O$ in non-increasing order $\nu^{\downarrow}_1(O)\geq\cdots\geq\nu^{\downarrow}_D(O)$.
Then, recalling that for any operator $O$, we have $\lambda_j(OO^\dag)=\nu_j(O)^2$, we get
\begin{equation}
\label{eghnetnheh}
\begin{split}
Q\leq &\dfrac{D-1}{K^2}\sum_{x_B}\lambda^{\downarrow}_2\left(FA_{x_N}\cdots A_{x_1}\sigma A_{x_1}^\dag\cdots A_{x_N}^\dag F^\dag\right), \\
\leq&\dfrac{D-1}{K^2}\sum_{x_B}\sqrt{\lambda^{\downarrow}_1(FA_{x_N}\cdots A_{x_1}\sigma A_{x_1}^\dag\cdots A_{x_N}^\dag F^\dag)\lambda^{\downarrow}_2(FA_{x_N}\cdots A_{x_1}\sigma A_{x_1}^\dag\cdots A_{x_N}^\dag F^\dag)},\\
=&\dfrac{D-1}{K^2}\sum_{x_B} \nu^{\downarrow}_1(FA_{x_N}\cdots A_{x_1}\sqrt{\sigma})\nu^{\downarrow}_2(FA_{x_N}\cdots A_{x_1}\sqrt{\sigma}),\\
= & \dfrac{D-1}{K^2}f(N),\\
& [\textrm{By Proposition \ref{PropMain}}]\\
\leq & \dfrac{D-1}{K^2}\overline{c}\gamma^N,\\
= & \frac{D-1}{\langle R|\mathbb{E}^{|A|+|B|+|C|}(|L\rangle\langle L|)|R\rangle}\overline{c}\gamma^{|B|},
\end{split}
\end{equation}
where recall that $N = |B|$ and the constant $\overline{c}$ and $0< \gamma<1$ are independent of $\sigma := \mathbb{E}^{|A|}(L)$ and   $F$ as in (\ref{mainFdef}), and consequently are independent of $|A|$ and $|C|$ (as well as of $|B|$).

By Lemma \ref{lemmaadflbdlfk}, there exist constants $r>0$ and $N_0$ such that $\langle R|\mathbb{E}^{|A|+|B|+|C|}(|L\rangle\langle L|)|R\rangle\geq r$ for all $|B|\geq N_0$. By Lemma \ref{lemmaadflbdlfk} we know that $r$ and $N_0$ do not depend on $|A|, |B|, |C|, |R\rangle,|L\rangle$.
By combining this observation with (\ref{eghnetnheh}), we can conclude that 
\begin{equation}
\label{sfgnsgns}
\begin{split}
Q \leq \tilde{c}\gamma^{|B|},\quad\text{with}\quad\tilde{c} := \frac{D-1}{r}\overline{c},\quad \forall |B|\geq N_0,
\end{split}
\end{equation}
where we note that $\tilde{c}$  and $N_0$  do not depend on $|A|, |B|, |C|, |R\rangle,|L\rangle$. 
By inspection of the definition of $Q$ in (\ref{fdadfbadfb}), one can see that 
\begin{equation}
\label{gnfgnsgn}
Q\leq 1
\end{equation}
is trivially true. By combining (\ref{sfgnsgns}) and (\ref{gnfgnsgn}), we thus get
\begin{equation}
\label{dfbadfba}
Q\leq t, \quad \forall |B|\geq N_0,\quad \text{with}\quad  t:= \min\Big[1,\tilde{c}\gamma^{|B|}\Big], 
\end{equation}
where $t$ by necessity is contained in the interval $[0,1]$. 

We next combine  (\ref{e:dfbfsbs}) and (\ref{dfbadfba}) with the monotonicity of  $g$ to obtain
\begin{equation}
\begin{split}
\langle S\left[\Psi_C(x_B)\right]\rangle_{p_\Psi(x_B)}
\leq  &  g(Q) +Q\log(D-1), \\
& [\textrm{Monotonicity of $g$, Lemma \ref{enhethnet}, together with (\ref{dfbadfba})}]\\
\leq & g(t) + t\log(D-1),\\
=  & -t\log t  + t\left[1  + \log(D-1)\right],\\
& [\textrm{By Lemma  \ref{fgbfnfnhfg}}]\\
\leq  &  \frac{1}{\epsilon}t^{1-\epsilon}   + t\left[1  + \log(D-1)\right],\\
& [\textrm{By $t \leq t^{1-\epsilon},\quad 0\leq t\leq 1,\quad 0< \epsilon < 1$}]\\
\leq  &  \left[\frac{1}{\epsilon}   + 1  + \log(D-1)\right]t^{1-\epsilon}.
\end{split}
\end{equation}
Since  $t= \min\Big[1,\tilde{c}\gamma^{|B|}\Big]\leq \tilde{c}\gamma^{|B|}$ we get $t^{1-\epsilon} \leq \tilde{c}^{1-\epsilon}\gamma^{|B|(1-\epsilon)}$,
and thus
\begin{equation}
\begin{split}
\langle S\left[\Psi_C(x_B)\right]\rangle_{p_\Psi(x_B)}\leq  &  \left[\frac{1}{\epsilon}   + 1  + \log(D-1)\right] \tilde{c}^{1-\epsilon}\gamma^{|B|(1-\epsilon)}.
\end{split}
\end{equation}
By combining this with (\ref{adfbadfbdf}) we get
\begin{equation}
\begin{split}
 I_{\Phi_B(\Psi)}(A:C|B) \leq  &  \tilde{c}_{\epsilon} \gamma^{|B|(1-\epsilon)},\quad  \forall |B|\geq N_0,\\
  \text{with}\quad\tilde{c}_{\epsilon}  = &   2\left[\frac{1}{\epsilon}   + 1  + \log(D-1)\right] \tilde{c}^{1-\epsilon}.
\end{split}
\end{equation}
Finally we should remove the restriction that $|B|\geq N_0$. By (\ref{adfbadfbdf}) and (\ref{e:averageS1}) we can conclude that 
$I_{\Phi_B(\Psi)}(A:C|B) =  2\sum_{x_B}p_\Psi(x_B)S\big(\Psi_C(x_B)\big)\leq 2\log D$, where the last inequality follows since $\Psi_C(x_B)$ is (up to zero eigenvalues) isospectral to the density operator in (\ref{e:isoreducedstate}), and thus the entropy of these two states are equal. The state in (\ref{e:isoreducedstate}) is a density operator on $\mathcal{H}$, which has dimension $D$, and thus the entropy is bounded by $\log D$. Let
\begin{equation}
c_{\epsilon} := \max\left(\tilde{c}_{\epsilon}, 2\log(D)\gamma^{-(N_0-1)(1-\epsilon)}\right).
\end{equation}
One can confirm that this guarantees that 
\begin{equation}
\label{dfbsfgnfgnf}
 I_{\Phi_B(\Psi)}(A:C|B)\leq c_{\epsilon}\gamma^{|B|(1-\epsilon)}
\end{equation}
for all $|B| = 1,2,\ldots$ . The resulting constant $c_{\epsilon}$ is independent of $|A|$, $|B|$, $|C|$, $|L\rangle$ and $|R\rangle$. By combining (\ref{dfbsfgnfgnf}) with 
\begin{equation}
I_{p_{\Psi}}(A:C|B)  = I_{\Phi_{\Lambda}(\Psi)}(A:C|B) \leq I_{\Phi_B(\Psi)}(A:C|B),
\end{equation}
we obtain 
\begin{equation}
 I_{p_{\Psi}}(A:C|B)\leq c_{\epsilon}\gamma^{|B|(1-\epsilon)}.
\end{equation}

\end{proof}

By combining Lemma \ref{fbsfgbsfg} with Theorem \ref{hndghngdh} (for the inequality $I_{p_{\Psi}}(A:C|B) \leq  c_{\epsilon} \gamma^{|B|(1-\epsilon)}$), and defining $\kappa := \gamma^{1-\epsilon}$ and $c := c_{\epsilon}$ for some arbitrary but fixed $0<\epsilon< 1$, we get the following.

\begin{corollary}
\label{nzrgnthmtm}
For a set of operators $\{A_x\}_{x=0}^{d-1}$ on a Hilbert space $\mathcal{H}$ with $D := \dim\mathcal{H}\geq 2$, and normalized $|R\rangle,|L\rangle\in\mathcal{H}$, let  $\Psi$ be the MPS as defined in (\ref{e:MPS}) on a region $\Lambda = ABC$. The set  $\{A_x\}_{x=0}^{d-1}$ is such that $\sum_{x=0}^{d-1}A_x^{\dagger}A_x = \1$, satisfies the purity condition in Definition \ref{DefPur}, and is such that $\mathbb{E}(\cdot) :=\sum_{x=0}^{d-1}A_x \cdot A_x^{\dagger}$ is primitive.
Let $p_{1,\ldots,|\Lambda|}(x_1,\ldots,x_{|\Lambda|})= \langle x_{\Lambda}|\Psi| x_{\Lambda}\rangle$ be the classical restriction of $\Psi$, and assume that this restriction is such that $p_{1,\ldots,|\Lambda|}(x_1,\ldots,x_{|\Lambda|})>0$ for all $x_1,\ldots,x_{|\Lambda|}$. Let $p^{\ell}_{1,\ldots,|\Lambda|}$ be as defined in (\ref{fbsfgnsfgn1}-\ref{dvadfv1}). 
Then, there exist constants, $0 \leq c$ and $0 < \kappa < 1$,  such that
\begin{equation}
\begin{split}
S(p_{1,\ldots,|\Lambda|}\Vert p^{\ell}_{1,\ldots,|\Lambda|})  \leq & c|\Lambda|\kappa^{\ell},\quad 1\leq \ell \leq |\Lambda|-2.
\end{split}
\end{equation}
\end{corollary}

\subsection{\label{app:SecondHalf} Proof of the second part of Theorem \ref{thm:main}}

In this section we prove the second part of Theorem \ref{thm:main}, i.e., that the exponential decay of the quantum CMI $I_{\Phi_B(\Psi)}(A:C|B)$ implies purity of the set of matrices $\{A_x\}_{x=0}^{d-1}$ associated to $\ket{\Psi}$. In order to do this, we show several bounds that all combined will allow us to bound $f(N)$, and thus the quantum CMI.

The first step is to find a lower bound to the binary entropy.
\begin{lemma}
\label{ngfnsfgnsf}
Define $H_B:[0,1]\rightarrow\mathbb{R}$ by
\begin{equation}
\label{bgfnfgn}
\begin{split}
H_B(\lambda) := & -\lambda\log\lambda -(1-\lambda)\log(1-\lambda),\quad 0< \lambda<1,\\
H_B(0) := & 0,\quad H_B(1) := 0.
\end{split}
\end{equation}
Then,
\begin{equation}
\label{dfvdfbbg}
H_B(\lambda) \geq  4\log(2)  \lambda(1-\lambda),\quad 0\leq  \lambda \leq 1.
\end{equation}
\end{lemma}
\begin{proof}
Define
\begin{equation}
\begin{split}
f(\lambda) := & \frac{H_B(\lambda)}{\lambda(1-\lambda)}=  g(\lambda) + g(1-\lambda),\quad\quad g(\lambda) :=  -\frac{\log(1-\lambda)}{\lambda}.
\end{split}
\end{equation}
It is clear that $f$ is symmetric around $\lambda = 1/2$, i.e., it is symmetric under the map $\lambda \mapsto 1-\lambda$.
In the following, we will prove that $g$ is a convex function on the interval $[0,1]$, with the consequence that $f$ also is a convex function. 

Consider the function $r(\lambda) := e^{-\lambda^2/2-\lambda}$. One can confirm that $r''(\lambda)\geq 0$ for all $\lambda\geq 0$, and thus $r$ is convex on $[0,+\infty)$. Consequently, the tangent of $r$ at $\lambda = 0$ is a lower bound to $r$ on  $[0,+\infty)$. Since the tangent at $\lambda = 0$ is $1-\lambda$, we can conclude that
\begin{equation}
\label{bdfbsfdgnsfg}
 e^{-\lambda^2/2-\lambda}\geq 1-\lambda,\quad \lambda \geq 0.
\end{equation} 
Assuming that $1>\lambda >0$, we can rewrite (\ref{bdfbsfdgnsfg}) as
\begin{equation}
\label{fnsfgnfgsfm}
g(\lambda)  \geq  \frac{\lambda}{2}+1,
\end{equation}
with the result that 
\begin{equation}
\begin{split}
g''(\lambda)  = & -2\lambda^{-2}(1-\lambda)^{-1}  +\lambda^{-1}(1-\lambda)^{-2}   +2\lambda^{-2}g(\lambda) 
\geq  \frac{\lambda}{(1-\lambda)^2} \geq  0.
\end{split}
\end{equation}
Hence $g$ is convex on $0< \lambda < 1$. As mentioned above, it thus follows that $f$ is a convex function. Since $f$ moreover is  symmetric under the map $\lambda\mapsto 1-\lambda$, it follows that the minimum is attained at $\lambda = 1/2$. 
Hence,  $f(\lambda) \geq f(1/2) =  4\log 2$,
which in turn yields (\ref{dfvdfbbg}).
\end{proof}

While the above lemma bounds the binary entropy, the following lemma  relates it to the von Neuman entropy of a density operator.

\begin{lemma}
\label{gfnsfnsmsr}
Let $\rho$ be a density operator on a finite-dimensional Hilbert space. Then
\begin{equation}
H_B\big(\lambda^{\downarrow}_1(\rho)\big)\leq S(\rho),
\end{equation}
where $H_B$ is as defined in (\ref{bgfnfgn}), and $\lambda^{\downarrow}_1(\rho)$ denotes the largest eigenvalue of $\rho$.
\end{lemma}
\begin{proof}
Recall that a  vector $\boldsymbol{a}\in\mathbb{R}^N$  majorizes a vector  $\boldsymbol{b}\in\mathbb{R}^N$  (denoted $\boldsymbol{b}\prec\boldsymbol{a}$) if
$\sum_{j=1}^{k}\boldsymbol{b}^{\downarrow}_j \leq \sum_{j=1}^{k}\boldsymbol{a}^{\downarrow}_j$ for $k = 1,\ldots, N-1$, and 
$\sum_{j=1}^{N}\boldsymbol{b}^{\downarrow}_j = \sum_{j=1}^{N}\boldsymbol{a}^{\downarrow}_j$,
where $\boldsymbol{a}^{\downarrow}$  denotes the vector that we obtain by permuting the components of $\boldsymbol{a}$, such that they occur in a non-increasing order $\boldsymbol{a}^{\downarrow}_1\geq \boldsymbol{a}^{\downarrow}_2\geq\cdots \geq\boldsymbol{a}^{\downarrow}_N$, and analogous for $\boldsymbol{b}^{\downarrow}$.
 Also recall that if  $\boldsymbol{a}$ and $\boldsymbol{b}$ can be regarded as probability distributions, then $\boldsymbol{b}\prec\boldsymbol{a}$
implies that the Shannon entropy of $\boldsymbol{a}$ is lower than the Shannon entropy of $\boldsymbol{r}$, i.e., $H(\boldsymbol{a}) \leq H(\boldsymbol{b})$ \cite{Wehrl}.

Let $N$ be the dimension of the Hilbert space, and consider the vector of ordered eigenvalues 
$\boldsymbol{b} := \big(\lambda_{1}^{\downarrow}(\rho),\ldots, \lambda_{N}^{\downarrow}(\rho)\big)$.
We also consider the  vector $\boldsymbol{a} := \big(\lambda_{1}^{\downarrow}(\rho), 1-\lambda_{1}^{\downarrow}(\rho),0,\ldots, 0\big)$. We notice that both $\boldsymbol{a}$ and $\boldsymbol{b}$ can be regarded as probability distributions.
In both cases, $\lambda_{1}^{\downarrow}(\rho) \geq 1 - \lambda_{1}^{\downarrow}(\rho)$ and $\lambda_{1}^{\downarrow}(\rho) < 1 - \lambda_{1}^{\downarrow}(\rho)$, one can confirm that $\boldsymbol{b} \prec \boldsymbol{a}$, and consequently $H(\boldsymbol{a}) \leq H(\boldsymbol{b})$. The claim of the lemma follows by the observations that $H(\boldsymbol{a})  = H_B\big(\lambda^{\downarrow}_1(\rho)\big)$ and $H(\boldsymbol{b}) = S(\rho)$.
\end{proof}

Next, we  consider a collection of density operators and use Lemmas \ref{ngfnsfgnsf} and \ref{gfnsfnsmsr} to bound the average von Neumann entropy of these density operators.

\begin{lemma}
\label{fgnsfgnfg}
Let $\{\rho_x\}_{x = 0}^{M-1}$ be density operators on a finite-dimensional Hilbert space, and let $\{p_x\}_{x = 0}^{M-1}$ be such that $p_x\geq 0$ and $\sum_{x = 0}^{M-1}p_x = 1$. Then,
\begin{equation}
\label{fdbdfgn}
\sum_{x=0}^{M-1}p_x\sqrt{\lambda_1^{\downarrow}(\rho_x)\lambda_2^{\downarrow}(\rho_x)}\leq \frac{1}{2\sqrt{\log(2)}}\sqrt{\sum_{x=0}^{M-1}p_xS(\rho_x)}.
\end{equation}
\end{lemma}
\begin{proof}
We first note that due to the convexity of $x\mapsto x^2$, we have
\begin{equation}
\label{aerbaenrtgn}
\begin{split}
 \left(\sum_{x=0}^{M-1}p_x\sqrt{\lambda_1^{\downarrow}(\rho_x)\lambda_2^{\downarrow}(\rho_x)}\right)^2
\leq  &  \sum_{x=0}^{M-1}p_x\left(\sqrt{\lambda_1^{\downarrow}(\rho_x)\lambda_2^{\downarrow}(\rho_x)}\right)^2, \\
= &  \sum_{x=0}^{M-1}p_x \lambda_1^{\downarrow}(\rho_x)\lambda_2^{\downarrow}(\rho_x).
\end{split}
\end{equation}
Next we note that $1-\lambda_1^{\downarrow}(\rho_x) =   \sum_{k=2}^{D}\lambda_k^{\downarrow}(\rho_x)
\geq  \lambda_2^{\downarrow}(\rho_x)$, which implies $\lambda_1^{\downarrow}(\rho_x)\lambda_2^{\downarrow}(\rho_x)\leq \lambda_1^{\downarrow}(\rho_x)\big(1-\lambda_1^{\downarrow}(\rho_x)\big)$.
This, combined with (\ref{aerbaenrtgn}), yields
\begin{equation}
\label{bsfgsfgn}
\begin{split}
 4\log(2)\left(\sum_{x=0}^{M-1}p_x\sqrt{\lambda_1^{\downarrow}(\rho_x)\lambda_2^{\downarrow}(\rho_x)}\right)^2
\leq  &  4\log(2)\sum_{x=0}^{M-1}p_x \lambda_1^{\downarrow}(\rho_x)\lambda_2^{\downarrow}(\rho_x),\\
\leq  &  \sum_{x=0}^{M-1}p_x 4\log(2)\lambda_1^{\downarrow}(\rho_x)\Big(1-\lambda_1^{\downarrow}(\rho_x)\Big),\\
& [\textrm{By Lemma \ref{ngfnsfgnsf}}]\\
\leq  &  \sum_{x=0}^{M-1}p_x H_B\big(\lambda_1^{\downarrow}(\rho_x)\big),\\
& [\textrm{By Lemma \ref{gfnsfnsmsr}}]\\
\leq  &  \sum_{x=0}^{M-1}p_x S(\rho_x),
\end{split}
\end{equation}
which implies (\ref{fdbdfgn}).
\end{proof}

The last result that we need before stating the second part of Theorem \ref{thm:main} is a bound on the function $f(N)$ in terms of the average entropy, which is obtained in the following lemma.

\begin{prop}
\label{nfggfnfd}
For a set of operators $\{A_x\}_{x=0}^{d-1}$ on a Hilbert space $\mathcal{H}$ with $D := \dim\mathcal{H}\geq 2$, and normalized $|R\rangle,|L\rangle\in\mathcal{H}$, let  $\Psi$ be the MPS as defined in (\ref{e:MPS}) on a region $\Lambda = ABC$. The set  $\{A_x\}_{x=0}^{d-1}$ is such that $\sum_{x=0}^{d-1}A_x^{\dagger}A_x = \1$. 
Then,
\begin{equation}
\label{qwfwabaer}
\begin{split}
f(N) \leq\frac{1}{2\sqrt{\log(2)}}\sqrt{\sum_{x_B}p_{\Psi}(x_B) S[\Psi_C(x_b)]},
\end{split}
\end{equation}
where $x_B = x_1,\ldots, x_N$ and
\begin{equation}
\label{tmsfrmfmh}
f(N) :=\sum_{x_B}\nu_1^{\downarrow}(FA_{x_N}\cdots A_{x_1}\sqrt{\sigma})\nu_2^{\downarrow}(FA_{x_N}\cdots A_{x_1}\sqrt{\sigma}),
\end{equation}
and where  $p_\Psi(x_B)$ is as defined in (\ref{ghmdghm}), $\Psi_C(x_b)$ as in (\ref{fsgnsfgnmgfhm}), as well as $F$ and $\sigma$ as in (\ref{fgnsrgmsmr}).
\end{prop}
\begin{proof}
First recall that if $\nu_k^{\downarrow}$ denotes the $k$:th singular value in non-increasing order, and $\lambda_{k}^{\downarrow}$ denotes the $k$:th eigenvalue in non-decreasing order, then $\nu_k^{\downarrow}(Q) = \sqrt{\lambda_{k}^{\downarrow}(Q^{\dagger}Q)}$. With this observation in mind, we find that 
\begin{equation}
\begin{split}
f(N)= &\sum_{x_B}\sqrt{\lambda_1^{\downarrow}\left(FA_{x_N}\cdots A_{x_1}\sigma A_{x_1}^{\dagger}\cdots A_{x_N}^{\dagger}F^{\dagger}\right)\lambda_2^{\downarrow}\left(FA_{x_N}\cdots A_{x_1}\sigma A_{x_1}^{\dagger}\cdots A_{x_N}^{\dagger}F^{\dagger}\right)},\\
= &K^2 \sum_{x_B}p_\Psi(x_B)\sqrt{\lambda_1^{\downarrow}\big(\Psi_C(x_b)\big)\lambda_2^{\downarrow}\big(\Psi_C(x_b)\big)},\\
&[\textrm{By Lemma \ref{fgnsfgnfg}}]\\
\leq   & K^2  \frac{1}{2\sqrt{\log(2)}} \sqrt{\sum_{x_B}p_\Psi(x_B) S[\Psi_C(x_b)]},\\
=  &  \Tr\big(F^{\dagger}F\mathbb{E}^{N}(\sigma)\big)  \frac{1}{2\sqrt{\log(2)}}  \sqrt{\sum_{x_B}p_\Psi(x_B) S[\Psi_C(x_b)]},\\
& \Big[\quad F^{\dagger}F\leq 1,\quad \mathbb{E}^{N}(\sigma)  \quad\textrm{is a density operator since $\mathbb{E}$ is a channel} \quad\Big]\\
\leq &    \frac{1}{2\sqrt{\log(2)}}\sqrt{\sum_{x_B}p_\Psi(x_B) S[\Psi_C(x_b)]}.
\end{split}
\end{equation}
\end{proof}

Finally, in the theorem below we use all bounds derived previously to prove that an exponential decay of $I_{\Phi_B(\Psi)}(A:C|B)$ implies purity. One may note that Theorem \ref{thm:main} assumes that the MPS is injective, while this is not strictly speaking needed for Theorem  \ref{TheoremSecondPart}, where the relevant assumption rather is that $\sigma:=\mathbb{E}^{|A|}(L)$, and $F^\dag F =\mathbb{E}^{*|C|}(R)$ are full rank operators.

\begin{thm}
\label{TheoremSecondPart}
For a set of operators $\{A_x\}_{x=0}^{d-1}$, such that $\sum_{x=0}^{d-1}A_x^{\dagger}A_x = \1$, on a Hilbert space $\mathcal{H}$ with $D := \dim\mathcal{H}\geq 2$, and normalized $|R\rangle,|L\rangle\in\mathcal{H}$, let  $\Psi$ be the MPS as defined in (\ref{e:MPS}) on a region $\Lambda = ABC$. Suppose that 
$|R\rangle$, $|L\rangle$, $|A|$, and $|C|$ are such that 
$\sigma:=\mathbb{E}^{|A|}(L)$, and $F^\dag F =\mathbb{E}^{*|C|}(R)$ are full rank operators, where $\mathbb{E}(\cdot) :=\sum_{x=0}^{d-1}A_x \cdot A_x^{\dagger}$. Moreover suppose that there  exist constants $\tilde{c}$ and $0 \leq \tilde{\kappa} < 1$, such that 
\begin{equation}
\label{vdfbdfbdf}
\begin{split}
  I_{\Phi_B(\Psi)}(A:C|B)\leq & \tilde{c} \tilde{\kappa}^{|B|}, \quad  |B| = 1,2,\ldots.
\end{split}
\end{equation}
Then, $\{A_x\}_{x=0}^{d-1}$ satisfies the purity condition in Definition \ref{DefPur}.
\end{thm}
\begin{proof}
We first note that 
\begin{equation}
I_{\Phi_B(\Psi)}(A:C|B) = 2\sum_{x_B}p_\Psi(x_B)  S[\Psi_C(x_b)],
\end{equation} 
with $ p_\Psi(x_B)$ as in (\ref{ghmdghm}) and $\Psi_C(x_b)$ as in (\ref{fsgnsfgnmgfhm}). 
By comparing with Proposition \ref{nfggfnfd} we can thus conclude that
\begin{equation}
\label{bdsfbdfg}
\begin{split}
f(N) \leq     \frac{1}{2\sqrt{2\log(2)}}\sqrt{I_{\Phi_B(\Psi)}(A:C|B)},
\end{split}
\end{equation}
with $f(N)$ as in (\ref{tmsfrmfmh}).
By assumption, there exist constants $\tilde{c}$ and $0 \leq \tilde{\kappa} < 1$, such that  (\ref{vdfbdfbdf}) holds.
Combining (\ref{vdfbdfbdf}) and (\ref{bdsfbdfg}) yields $f(N) \leq     c\gamma^{|B|}$,  with $c :=  \frac{1}{2\sqrt{2\log(2)}}\tilde{c}$, and $\gamma :=  \tilde{\kappa}^{1/2}$. Since we moreover assume that $\sigma:=\mathbb{E}^{|A|}(L)$, and $F^\dag F =\mathbb{E}^{*|C|}(R)$ are full rank operators, we can conclude that the conditions of Proposition \ref{PropMain}  are satisfied, and thus $\{A_x\}_{x=0}^{d-1}$ satisfies the purity condition in Definition \ref{DefPur}.
\end{proof}

We can conclude from Theorems \ref{hndghngdh} and  \ref{TheoremSecondPart} that purity in all essence is a necessary and sufficient condition for the exponential decay of $I_{\Phi_B(\Psi)}(A:C|B)$, and thus the typical post-measurement state $\Psi(x_B)$ loosely speaking approaches a pure product state. Concerning the classical CMI $I_{p_{\Psi}}(A:C|B)$, purity is only stated as a sufficient condition for exponential decay. 
 In relation to the question whether purity  also is necessary, one can note that for a normalized $|\psi\rangle\in\mathcal{H}_A\otimes\mathcal{H}_B\otimes\mathcal{H}_C$, it is the case that 
 \begin{equation}
 I_{p_{\psi}}(A:C|B) = I_{\Phi_{\Lambda}(\psi)}(A:C|B) = 0
 \end{equation}
 if and only if $\Phi_B(|\psi\rangle\langle\psi|)$ can be written
 \begin{equation}
\begin{split}
\Phi_B(|\psi\rangle\langle\psi|) = & \sum_{x_B}p^{B}(x_B)|\chi_{x_B}\rangle\langle \chi_{x_B}|,\\
|\chi_{x_B}\rangle = & \sum_{x_A,x_C}e^{i\theta(x_A,x_B,x_C)}\sqrt{p^A_{x_B}(x_A)}\sqrt{p^{C}_{x_B}(x_C)}|x_A\rangle|x_C\rangle,
\end{split}
\end{equation}
where $\{p^{B}(x_B)\}$ is a probability distribution, and for each $x_B$ it is the case that $\{p^A_{x_B}(x_A)\}_{x_A}$ and $\{p^C_{x_B}(x_C)\}_{x_C}$ are probability distributions, and $\theta(x_A,x_B,x_C)\in\mathbb{R}$. One can also realize that the states $|\chi_{x_B}\rangle$ in some sense are typically not product states, because of the arbitrary phase factors $e^{i\theta(x_A,x_B,x_C)}$,  and thus $I_{\Phi_B(\psi)}(A:C|B) \neq 0$. This appears to leave room for the possibility that there may exist cases of exponential decay of $I_{p_{\psi}}(A:C|B)$, even though $I_{\Phi_B(\psi)}(A:C|B)$ does not decay, and that there thus may exist a weaker condition than purity for the exponential decay of $I_{p_{\psi}}(A:C|B)$. However, if this indeed can happen for MPSs with a fixed finite-dimensional virtual space, is an open question which we leave for future investigations.

\section{\label{SecTechnicalReview} Notions from probability theory}

As mentioned in the main text, and in the proof overview, the proof of Proposition \ref{PropMain} relies  on various probabilistic concepts. Here, we briefly review the pertinent notions, and also collect the technical results that we will need at various points along the proof. Throughout these derivations we will use bold letters, such as $\boldsymbol{x}$, $\boldsymbol{Y}$, etc, to denote random variables and random operators (where `random variables'  by default are  real-valued measurable functions on the underlying probability space, while `random operators' are operator-valued measurable functions). In the following $E(\boldsymbol{x})$, $E(\boldsymbol{y})$, etc, denote the expectation value, and $E(\boldsymbol{y}|\boldsymbol{x})$ denotes the expectation value of $\boldsymbol{y}$ conditioned on $\boldsymbol{x}$. One should keep in mind that $E(\boldsymbol{y}|\boldsymbol{x})$ is a random variable (due to $\boldsymbol{x}$). One should also keep in mind the general relation $E\big( E(\boldsymbol{y}|\boldsymbol{x})\big) = E(\boldsymbol{y})$.

\subsection{Almost surely}
When we say that a relation for one, or several, random variables holds  almost surely ($a.s.$), it means that the relation is true apart from a set of probability zero. Put differently, the relation is true with probability one. For example, $\boldsymbol{x} = \boldsymbol{y}\,\, a.s.$ means that $P(\{\omega\in\Omega: \boldsymbol{x}(\omega) = \boldsymbol{y}(\omega)\}) = 1$, where $\Omega$ denotes the underlying sample space, and $\omega$ an element of the sample space.

As examples, one can consider various notions that intuitively remain true even if  `a few' points are excluded.  For example, if $\boldsymbol{x}$ and $\boldsymbol{y}$ are such that $\boldsymbol{x}\leq \boldsymbol{y}$ then (if the expectations exist) $E(\boldsymbol{x})\leq E(\boldsymbol{y})$. This conclusion still holds, even if the inequality only holds almost everywhere (see e.g.~Theorem 4.4 in chapter 2 of Ref.~\onlinecite{Gut}).
\begin{lemma}
\label{htdhgtmdg}
If $\boldsymbol{x}$ and $\boldsymbol{y}$ are  non-negative random variables, then
$\boldsymbol{x}\leq \boldsymbol{y}\,\, a.s$ implies  $E(\boldsymbol{y})\leq E(\boldsymbol{x})$.
\end{lemma}
Another statement in a similar spirit is the following. (The claim of the lemma is contained in Theorem 4.4 in chapter 2 of Ref.~\onlinecite{Gut}.) 
\begin{lemma}
\label{ztuztut}
If $\boldsymbol{x}$ is a  non-negative random variable, then  $\boldsymbol{x} = 0\,\, a.s.$ if and only if $E(\boldsymbol{x})= 0$.
\end{lemma}
For a random variable $\boldsymbol{x}$, we define the positive and negative components by $\boldsymbol{x}^{+} :=  \max(\boldsymbol{x},0)$, $\boldsymbol{x}^{-} :=  \max(-\boldsymbol{x},0)$.
By this construction, it is the case that $\boldsymbol{x}^{+}\geq  0$, $\boldsymbol{x}^{-}\geq  0$, $\boldsymbol{x} =  \boldsymbol{x}^{+}-\boldsymbol{x}^{-}$.
The expectation value  $E(\boldsymbol{x})$ of a random variable, $\boldsymbol{x}$, is defined as $E(\boldsymbol{x}) = E(\boldsymbol{x}^{+})- E(\boldsymbol{x}^{-})$ if at least one of $E(\boldsymbol{x}^{+})$ and $E(\boldsymbol{x}^{-})$ is finite.
\begin{lemma}
\label{ljhvxrushkl}
If $\boldsymbol{x}$ is a  random variable such that $\boldsymbol{x}\geq  0\,\, a.s.$ and $E(\boldsymbol{x}) =  0$, then $\boldsymbol{x} = 0\,\, a.s.$
\end{lemma}
\begin{proof}
Since $\boldsymbol{x}\geq 0$ almost surely, we can conclude that $\boldsymbol{x}^{-} = 0$ almost surely.
Since $\boldsymbol{x}^{-}$ by construction is a non-negative random variable, it follows by Lemma \ref{ztuztut} that  
 $E(\boldsymbol{x}^{-}) = 0$.
We can conclude that $E(\boldsymbol{x})$ is well defined, and $E(\boldsymbol{x}) = E(\boldsymbol{x}^{+}) - E(\boldsymbol{x}^{-})$ and thus $E(\boldsymbol{x}) = 0$ implies
 $E(\boldsymbol{x}^{+}) = E(\boldsymbol{x}^{-}) =0$.
Since $\boldsymbol{x}^{+}$ by construction is non-negative, Lemma \ref{ztuztut} implies $\boldsymbol{x}^{+}= 0\,\, a.s.$
We can thus conclude that $\boldsymbol{x} = \boldsymbol{x}^{+}-\boldsymbol{x}^{-} = 0\,\, a.s.$
\end{proof}

\subsection{Stochastic convergence of real-valued sequences}

A sequence of  random variables $(\boldsymbol{x}_N)_{N\in\mathbb{N}}$ is said to converge almost surely to a random variable $\boldsymbol{x}_{\infty}$ (denoted $\lim_{N\rightarrow\infty}\boldsymbol{x}_N = \boldsymbol{x}_{\infty}\,\,a.s.$)  if
\begin{equation}
\label{fdbsfbsf}
P(\{\omega\in\Omega:\lim_{N\rightarrow\infty}\boldsymbol{x}_N(\omega) = \boldsymbol{x}_{\infty}(\omega)\}) = 1.
\end{equation}
As mentioned above, $\omega$ is an element of the underlying sample space, $\Omega$, and $\boldsymbol{x}_N(\omega)$ is a specific realization of the stochastic process. (If one thinks of an infinite sequence of coin-tosses, then $\boldsymbol{x}_N$ vaguely stands for all possible sequences of coin-tosses, while $\boldsymbol{x}_N(\omega)$ means a specific sequence of heads and tails.)
What (\ref{fdbsfbsf}) essentially says is that if we look at the set of all sequences $\boldsymbol{x}_N(\omega)$ and $\boldsymbol{x}_{\infty}(\omega)$, such that $\boldsymbol{x}_N(\omega)$ actually do converge to $\boldsymbol{x}_{\infty}(\omega)$, then this set has probability $1$.

In these derivations, we will often start with a process that converges $\lim_{N\rightarrow\infty}\boldsymbol{x}_N = \boldsymbol{x}_{\infty}$ almost surely, but we want to show that $\lim_{N\rightarrow\infty}E(\boldsymbol{x}_N) = E(\boldsymbol{x}_{\infty})$. 
This is not generally true, but as a consequence of Lebesgues dominated convergence theorem (see e.g., Theorem 5.3 in chapter 2 of Ref.~\onlinecite{Gut}) we have the following.
\begin{prop}
\label{Lebesgues}
Suppose that $\boldsymbol{x}_N$, $\boldsymbol{x}_{\infty}$ and $\boldsymbol{y}$ are random variables such that 
$|\boldsymbol{x}_N|\leq \boldsymbol{y}$ for all $N$, where $E(\boldsymbol{y}) < +\infty$, and that $\boldsymbol{x}_N\rightarrow \boldsymbol{x}_{\infty}\,\, a.s.$, then 
\begin{equation}
\lim_{N\rightarrow\infty}E(\boldsymbol{x}_N)=  E(\boldsymbol{x}_{\infty}).
\end{equation}
\end{prop}
In the special case that $\boldsymbol{y}$ is equal to a constant $C$, one obtains the following special case (sometimes referred to as the bounded convergence theorem).
\begin{prop}
\label{fgsbsfgnfg}
Suppose that $\boldsymbol{x}_N$, $\boldsymbol{x}_{\infty}$ are random variables, and there exists a constant $C<+\infty$, such that $|\boldsymbol{x}_N|\leq C$ for all $N$. If $\boldsymbol{x}_N\rightarrow \boldsymbol{x}_{\infty}\,\, a.s.$, then
\begin{equation}
\lim_{N\rightarrow\infty}E(\boldsymbol{x}_N) = E(\boldsymbol{x}_{\infty}).
\end{equation}
\end{prop}

The following lemma is a consequence of the Borel-Cantelli Lemma (and is included in Theorem 3.1 in chapter 5 of Ref.~\onlinecite{Gut}).
\begin{lemma}
 \label{CompleteConvergence}
 Let $(\boldsymbol{x}_n)_{n\in\mathbb{N}}$ and $\boldsymbol{x}$ be random variables such that $\sum_{n=1}^{\infty}P(|\boldsymbol{x}_n-\boldsymbol{x}|>\epsilon) < +\infty$ for all  $\epsilon >0$ (sometimes referred to as complete convergence of $(\boldsymbol{x}_n)_{n\in\mathbb{N}}$ to $\boldsymbol{x}$). Then, $(\boldsymbol{x}_n)_{n\in\mathbb{N}}$ converges almost surely to $\boldsymbol{x}$.
\end{lemma}
The above can be used to obtain the following.
\begin{lemma}
\label{dfknbfdknl}
Let $(\boldsymbol{r}_N)_{N\in\mathbb{N}}$ be a sequence of  random variables, such that $\boldsymbol{r}_N\geq 0$, and such that the expectations values $E(\boldsymbol{r}_N)$ exist and are finite. Suppose that there exists a number $R$ such that $\lim_{k\rightarrow\infty}E\Big(\sum_{N=1}^k\boldsymbol{r}_N\Big) = R < +\infty$, then $\lim_{N\rightarrow\infty}\boldsymbol{r}_N = 0\,\,\, a.s.$
\end{lemma}

\begin{proof}
We first want to note that if $(a_N)_{N\in\mathbb{N}}$ is a sequence of real numbers such that $a_N\geq 0$, and if $A_k := \sum_{N=1}^ka_N$ is such that $\lim_{k\rightarrow\infty}A_k = R <+\infty$, then we have $\lim_{N\rightarrow\infty}a_N = 0$.
With $a_N :=  E(\boldsymbol{r}_N)$, and $A_k :=\sum_{N=1}^{k}a_N = E\Big(\sum_{N=1}^{k}\boldsymbol{r}_N\Big)$,  it thus follows, by the assumptions of the lemma, that  $\lim_{N\rightarrow\infty}E(\boldsymbol{r}_N)= \lim_{N\rightarrow\infty}a_N = 0$.
By assumption, $\boldsymbol{r}_N\geq 0$ and $E(\boldsymbol{r}_N)$ are well defined and finite. Hence, by Markov's inequality, it follows that $P(\boldsymbol{r}_N>\epsilon) \leq E(\boldsymbol{r}_N)/\epsilon$ for all $\epsilon>0$.
Consequently, $\sum_{N=1}^kP(\boldsymbol{r}_N>\epsilon) \leq   1/\epsilon E\Big(\sum_{N=1}^k\boldsymbol{r}_N\Big)$, 
and thus 
\begin{equation}
\begin{split}
\sum_{N=1}^{\infty}P(|\boldsymbol{r}_N|>\epsilon) \leq   & \frac{1}{\epsilon} \lim_{N\rightarrow\infty}E\left(\sum_{N=1}^k\boldsymbol{r}_N\right)
=  \dfrac{R}{\epsilon} <+\infty,\quad \forall \epsilon>0.
\end{split}
\end{equation}
Hence, $(\boldsymbol{r}_N)_{N\in\mathbb{N}}$ converges completely to $0$. By Lemma \ref{CompleteConvergence}, we can conclude that 
$(\boldsymbol{r}_N)_{N\in\mathbb{N}}$ converges almost surely to $0$.
\end{proof}

\subsection{Stochastic convergence of operator-valued sequences}

Convergence of various sequences of operators play an important role in this investigation. Since we here exclusively will deal with finite-dimensional spaces, one may argue that the distinction between `random variables' and `random operators' is not very dramatic. For the sake of clarity, we will nevertheless throughout these derivations make a distinction of the these two types and, to further this, we will use small bold letters, $\boldsymbol{x}$, $\boldsymbol{y}$, etc to denote random variables, while capital bold letters $\boldsymbol{X}$, $\boldsymbol{Y}$, etc denote random operators.

Here, we briefly recall that, on finite-dimensional Hilbert spaces, all norms are metrically equivalent. Due to the metrical equivalence in finite dimensions (see e.g., Corollary 5.4.5 in Ref.~\onlinecite{HornJohnson}), we do not need to make a distinction between different norms when we discuss convergences of sequences of operators, and we can equivalently consider the element-wise convergence of the elements of the matrix-representation in some arbitrary basis.  In what follows we will switch between these equivalent manifestations of convergence without any further comments. For the operator norms, we will mainly be using the supremum norm, $\Vert O\Vert := \sup_{\Vert\psi\Vert =1}\Vert O|\psi\rangle\Vert$, and the trace-norm, $\Vert O\Vert_{1} := \Tr\sqrt{O^{\dagger}O}$, but also the Hilbert-Schmidt norm, $\Vert O\Vert_2 := \sqrt{\Tr(O^{\dagger}O)}$.

Let us now consider a sequence of random operators $(\boldsymbol{X}_N)_{N\in\mathbb{N}}$ and $\boldsymbol{X}_{\infty}$ on a finite-dimensional Hilbert space.  We interpret the convergence $\boldsymbol{X}_N\rightarrow\boldsymbol{X}_{\infty}\,\, a.s.$ as 
\begin{equation}
\lim_{N\rightarrow\infty}\Vert \boldsymbol{X}_N-\boldsymbol{X}_{\infty}\Vert = 0\quad a.s.,
\end{equation}
or equivalently for any other operator norm (since the underlying Hilbert space is finite-dimensional),  
or as
\begin{equation}
\lim_{N\rightarrow\infty}\langle k|\boldsymbol{X}_N|k'\rangle = \langle k|\boldsymbol{X}_{\infty}|k'\rangle\quad a.s.\quad \forall  k,k' = 1,\ldots, D. 
\end{equation}

The following is a counterpart of Proposition \ref{fgsbsfgnfg}, which can be obtained by applying  Proposition \ref{fgsbsfgnfg} to the real and imaginary matrix components with respect to a basis, i.e., $\Real\langle k|\boldsymbol{X}_N|k'\rangle$ and $\Imag\langle k|\boldsymbol{X}_N|k'\rangle$.

\begin{prop}
\label{dafbadfb}
Suppose that $\boldsymbol{X}_N$ and $\boldsymbol{X}_{\infty}$ are random operators on a complex Hilbert space with finite dimension, and that there exists a constant $C<+\infty$, such that $\Vert \boldsymbol{X}_N\Vert\leq C$ for all $N$. If $\boldsymbol{X}_N\rightarrow \boldsymbol{X}_{\infty}\,\, a.s.$, then
\begin{equation}
\lim_{N\rightarrow\infty}E(\boldsymbol{X}_N) = E(\boldsymbol{X}_{\infty}).
\end{equation}
\end{prop}

\subsection{Martingales}

Our primary interest in martingales is that they allow for statements concerning the stochastic convergence of sequences of random variables.  However, in order to connect to the manner that these convergence-theorems typically are phrased in the literature, we need to briefly discuss some technical concepts. (For a more thorough introduction, see, e.g., chapter 10 in Ref.~\onlinecite{Gut}.)

Consider a sequence of random variables $(\boldsymbol{y}_N)_{N\in\mathbb{N}}$ on a probability space $(\Omega,\mathcal{F},P)$, where  $\Omega$ is the sample space, $\mathcal{F}$ is a $\sigma$-algebra (the event space), and $P$ a probability measure.  We also consider a filtration, i.e.,  a non-decreasing sequence of $\sigma$-subalgebras $\mathcal{F}_1\subset \mathcal{F}_{2}\subset\cdots \subset\mathcal{F}$.  A sequence $(\boldsymbol{y}_N)_{N\in\mathbb{N}}$ of random variables is said to be adapted to $(\mathcal{F}_N)_{N\in\mathbb{N}}$ if each $\boldsymbol{y}_N$ is measurable with respect to $\mathcal{F}_N$. A sequence $(\boldsymbol{y}_N)_{N\in\mathbb{N}}$ is a martingale with respect to $(\mathcal{F}_N)_{N\in\mathbb{N}}$  if  $(\boldsymbol{y}_N)_{N\in\mathbb{N}}$ is adapted to $(\mathcal{F}_N)_{N\in\mathbb{N}}$, satisfies $E(\boldsymbol{y}_{N+1}|\mathcal{F}_N) = \boldsymbol{y}_N\,\, a.s.$, as well as $E(|\boldsymbol{y}_N|) < +\infty$.  Intuitively, $\mathcal{F}_N$ stands for the information available to us at step $N$. In our setting, this information corresponds to variables $\boldsymbol{x}_1,\ldots,\boldsymbol{x}_N$ (which are assumed to also be random variables on the same underlying probability space  $(\Omega,\mathcal{F},P)$). More precisely,  $\mathcal{F}_N := \sigma(\boldsymbol{x}_N,\ldots,\boldsymbol{x}_1)$, which denotes the $\sigma$-algebra generated by  $\boldsymbol{x}_1,\ldots,\boldsymbol{x}_N$ and often is referred to as the natural filtration of $\boldsymbol{x}_1,\ldots,\boldsymbol{x}_N$. Since in the following we exclusively will use the natural filtrations, we will employ the more succinct notation 
$E(\boldsymbol{y}_{N+1}|\boldsymbol{x}_N,\ldots,\boldsymbol{x}_1) := E(\boldsymbol{y}_{N+1}|\mathcal{F}_N)$, with $\mathcal{F}_N := \sigma(\boldsymbol{x}_1,\ldots,\boldsymbol{x}_N)$. We moreover say that  $(\boldsymbol{y}_N)_{N\in\mathbb{N}}$ is a martingale with respect to $(\boldsymbol{x}_N)_{N\in\mathbb{N}}$ if
\begin{equation}
\label{InformalMartingale}
\begin{split}
    E(\boldsymbol{y}_{N+1}|\boldsymbol{x}_N,\ldots,\boldsymbol{x}_1) = & \boldsymbol{y}_N\quad a.s., \quad\text{and}\quad E(|\boldsymbol{y}_N|) < +\infty,\\
    \text{with}\quad\boldsymbol{y}_N = & f_N(\boldsymbol{x}_N,\ldots,\boldsymbol{x}_1),
\end{split}
\end{equation}
for (Borel measurable) functions $f_N$. The construction with the functions $f_N$ guarantees that $(\boldsymbol{y}_N)_{N\in\mathbb{N}}$ is adapted to the natural filtration of $(\boldsymbol{x}_N)_{N\in\mathbb{N}}$. The following proposition is obtained as a special case of Theorem 12.1 in chapter 10 of Ref.~\onlinecite{Gut}.

\begin{prop}
\label{bsdgnsfgn}
Let $(\boldsymbol{y}_N)_{N\in\mathbb{N}}$ be a martingale with respect to another process $(\boldsymbol{x}_N)_{N\in\mathbb{N}}$, and suppose that there exists a real number $C$ such that 
$| \boldsymbol{y}_N| \leq C$ for all $N\in\mathbb{N}$,
then there exists a random variable $\boldsymbol{y}_{\infty}$ such that 
\begin{equation}
\lim_{N\rightarrow\infty}\boldsymbol{y}_N = \boldsymbol{y}_{\infty}\quad a.s.\quad\textrm{and}\quad
E(|\boldsymbol{y}_{\infty}|)< +\infty.
\end{equation}
\end{prop}
As a technical remark concerning the relation to Theorem 12.1 in chapter 10 of Ref.~\onlinecite{Gut}, one may note that the condition $| \boldsymbol{y}_N| \leq C$ implies that $(\boldsymbol{y}_n)_{n\in\mathbb{N}}$ is uniformly integrable.

Our main interest is not these `standard' real-valued martingales, but rather operator-valued martingales. It is again worth recalling that we here only consider finite-dimensional spaces, and hence we can represent each operator as a finite matrix with respect to some choice of basis. With this in mind, we say that an operator-valued process $(\boldsymbol{Y}_N)_{N\in\mathbb{N}}$ on a finite-dimensional Hilbert space is an operator-valued martingale with respect to a stochastic process $(\boldsymbol{x}_N)_{N\in\mathbb{N}}$ if each of $\Real\langle k|\boldsymbol{Y}_N|k'\rangle$ and $\Imag\langle k|\boldsymbol{Y}_N|k'\rangle$ are a martingale  with respect to $(\boldsymbol{x}_N)_{N\in\mathbb{N}}$ for some fixed orthonormal basis $\{|k\rangle\}_{k=1}^{D}$.
One may note that the condition that $E(|\Real\langle k|\boldsymbol{Y}_N|k'\rangle|)< +\infty$ and $E(|\Imag\langle k|\boldsymbol{Y}_N|k'\rangle|)< +\infty$ in the finite-dimensional case is equivalent to
\begin{equation}
E(\Vert\boldsymbol{Y}_N\Vert) < +\infty.
\end{equation}
Similarly, the conditions 
 \begin{equation}
 \begin{split}
& E\big(\Real\langle k|\boldsymbol{Y}_{N+1}|k'\rangle\big|\boldsymbol{x}_N,\ldots,\boldsymbol{x}_1\big) = \Real\langle k|\boldsymbol{Y}_{N+1}|k'\rangle\hspace{0.2cm}a.s.,\\
 & E\big(\Imag\langle k|\boldsymbol{Y}_{N+1}|k'\rangle\big|\boldsymbol{x}_N,\ldots,\boldsymbol{x}_1\big) = \Imag\langle k|\boldsymbol{Y}_{n+1}|k'\rangle\hspace{0.2cm}a.s.,
 \end{split}
 \end{equation}
  can equivalently be stated as
 \begin{equation}
E(\boldsymbol{Y}_{N+1}|\boldsymbol{x}_N,\ldots,\boldsymbol{x}_1) = \boldsymbol{Y}_N\hspace{0.2cm}a.s.
\end{equation}
In a similar manner, Proposition \ref{bsdgnsfgn} can be applied to the real and imaginary components of an operator-valued martingale, which yields the following `operator counterpart' to Proposition \ref{bsdgnsfgn}.
\begin{prop}
\label{MatrixMartingaleConvergence}
Let $(\boldsymbol{Y}_N)_{N\in\mathbb{N}}$ be  an operator-valued martingale on a finite-dimensional complex Hilbert space with respect to a real-valued process $(\boldsymbol{x}_N)_{N\in\mathbb{N}}$, and suppose that there exists real number $C$ such that 
$\Vert \boldsymbol{Y}_N\Vert \leq C$ for all  $N\in\mathbb{N}$,
then there exists a random operator, $\boldsymbol{Y}_{\infty}$, such that 
\begin{equation}
\lim_{N\rightarrow\infty}\boldsymbol{Y}_N = \boldsymbol{Y}_{\infty}\quad a.s.\quad\textrm{and}\quad
E(\Vert \boldsymbol{Y}_{\infty}\Vert)< +\infty.
\end{equation}
\end{prop}

\section{\label{TheProcess}Stochastic process of measurements}

Consider a set of operators  $\{A_x\}_{x=0}^{d-1}$ on a finite-dimensional Hilbert space, $\mathcal{H}$, with dimension $D := \dim\mathcal{H}$, such that $\sum_{x=0}^{d-1}A_x^{\dagger}A_x =\1$. We introduce the stochastic process $(\boldsymbol{x}_N)_{N\in\mathbb{N}}$ with a joint distribution such that, for each $N$, the marginal distribution of $\boldsymbol{x}_1,\ldots,\boldsymbol{x}_N$ is given by
\begin{equation}
\label{dgnfsgn}
\begin{split}
& P(\boldsymbol{x}_N = x_N,\ldots,\boldsymbol{x}_1 = x_1) = \frac{1}{D}\Tr\left(A_{x_N}\cdots A_{x_1}A_{x_1}^{\dagger}\cdots A_{x_N}^{\dagger}\right).
\end{split}
\end{equation}
This means that $\boldsymbol{x}_1,\ldots,\boldsymbol{x}_N$ can be interpreted as the outcomes of a sequence of measurements, where the initial state is maximally mixed, i.e., $\sigma = \1/D$.
Note that when we in the following refer to an expectation value,  the underlying probability distribution is assumed to be (\ref{dgnfsgn}) unless otherwise stated.
It will be useful to note that
\begin{equation}
\label{sfgnsfgnmsd}
\begin{split}
 P(\boldsymbol{x}_{N+1} = x_{N+1}|\boldsymbol{x}_N = x_N,\ldots,\boldsymbol{x}_1 = x_1) 
= & 
\frac{\Tr( A_{x_1}^{\dagger}\cdots A_{x_N}^{\dagger}A_{x_{N+1}}^{\dagger} A_{x_{N+1}}A_{x_N}\cdots A_{x_1} )}{\Tr( A_{x_1}^{\dagger}\cdots A_{x_N}^{\dagger}A_{x_N}\cdots A_{x_1} )},
\end{split}
\end{equation}
where $P(y|x):=P(y,x)/P(x)$ denotes the conditional probability.

Based on the process  $(\boldsymbol{x}_N)_{N\in\mathbb{N}}$, we define sequences of random operators
\begin{equation}
\label{gndghnghn}
\begin{split}
\boldsymbol{A}_N := & A_{\boldsymbol{x}_N},\\
\boldsymbol{W}_N := & \boldsymbol{A}_N \cdots \boldsymbol{A}_1 = A_{\boldsymbol{x}_N}\cdots A_{\boldsymbol{x}_1},\\
\boldsymbol{M}_N :=& \left\{ \begin{matrix} \frac{\boldsymbol{W}_N^{\dagger}\boldsymbol{W}_N}{\Tr(\boldsymbol{W}_N^{\dagger}\boldsymbol{W}_N)} & \textrm{if} &  \Tr(\boldsymbol{W}_N^{\dagger}\boldsymbol{W}_N) \neq 0,\\
0  & \textrm{if} &  \Tr(\boldsymbol{W}_N^{\dagger}\boldsymbol{W}_N) = 0.
\end{matrix}\right.
\end{split}
\end{equation}
One should note that $\boldsymbol{M}_N$ is Hermitian, i.e., $\boldsymbol{M}_N^{\dagger} =\boldsymbol{M}_N$. Moreover, $\boldsymbol{M}_N$ 
is positive semi-definite, and either has trace $1$ or trace $0$. Hence, $\boldsymbol{M}_N$ is either a density operator, or the zero operator. One may further note that 
\begin{equation}
\Tr(\boldsymbol{W}_N^{\dagger}\boldsymbol{W}_N)=\Tr( A^\dag_{\boldsymbol{x}_1} \cdots A_{\boldsymbol{x}_N}^{\dagger} A_{\boldsymbol{x}_N}\cdots A_{\boldsymbol{x}_1})= P(\boldsymbol{x}_N,\ldots,\boldsymbol{x}_1)D,
\end{equation}
from which we can conclude that $\boldsymbol{M}_N =0$ with probability zero, i.e., $\boldsymbol{M}_N$ is almost surely a density operator.

As a side-remark, one might note that $\boldsymbol{M}_N$ is \emph{not} the post-measurement state of the measurement process. The post-measurement state  would rather be $\boldsymbol{\rho}_N := \boldsymbol{W}_N\boldsymbol{W}_N^{\dagger}/\Tr(\boldsymbol{W}_N^{\dagger}\boldsymbol{W}_N)$. However, $\boldsymbol{M}_N$ and $\boldsymbol{\rho}_N$ have the same non-zero eigenvalues (as can be seen by a singular value decomposition of $\boldsymbol{W}_N$). The main reason for why it is convenient to use $\boldsymbol{M}_N$, rather than $\boldsymbol{\rho}_N$, is that on $\boldsymbol{M}_N$ we can directly utilize $\sum_x A_x^{\dagger}A_x = \1$, which for example is used in the proof of the martingale property in Lemma \ref{PropMartingale}.

In the following it will be useful to observe that since $\boldsymbol{A}_N$ is a (deterministic) function of $\boldsymbol{x}_N$ [as seen by (\ref{gndghnghn})] it is the case that
\begin{equation}
E(\boldsymbol{A}_N|\boldsymbol{x}_N) = A_{\boldsymbol{x}_N} = \boldsymbol{A}_N.
\end{equation}
Analogously,  $E(\boldsymbol{W}_N|\boldsymbol{x}_N,\ldots,\boldsymbol{x}_1) =  \boldsymbol{W}_N$, 
and similarly $E(\boldsymbol{M}_N|\boldsymbol{x}_N,\ldots,\boldsymbol{x}_1) =  \boldsymbol{M}_N$.

\section{\label{AppProofPropMain}Elements of the proof of Proposition \ref{PropMain}}

\subsection{\label{SecLimit} $\lim_{N\rightarrow \infty} \boldsymbol{M}_N  = \boldsymbol{M}_{\infty}\,\, a.s.$}

The purpose of this section is to show that $\boldsymbol{M}_N$ has limit operator $\boldsymbol{M}_{\infty}$ in a sufficiently strong sense, and that this limit operator has `nice' properties. We do this by first showing that $\boldsymbol{M}_N$ is a martingale relative to the sequence of measurement outcomes $\boldsymbol{x}_N$. This in turn yields almost sure convergence to limiting operator $\boldsymbol{M}_{\infty}$. Recall that the underlying probability distribution is assumed to be (\ref{dgnfsgn}), and that all expectations are taken with respect to this distribution.

\begin{lemma}
\label{PropMartingale}
Let $\{A_x\}_{x=0}^{d-1}$ be linear operators on a  finite-dimensional complex Hilbert space, such that $\sum_{x=0}^{d-1}A_x^{\dagger}A_x = \1$. Then, $(\boldsymbol{M}_N)_{N\in\mathbb{N}}$, defined by (\ref{gndghnghn}), is an operator-valued martingale  with respect to $(\boldsymbol{x}_N)_{N\in\mathbb{N}}$ with distribution (\ref{dgnfsgn}).
\end{lemma}

\begin{proof}
From the fact that each $\boldsymbol{M}_N$ is a density operator, or the zero operator, it follows that $\Vert \boldsymbol{M}_N\Vert \leq 1$, and thus in particular that $E(\Vert \boldsymbol{M}_N\Vert) \leq 1 < +\infty$. Moreover, by the construction in (\ref{gndghnghn}), it is the case that $\boldsymbol{M}_N$ (and thus the matrix-elements with respect to a given basis) are functions of $\boldsymbol{x}_N,\ldots,\boldsymbol{x}_1$. By (\ref{sfgnsfgnmsd}) and $\sum_{x_{N+1}=0}^{d-1}A_{x_{N+1}}^{\dagger}A_{x_{N+1}} = \1$ we find that
\begin{equation*}
\begin{split}
E(\boldsymbol{M}_{N+1}|\boldsymbol{x}_{N}& = x_N,\ldots,\boldsymbol{x}_1=x_1) \\
& = \sum_{x_{N+1}} \frac{A_{x_1}^{\dagger}\cdots A_{x_N}^{\dagger}A_{x_{N+1}}^{\dagger}A_{x_{N+1}}A_{x_N}\cdots A_{x_1} }{\Tr(A_{x_1}^{\dagger}\cdots A_{x_N}^{\dagger}A_{x_{N+1}}^{\dagger}A_{x_{N+1}}A_{x_N}\cdots A_{x_1} )}P(  \boldsymbol{x}_{N+1} = x_{N+1} | \boldsymbol{x}_{N} = x_N,\ldots,\boldsymbol{x}_1=x_1),\\
& = E(\boldsymbol{M}_{N}|
\boldsymbol{x}_{N} = x_N,\ldots,\boldsymbol{x}_1=x_1).
\end{split}
\end{equation*}
We can conclude that $E(\boldsymbol{M}_{N+1}|\boldsymbol{x}_{N},\ldots,\boldsymbol{x}_1) 
=  E(\boldsymbol{M}_{N}|
\boldsymbol{x}_{N},\ldots,\boldsymbol{x}_1) =  \boldsymbol{M}_{N}$. 
Hence, $(\boldsymbol{M}_N)_N$ is a martingale sequence with respect to $(\boldsymbol{x}_N)_{N\in\mathbb{N}}$. 
\end{proof}

\begin{lemma}
\label{fnjjnfddanj}
Let $\{A_x\}_{x=0}^{d-1}$ be linear operators on a finite-dimensional complex Hilbert space, such that $\sum_{x=0}^{d-1}A_x^{\dagger}A_x = \1$. Let $(\boldsymbol{M}_N)_{N\in\mathbb{N}}$ be as defined in (\ref{gndghnghn}) with respect to $(\boldsymbol{x}_N)_{N\in\mathbb{N}}$ and distributed as in (\ref{dgnfsgn}). Then, there exists a random operator, $\boldsymbol{M}_{\infty}$, such that
\begin{eqnarray}
 \label{bfgnfg}
& & \lim_{N\rightarrow \infty} \boldsymbol{M}_N  = \boldsymbol{M}_{\infty}\quad a.s.,\\
\nonumber & & \\
& & \textrm{$\boldsymbol{M}_{\infty}$ is almost surely a density operator,}\\
\nonumber & & \\
 \label{nghngfhm}
& & \lim_{N\rightarrow\infty}E(\boldsymbol{M}_N) = E(\boldsymbol{M}_{\infty}),\\
\nonumber & & \\
\label{dghndghnh}
& & \lim_{N\rightarrow\infty} E(\Vert\boldsymbol{M}_N\Vert) =  E(\Vert\boldsymbol{M}_{\infty}\Vert),\\ 
\nonumber & & \\
& & E(\Vert\boldsymbol{M}_{\infty}\Vert)  <+\infty.
\end{eqnarray}
\end{lemma}
\begin{proof}
By Lemma \ref{PropMartingale} we know that $(\boldsymbol{M}_N)_{N\in\mathbb{N}}$ is a martingale with respect to $(\boldsymbol{x}_N)_{N\in\mathbb{N}}$. 
From the fact that each $\boldsymbol{M}_N$ is a density operator, or the zero operator, it follows that 
\begin{equation}
\label{fghnsfgh}
\Vert \boldsymbol{M}_N\Vert \leq 1,\quad \forall N\in\mathbb{N}.
\end{equation} 
By Proposition \ref{MatrixMartingaleConvergence} it follows that there exists a random operator, $\boldsymbol{M}_{\infty}$, such that
\begin{equation}
\label{fgndghm}
\lim_{N\rightarrow\infty}\boldsymbol{M}_N = \boldsymbol{M}_{\infty}\quad a.s.,
\end{equation}
with $E(\Vert \boldsymbol{M}_{\infty}\Vert) <+\infty$. 
By combining (\ref{fgndghm}) with (\ref{fghnsfgh}), Proposition \ref{dafbadfb} yields $E(\boldsymbol{M}_N)\rightarrow E(\boldsymbol{M}_{\infty})$.
Moreover,  (\ref{fgndghm}) yields $\lim_{N\rightarrow \infty} \Vert\boldsymbol{M}_N\Vert  = \Vert\boldsymbol{M}_{\infty}\Vert\,\, a.s.$ 
By this observation together with (\ref{fghnsfgh}), Proposition \ref{fgsbsfgnfg} with $\boldsymbol{x}_N :=  \Vert\boldsymbol{M}_N\Vert$ and $\boldsymbol{x}_{\infty}:=\Vert\boldsymbol{M}_{\infty}\Vert$ yields $E(\Vert\boldsymbol{M}_N\Vert)\rightarrow E(\Vert\boldsymbol{M}_{\infty}\Vert)$. 
Finally, we should show that $\boldsymbol{M}_{\infty}$ almost surely is a density operator, i.e., that $\boldsymbol{M}_{\infty}\geq 0$ almost surely, and that $\Tr\boldsymbol{M}_{\infty} = 1$ almost surely.  From (\ref{fgndghm}) it follows that $\lim_{N\rightarrow\infty}\langle\psi|\boldsymbol{M}_N|\psi\rangle = \langle\psi|\boldsymbol{M}_{\infty}|\psi\rangle\,\, a.s$. Since $\langle\psi|\boldsymbol{M}_N|\psi\rangle\geq 0$, it follows that $\langle\psi|\boldsymbol{M}_{\infty}|\psi\rangle\geq 0\,\, a.s$.  Analogously, since $\Tr\boldsymbol{M}_N = 1$ almost surely, it follows that $\lim_{N\rightarrow\infty}\Tr\boldsymbol{M}_N = \Tr\boldsymbol{M}_{\infty} = 1\,\, a.s$.
Hence, $\boldsymbol{M}_{\infty}$ is almost surely a density operator.
\end{proof}

\subsection{\label{SecPurImpliesRankOne} If $\{A_x\}_{x=0}^{d-1}$ satisfies the purity condition, then  $\mathrm{rank}(\boldsymbol{M}_{\infty}) = 1\,\, a.s.$ }

The purpose of this section is to show that the limit operator $\boldsymbol{M}_{\infty}$, more or less always, is a rank-one operator whenever  $\{A_x\}_{x=0}^{d-1}$ satisfies the purity condition. The first step (Lemma \ref{ghdgguk})  is to show that the difference between the operators $\boldsymbol{M}_{N+p}$ and $\boldsymbol{M}_{N}$ tends to vanish as $N$ increases, even when conditioned on all the measurement outcomes $\boldsymbol{x}_N,\ldots, \boldsymbol{x}_1$.

\begin{lemma}
\label{ghdgguk}
Let $\{A_x\}_{x=0}^{d-1}$ be linear operators on  a finite-dimensional complex Hilbert space, such that $\sum_{x=0}^{d-1}A_x^{\dagger}A_x = \1$. Let $(\boldsymbol{M}_N)_{N\in\mathbb{N}}$  be as defined in (\ref{gndghnghn}) with respect to $(\boldsymbol{x}_N)_{N\in\mathbb{N}}$ and distributed as in (\ref{dgnfsgn}). Then,
\begin{equation}
\label{dfgnsfgnfg}
\begin{split}
 & \lim_{N\rightarrow\infty} E\Big(\Vert\boldsymbol{M}_{N+p}-\boldsymbol{M}_{N}\Vert \Big|\boldsymbol{x}_N,\ldots,\boldsymbol{x}_1\Big) =0\quad a.s.
\end{split}
\end{equation}
\end{lemma}

\begin{proof}
Recall that   $\boldsymbol{M}_{N}$ is a deterministic function of $\boldsymbol{x}_{N},\ldots,\boldsymbol{x}_1$, and thus
\begin{equation}
\label{fgbfgbf}
E(\boldsymbol{M}_{N}| \boldsymbol{x}_{N},\ldots,\boldsymbol{x}_1) = \boldsymbol{M}_{N}.
\end{equation}
A direct consequence is that $\boldsymbol{M}_{N}$ and $\boldsymbol{M}_{N+p}$ are independent when conditioned on $\boldsymbol{x}_{N},\ldots,\boldsymbol{x}_1$, and thus
\begin{equation}
\label{fgnsfgns}
E(\boldsymbol{M}_{N+p}\boldsymbol{M}_{N}| \boldsymbol{x}_{N},\ldots,\boldsymbol{x}_1)= E(\boldsymbol{M}_{N+p}| \boldsymbol{x}_{N},\ldots,\boldsymbol{x}_1)E(\boldsymbol{M}_{N}| \boldsymbol{x}_{N},\ldots,\boldsymbol{x}_1).
\end{equation}

By expanding $E\big((\boldsymbol{M}_{N+p}-\boldsymbol{M}_{N})^2\big)$ one obtains cross-terms such as 
\begin{equation}
\begin{split}
 E(\boldsymbol{M}_{N+p}\boldsymbol{M}_{N})  &= E\Big(E(\boldsymbol{M}_{N+p}\boldsymbol{M}_{N}| \boldsymbol{x}_{N},\ldots,\boldsymbol{x}_1)\Big), \\
 &=  E\Big(E(\boldsymbol{M}_{N+p}| \boldsymbol{x}_{N},\ldots,\boldsymbol{x}_1    )E(\boldsymbol{M}_{N}| \boldsymbol{x}_{N},\ldots,\boldsymbol{x}_1)\Big),
\end{split}
\end{equation}
where the last equality follows by the conditional independence in (\ref{fgnsfgns}). By combining these observations with the martingale property, as shown in Lemma \ref{PropMartingale},  with (\ref{fgbfgbf}), $E\big((\boldsymbol{M}_{N+p}-\boldsymbol{M}_{N})^2\big)$ results in
\begin{equation}
\label{svbsfv}
\begin{split}
 E\big((\boldsymbol{M}_{N+p}&-\boldsymbol{M}_{N})^2\big) 
 \\
= & E(\boldsymbol{M}_{N+p}^2) + E(\boldsymbol{M}_{N}^2)-2E\Big(E(\boldsymbol{M}_{N}| \boldsymbol{x}_{N},\ldots,\boldsymbol{x}_1)E(\boldsymbol{M}_{N}| \boldsymbol{x}_{N},\ldots,\boldsymbol{x}_1    )\Big),\\
= & E(\boldsymbol{M}_{N+p}^2)  - E(\boldsymbol{M}_{N}^2),
\end{split}
\end{equation}
where we in the last step have used (\ref{fgbfgbf}).

By the observation that $\boldsymbol{M}_{N}$ is Hermitian, it follows that $\Tr E(\boldsymbol{M}_{N+p}^2)  =  E(\Vert\boldsymbol{M}_{N+p}\Vert^2_2)$, $\Tr E(\boldsymbol{M}_{N}^2) = E(\Vert\boldsymbol{M}_{N}\Vert^2_2)$ and $\Tr E\big((\boldsymbol{M}_{N+p}-\boldsymbol{M}_{N})^2\big) =    E\big(  \Vert \boldsymbol{M}_{N+p}-\boldsymbol{M}_{N}\Vert^2_2\big)$, which with (\ref{svbsfv}) yields
\begin{equation}
\label{rzmur}
E(\Vert \boldsymbol{M}_{N+p}\Vert^2_2)  - E(\Vert\boldsymbol{M}_{N}\Vert^2_2)\geq 0.
\end{equation}
By using  (\ref{svbsfv}), we next observe that
\begin{equation}
\label{dnnsmz}
\begin{split}
  \sum_{N=0}^{p-1}E(\boldsymbol{M}^2_{N+k+1})-\sum_{N=0}^{p-1}E(\boldsymbol{M}^2_{N})=& \sum_{N=0}^{k}E(\boldsymbol{M}_{N+p}^2)-\sum_{N=0}^{k}E(\boldsymbol{M}_{N}^2),\\
= & \sum_{N=0}^{k}E\big((\boldsymbol{M}_{N+p}-\boldsymbol{M}_{N})^2\big), \\
= & E\left(\sum_{N=0}^kE\Big((\boldsymbol{M}_{N+p}-\boldsymbol{M}_{N})^2\Big|\boldsymbol{x}_N,\ldots,\boldsymbol{x}_1\Big) \right),
\end{split}
\end{equation}
where we in the last step use the general relation $E\big(E(\boldsymbol{y}|\boldsymbol{x})\big) = E(\boldsymbol{y})$. Recall that $\Vert O\Vert^2_2 := \Tr(O^2)$. By applying the trace to (\ref{dnnsmz}) we obtain
\begin{equation*}
    E\left(\sum_{N=0}^kE\Big(\Vert \boldsymbol{M}_{N+p}-\boldsymbol{M}_{N}\Vert^2_2\Big|\boldsymbol{x}_N,\ldots,\boldsymbol{x}_1\Big) \right)= \sum_{N=0}^{p-1} \bigg[ E(\Vert\boldsymbol{M}_{N+k+1}\Vert^2_2)  - E(\Vert\boldsymbol{M}_{N}\Vert^2_2)\bigg].
\end{equation*}

By Lemma \ref{fnjjnfddanj} we know that $\lim_{N\rightarrow\infty}\boldsymbol{M}_N = \boldsymbol{M}_{\infty}$ almost surely. From this observation it follows that  $\lim_{N\rightarrow\infty}\Vert\boldsymbol{M}_{N+k+1}\Vert^2_2 = \Vert\boldsymbol{M}_{\infty}\Vert^2_2\,\, a.s.$ 
Next, we note that $\boldsymbol{M}_N$ is a density operator, or the zero operator, and thus it follows that $\Vert\boldsymbol{M}_N\Vert_2^2\leq 1$. With $\boldsymbol{x}_N := \Vert\boldsymbol{M}_N\Vert_2^2$  and $\boldsymbol{x}_{\infty}:= \Vert\boldsymbol{M}_{\infty}\Vert^2_2$, it follows by Proposition \ref{fgsbsfgnfg} that
\begin{equation}
\lim_{N\rightarrow\infty}E(\Vert\boldsymbol{M}_{N+k+1}\Vert^2_2) = E(\Vert\boldsymbol{M}_{\infty}\Vert^2_2).
\end{equation}
By Lemma \ref{fnjjnfddanj} we know that $\boldsymbol{M}_{\infty}$ is almost surely a density operator, from which it follows that $\Vert\boldsymbol{M}_{\infty}\Vert_2^2 \leq 1\,\, a.s.$  With $\boldsymbol{x}:= \Vert\boldsymbol{M}_{\infty}\Vert_2^2$ and $\boldsymbol{y} :=1$ in  Lemma \ref{htdhgtmdg}, we get
\begin{equation} 
\label{dfbsfbfg}
 E(\Vert\boldsymbol{M}_{\infty}\Vert_2^2) \leq 1.
 \end{equation} 
With $p:=k+1$ in the inequality (\ref{rzmur}), it follows that $E(\Vert\boldsymbol{M}_{N+k+1}\Vert^2_2)-E(\Vert\boldsymbol{M}_{N}\Vert^2_2)\geq 0$, which implies $E(\Vert\boldsymbol{M}_{\infty}\Vert^2_2)-E(\Vert\boldsymbol{M}_{N}\Vert^2_2)\geq 0$.
Hence,
\begin{equation}
\label{dgfmndghm}
\begin{split}
    \lim_{k\rightarrow\infty} E\left(\sum_{N=0}^kE\Big(\Vert\boldsymbol{M}_{N+p}-\boldsymbol{M}_{N}\Vert^2_2\Big|\boldsymbol{x}_N,\ldots,\boldsymbol{x}_1\Big) \right)&= \sum_{N=0}^{p-1}\bigg[ E(\Vert\boldsymbol{M}_{\infty}\Vert^2_2)  - E(\Vert\boldsymbol{M}_{N}\Vert^2_2)\bigg], \\
    &=:  R(p).
\end{split}
\end{equation}

By $E(\Vert\boldsymbol{M}_{\infty}\Vert^2_2)  -  E(\Vert\boldsymbol{M}_{N}\Vert^2_2)\geq 0$, it follows that $R(p)\geq 0$ and, by the inequality (\ref{dfbsfbfg}), it follows that $R(p)\leq p <+\infty$. Define  $\boldsymbol{r}_N := E(\Vert\boldsymbol{M}_{N+p}-\boldsymbol{M}_{N}\Vert^2_2|\boldsymbol{x}_N,\ldots,\boldsymbol{x}_1)$. Note that $\boldsymbol{r}_N\geq 0$. Moreover, $E(\boldsymbol{r}_N)   =  E(\Vert\boldsymbol{M}_{N+p}-\boldsymbol{M}_{N}\Vert^2_2)$. 
Since $\boldsymbol{M}_N$ is either a density operator, or the zero operator, it follows that $\Vert \boldsymbol{M}_N\Vert_2\leq 1$, and thus  $\Vert\boldsymbol{M}_{N+p}-\boldsymbol{M}_{N}\Vert^2_2  \leq  (\Vert\boldsymbol{M}_{N+p}\Vert_2+\Vert\boldsymbol{M}_{N}\Vert_2)^2\leq  4$. 
We conclude that $E(\Vert\boldsymbol{M}_{N+p}-\boldsymbol{M}_{N}\Vert^2_2)\leq4$, which together with $E(\boldsymbol{r}_N)= E(\Vert\boldsymbol{M}_{N+p}-\boldsymbol{M}_{N}\Vert^2_2)$ yields $E(\boldsymbol{r}_N)   \leq  4$. By Eq.~(\ref{dgfmndghm}), there exists a number $R(p)$ such that $\lim_{k\rightarrow\infty}E(\sum_{N=0}^{k}\boldsymbol{r}_N) = R(p) <+\infty$. All the conditions of Lemma \ref{dfknbfdknl} are thus satisfied and it yields 
\begin{equation}
\label{dgnsfgnsfg}
\begin{split}
 & \lim_{N\rightarrow\infty} E\Big(\Vert\boldsymbol{M}_{N+p}(\omega)-\boldsymbol{M}_{N}(\omega)\Vert^2_2\Big|\boldsymbol{x}_N(\omega),\ldots,\boldsymbol{x}_1(\omega)\Big) =0.
\end{split}
\end{equation}
Next, we note that $x\mapsto x^2$ is a convex function, and thus by Jensen's inequality $(E(X))^2 \leq E(X^2)$. 
By combining this observation with (\ref{dgnsfgnsfg}) we obtain
\begin{equation}
\label{bfdnsfgnsf}
\begin{split}
 & \lim_{N\rightarrow\infty}E\Big(\Vert\boldsymbol{M}_{N+p}-\boldsymbol{M}_{N}\Vert_2\Big|\boldsymbol{x}_N,\ldots,\boldsymbol{x}_1\Big) = 0\quad  a.s.
\end{split}
\end{equation}
By the general relation between the supremum norm and the Hilbert-Schmidt norm,  $\Vert R\Vert \leq \Vert R\Vert_{2}$,  we get
$E(\Vert \boldsymbol{R}\Vert) \leq E(\Vert \boldsymbol{R}\Vert_{2})$, 
and thus (\ref{bfdnsfgnsf}) yields (\ref{dfgnsfgnfg}).
\end{proof}

The following lemma provides a reformulation of  the purity condition that is better suited for the proof-technique that we employ. 
\begin{lemma} 
\label{netrswethn}
Let $\{A_x\}_{x=0}^{d-1}$ be linear operators on a complex finite-dimensional Hilbert space, $\mathcal{H}$. Then, $\{A_x\}_{x=0}^{d-1}$ satisfies the purity condition if and only if the following condition holds:
\begin{equation}
\label{ghmfhjm}
\begin{split}
& \textrm{If $O$ is an operator on $\mathcal{H}$ such that}\\
& O^{\dagger} A_{x_1}^{\dagger}\cdots A_{x_N}^{\dagger}A_{x_N}\cdots A_{x_1}O \propto O^{\dagger}O,\quad\forall N\in\mathbb{N},\quad \forall (x_1,\ldots,x_N) \in \{0,\ldots,d-1\}^{\times N},\\
& \textrm{then $\mathrm{rank}(O) = 1$}.
\end{split}
\end{equation}
\end{lemma}

\begin{proof}
We start proving the direction that, if $\{A_x\}_{x=0}^{d-1}$ satisfies condition (\ref{ghmfhjm}), then  $\{A_x\}_{x=0}^{d-1}$ also satisfies the purity condition.
Suppose that condition (\ref{ghmfhjm}) holds. For the subset of operators $O=P$ for projectors $P$, we thus find that condition (\ref{dsfbsfg}) holds, and hence $\{A_x\}_{x=0}^{d-1}$ satisfies the purity condition.

Conversely, we wish to show that, if  $\{A_x\}_{x=0}^{d-1}$ satisfies the purity condition, then $\{A_x\}_{x=0}^{d-1}$ also satisfies condition (\ref{ghmfhjm}). Hence, assume that  $\{A_x\}_{x=0}^{d-1}$ satisfies the purity condition. Let $O$ be any operator on $\mathcal{H}$ such that 
\begin{equation}
\label{bfbdfsg}
O^{\dagger} A_{x_1}^{\dagger}\cdots A_{x_N}^{\dagger}A_{x_N}\cdots A_{x_1}O \propto O^{\dagger}O
\end{equation}
for all $N$ and all $x_1,\ldots,x_N$.
We next note that $OO^{\dagger}$ is positive semi-definite, and let $(OO^{\dagger})^{\ominus}$ denote the inverse on the support of  $OO^{\dagger}$, such that $(OO^{\dagger})^{\ominus}OO^{\dagger} = OO^{\dagger}(OO^{\dagger})^{\ominus} = P$, where $P$ is the projector onto the support of $OO^{\dagger}$. Multiplying (\ref{bfbdfsg}) from the left with $(OO^{\dagger})^{\ominus}O$ and from the right with  $O^{\dagger}(OO^{\dagger})^{\ominus}$ results in $P A_{x_1}^{\dagger}\cdots A_{x_N}^{\dagger}A_{x_N}\cdots A_{x_1}P \propto P$.
Since the purity condition is assumed to hold, it follows that $\mathrm{rank}(P) = 1$. However, $\mathrm{rank}(P) = \mathrm{rank}(OO^{\dagger}) =\mathrm{rank}(O)$. We can thus conclude that if $\{A_x\}_{x=0}^{d-1}$ satisfies the purity condition, then $\{A_x\}_{x=0}^{d-1}$ also satisfies condition (\ref{ghmfhjm}).
\end{proof}

In the following lemma we use the convergence in (\ref{bfdnsfgnsf}) to show that $\boldsymbol{M}_{\infty}$ almost surely is a rank-one operator. A key-step in the proof is the equality (\ref{ndhndh}) below, which with  (\ref{bfdnsfgnsf}) and the observation that $\Vert \sqrt{\boldsymbol{M}_N}\Vert \leq 1$ yields the limit  in (\ref{awaaerber}). Given that we know that  $\lim_{N\rightarrow\infty}\boldsymbol{M}_N = \boldsymbol{M}_{\infty}\,\, a.s.$, it seems reasonable that we in the limit $N\rightarrow\infty $ obtain the proportionality in  (\ref{kjgbangjb}). The latter does via Lemma  \ref{netrswethn} imply the desired result that $\boldsymbol{M}_{\infty}$ almost surely is a rank-one operator. However, there is a complication to this reasoning, namely the sequence of unitary operators, $\boldsymbol{U}_{N}$. These unitary operators are the result of a polar decomposition of the operators $A_{x_N}\cdots A_{x_1}/\sqrt{\Tr(A_{x_1}^{\dagger}  \cdots A_{x_N}^{\dagger}A_{x_N}\cdots A_{x_1})}$, and we have very little control of the sequence $(\boldsymbol{U}_{N})_{N\in\mathbb{N}}$, and in particular whether it possesses a limit $\boldsymbol{U}_{\infty}$. However, we can mend this issue by using the fact that the set of unitary operators on a finite-dimensional Hilbert space is sequentially compact.  Recall that a topological space, $C$, is sequentially compact if, for every sequence $(x_{j})_{j\in\mathbb{N}}\subset C$, there exists a subsequence $(x_{j_k})_{k\in\mathbb{N}}$ such that $x_{j_k}$ converges to an element in $C$. On a finite-dimensional complex  Hilbert space with dimension $D$, the set of unitary operators, $U(D)$, forms a sequentially compact (as well as compact) space. Hence,   whenever we have a  sequence $(U_j)_{j\in\mathbb{N}}$ in $U(D)$, then there exists a subsequence $(U_{j_k})_{k\in\mathbb{N}}$ such that $U_{j_k}$ converges to an element in $U(D)$. 

\begin{lemma}
\label{adfbafdba} 
With the assumptions in Lemma \ref{fnjjnfddanj}, let $\boldsymbol{M}_{\infty}$ be the  random operator  guaranteed by Lemma \ref{fnjjnfddanj}. If $\{A_x\}_{x=0}^{d-1}$ satisfies the purity condition, then 
 \begin{equation}
 \mathrm{rank}(\boldsymbol{M}_{\infty}) = 1\quad a.s. 
 \end{equation}
 \end{lemma}
 
\begin{proof}
In order to prove this lemma, let us start defining  $M_{x_N,\ldots,x_1} :=  A_{x_1}^{\dagger} \cdots A_{x_N}^{\dagger}A_{x_N}\cdots A_{x_1}/\Tr(A_{x_1}^{\dagger}  \cdots A_{x_N}^{\dagger}A_{x_N}\cdots A_{x_1})$, and thus we have $\boldsymbol{M}_N = M_{\boldsymbol{x}_N,\ldots,\boldsymbol{x}_1}$.
With a unitary operator $U_{x_N,\ldots,x_1}$, we make a polar decomposition such that $U_{x_N,\ldots,x_1}\sqrt{M_{x_N,\ldots,x_1}} = A_{x_N}\cdots A_{x_1}/\sqrt{\rule{0mm}{0ex}\Tr(A_{x_1}^{\dagger}  \cdots A_{x_N}^{\dagger}A_{x_N}\cdots A_{x_1})}$,  and define $\boldsymbol{U}_{N} := U_{\boldsymbol{x}_N,\ldots,\boldsymbol{x}_1}$. Then, we have
\begin{equation*}
\begin{split}
E\Big(\Vert \boldsymbol{M}_{N+p} -\boldsymbol{M}_N\Vert & \Big\vert \boldsymbol{x}_N = x_N,\ldots,\boldsymbol{x}_1 = x_1\Big) \\
= & \sum_{x_{N+p},\ldots,x_{N+1}}\bigg\Vert 
\sqrt{M_{x_N,\ldots,x_1}}U_{x_N,\ldots,x_1}^{\dagger}A_{x_{N+1}}^{\dagger}\cdots A_{x_{N+p}}^{\dagger}  A_{x_{N+p}} \cdots A_{x_{N+1}}U_{x_N,\ldots,x_1}\sqrt{M_{x_N,\ldots,x_1}}
\\
& \quad\quad- M_{x_N,\ldots,x_1} \Tr(  A_{x_{N+1}}^{\dagger} \cdots  A_{x_{N+p}}^{\dagger}A_{x_{N+p}} \cdots  A_{x_{N+1}}U_{x_N,\ldots,x_1}M_{x_N,\ldots,x_1}U_{x_N,\ldots,x_1}^{\dagger}  )
\bigg\Vert,\\
= & E\left( \sum_{x'_{p},\ldots,x'_{1}}\bigg\Vert 
\sqrt{\boldsymbol{M}_{N}}\boldsymbol{U}_{N}^{\dagger}A_{x'_{1}}^{\dagger}\cdots A_{x'_{p}}^{\dagger}  A_{x'_{p}} \cdots A_{x'_{1}}\boldsymbol{U}_{N}\sqrt{\boldsymbol{M}_{N}}\right.
\\
&\left. \quad\quad- \boldsymbol{M}_{N} \Tr(  A_{x'_{1}}^{\dagger} \cdots  A_{x'_{p}}^{\dagger}A_{x'_{p}} \cdots  A_{x'_{1}}\boldsymbol{U}_{N}\boldsymbol{M}_{N}\boldsymbol{U}_{N}^{\dagger}  )
\bigg\Vert \hspace{0.1cm}\bigg\vert \boldsymbol{x}_{N} = x_N,\ldots,\boldsymbol{x}_1=x_1  \right),
\end{split}
\end{equation*}
where we in the second equality have renamed the indices $x_{N+1},\ldots,x_{N+p}$ to $x'_{1},\ldots,x'_{p}$. Consequently,
\begin{equation}
\label{ndhndh}
\begin{split}
 E\Big(\Vert \boldsymbol{M}_{N+p} -\boldsymbol{M}_N\Vert \Big\vert \boldsymbol{x}_N ,\ldots,\boldsymbol{x}_1\Big)
& = \sum_{x'_{p},\ldots,x'_{1}} \bigg\Vert 
\sqrt{\boldsymbol{M}_{N}}\boldsymbol{U}_{N}^{\dagger}A_{x'_{1}}^{\dagger}\cdots A_{x'_{p}}^{\dagger}  A_{x'_{p}} \cdots A_{x'_{1}}\boldsymbol{U}_{N}\sqrt{\boldsymbol{M}_{N}}
\\
& \quad\quad\quad\quad- \boldsymbol{M}_{N} \Tr(  A_{x'_{1}}^{\dagger} \cdots  A_{x'_{p}}^{\dagger}A_{x'_{p}} \cdots  A_{x'_{1}}\boldsymbol{U}_{N}\boldsymbol{M}_{N}\boldsymbol{U}_{N}^{\dagger}  )\bigg\Vert,
\end{split}
\end{equation}
where we have used that $\boldsymbol{M}_N$ and $\boldsymbol{U}_N$ are deterministic functions of $ \boldsymbol{x}'_{N} ,\ldots,\boldsymbol{x}'_1$.
Since $\boldsymbol{M}_N$ is a density operator, or the zero operator, it follows that $\Vert \sqrt{\boldsymbol{M}_N}\Vert \leq 1$. By combining this observation with (\ref{ndhndh}), and with Lemma  \ref{ghdgguk}, it follows that  
\begin{equation}
\label{awaaerber}
\begin{split}
& \lim_{N\rightarrow \infty} \sum_{x_{p},\ldots,x_{1}} \bigg\Vert 
\boldsymbol{M}_{N}\boldsymbol{U}_{N}^{\dagger}A_{x_{1}}^{\dagger}\cdots A_{x_{p}}^{\dagger}  A_{x_{p}} \cdots A_{x_{1}}\boldsymbol{U}_{N}\boldsymbol{M}_{N}\\
& \quad \quad\quad \quad\quad \quad - \boldsymbol{M}^2_{N} \Tr(  A_{x_{1}}^{\dagger} \cdots  A_{x_{p}}^{\dagger}A_{x_{p}} \cdots  A_{x_{1}}\boldsymbol{U}_{N}\boldsymbol{M}_{N}\boldsymbol{U}_{N}^{\dagger}  )\bigg\Vert = 0\quad a.s.
\end{split}
\end{equation}
Next, we recall that Lemma \ref{fnjjnfddanj} guarantees that $\lim_{N\rightarrow\infty}\boldsymbol{M}_{N} = \boldsymbol{M}_{\infty}\,\, a.s.$, where $\boldsymbol{M}_{\infty}$ almost surely is a density operator. Let $\omega\in\Omega$ be such that $\lim_{N\rightarrow\infty}\boldsymbol{M}_{N}(\omega) = \boldsymbol{M}_{\infty}(\omega)$, where $\boldsymbol{M}_{\infty}(\omega)$ is a density operator, and the limit in (\ref{awaaerber}) holds. The latter implies a sequence of unitary operators  $(\boldsymbol{U}_{N}(\omega))_{N\in\mathbb{N}}\subset U(D)$. By the  sequential compactness of $U(D)$, it follows that there exists a subsequence $\big(\boldsymbol{U}_{N_k}(\omega)\big)_{k\in\mathbb{N}}$ and an element  $\boldsymbol{U}_{\infty}(\omega)\in U(D)$, such that $\lim_{k\rightarrow\infty}\boldsymbol{U}_{N_k}(\omega) = \boldsymbol{U}_{\infty}(\omega)$. It still remains true that $\lim_{k\rightarrow\infty}\boldsymbol{M}_{N_k}(\omega) = \boldsymbol{M}_{\infty}(\omega)$, and similarly the limit in  (\ref{awaaerber}) remains true with $N$ replaced with $N_k$. With the definition
\begin{align*}
B_k(\omega) := \sum_{x_{p},\ldots,x_{1}} \bigg\Vert
\boldsymbol{M}_{N_k}&(\omega)\boldsymbol{U}_{N_k}^{\dagger}(\omega)A_{x_{1}}^{\dagger}\cdots A_{x_{p}}^{\dagger} A_{x_{p}} \cdots A_{x_{1}}\boldsymbol{U}_{N_k}(\omega)\boldsymbol{M}_{N_k}(\omega) \\
 & \quad \quad - \boldsymbol{M}^2_{N_k}(\omega)\Tr\big[  A_{x_{1}}^{\dagger} \cdots  A_{x_{p}}^{\dagger}A_{x_{p}} \cdots  A_{x_{1}}\boldsymbol{U}_{N_k}(\omega)\boldsymbol{M}_{N_k}(\omega)\boldsymbol{U}_{N_k}(\omega)^{\dagger}  \big]\bigg\Vert,
\end{align*}
it thus follows by (\ref{awaaerber}) that $B_k(\omega) \rightarrow 0$.
Define
\begin{equation}
\begin{split}
\label{hmdghmg}
B_{\infty}(\omega) & :=  \sum_{x_{p},\ldots,x_{1}} \bigg\Vert 
\boldsymbol{M}_{\infty}(\omega)\boldsymbol{U}_{\infty}(\omega)^{\dagger}A_{x_{1}}^{\dagger}\cdots A_{x_{p}}^{\dagger} A_{x_{p}} \cdots A_{x_{1}}\boldsymbol{U}_{\infty}(\omega)\boldsymbol{M}_{\infty}(\omega) \\
& \quad \quad \quad \quad \quad \quad - \boldsymbol{M}^2_{\infty}(\omega) \Tr[ A_{x_{1}}^{\dagger} \cdots  A_{x_{p}}^{\dagger}A_{x_{p}} \cdots A_{x_{1}}\boldsymbol{U}_{\infty}(\omega)\boldsymbol{M}_{\infty}(\omega)\boldsymbol{U}_{\infty}(\omega)^{\dagger}  ]\bigg\Vert.
\end{split}
\end{equation}
Next we wish to show that $B_{k}(\omega)\rightarrow B_{\infty}(\omega)$. By the inverted triangle inequality, a rearrangement, and the triangle inequality, one obtains
\begin{equation}
\label{gnsfgns}
\begin{split}
|B_{\infty}(\omega)-B_k(\omega)|
  \leq  &\sum_{x_{p},\ldots,x_{1}} \bigg\Vert 
\boldsymbol{M}_{\infty}(\omega)\boldsymbol{U}_{\infty}(\omega)^{\dagger}A_{x_{1}}^{\dagger}\cdots A_{x_{p}}^{\dagger}  A_{x_{p}} \cdots A_{x_{1}}\boldsymbol{U}_{\infty}(\omega)\boldsymbol{M}_{\infty}(\omega)\\
&  
 \quad \quad-\boldsymbol{M}_{N_k}(\omega)\boldsymbol{U}_{N_k}^{\dagger}(\omega)A_{x_{1}}^{\dagger}\cdots A_{x_{p}}^{\dagger}  A_{x_{p}} \cdots A_{x_{1}}\boldsymbol{U}_{N_k}(\omega)\boldsymbol{M}_{N_k}(\omega)\bigg\Vert \\
& +\sum_{x_{p},\ldots,x_{1}} \bigg\Vert  \boldsymbol{M}^2_{\infty}(\omega) \Tr\big(  A_{x_{1}}^{\dagger} \cdots  A_{x_{p}}^{\dagger}A_{x_{p}} \cdots  A_{x_{1}}\boldsymbol{U}_{\infty}(\omega)\boldsymbol{M}_{\infty}(\omega)\boldsymbol{U}_{\infty}(\omega)^{\dagger} \big)\\
&\quad \quad- \boldsymbol{M}^2_{N_k}(\omega) \Tr\big(  A_{x_{1}}^{\dagger} \cdots  A_{x_{p}}^{\dagger}A_{x_{p}} \cdots  A_{x_{1}}\boldsymbol{U}_{N_k}(\omega)\boldsymbol{M}_{N_k}(\omega)\boldsymbol{U}_{N_k}(\omega)^{\dagger}  \big)\bigg\Vert\\
\end{split}
\end{equation}
The goal is to utilize the fact that $\boldsymbol{M}_{N_k}(\omega)\rightarrow \boldsymbol{M}_{\infty}(\omega)$, and thus that  $\boldsymbol{M}^2_{N_k}(\omega)\rightarrow \boldsymbol{M}^2_{\infty}(\omega)$, and similarly that $\boldsymbol{U}_{N_k}(\omega)\rightarrow \boldsymbol{U}_{\infty}(\omega)$. To this end, in the first sum in (\ref{gnsfgns}), inside the norm, one can subtract and add $\boldsymbol{M}_{N_k}(\omega)\boldsymbol{U}_{N_k}^{\dagger}(\omega)A_{x_{1}}^{\dagger}\cdots A_{x_{p}}^{\dagger}  A_{x_{p}} \cdots A_{x_{1}}\boldsymbol{U}_{\infty}(\omega)\boldsymbol{M}_{\infty}(\omega)$. Similarly in the second sum, we subtract and add $\boldsymbol{M}^2_{N_k}(\omega) \Tr(  A_{x_{1}}^{\dagger} \cdots  A_{x_{p}}^{\dagger}A_{x_{p}} \cdots  A_{x_{1}}\boldsymbol{U}_{\infty}(\omega)\boldsymbol{M}_{\infty}(\omega)\boldsymbol{U}_{\infty}(\omega)^{\dagger}  )$ inside of the norm. 
One can repeatedly use the triangle inequality, subtractions and additions in the similar spirit as above, and general relations such as $\Vert AB\Vert \leq \Vert A\Vert \Vert B\Vert$, $|\Tr(AB)|\leq \Vert A\Vert \Vert B\Vert _{1}$, as well as observations such as 
$\Vert \boldsymbol{U}_{\infty}(\omega)\boldsymbol{M}_{\infty}(\omega)\Vert \leq 1$, $\big\Vert \boldsymbol{M}_{N_k}(\omega)\boldsymbol{U}_{N_k}^{\dagger}(\omega)\big\Vert\leq 1$, $\Vert \boldsymbol{U}_{\infty}(\omega)\boldsymbol{M}_{\infty}(\omega)\boldsymbol{U}_{\infty}(\omega)^{\dagger}  \Vert_{1}  = \Vert \boldsymbol{M}_{\infty}(\omega) \Vert_{1} = 1$,
$\sum_{x_{p},\ldots,x_{1}}\Vert A_{x_{1}}^{\dagger}\cdots A_{x_{p}}^{\dagger}  A_{x_{p}} \cdots A_{x_{1}}\Vert \leq d^p$, and $\Vert \boldsymbol{M}^2_{N_k}(\omega)\Vert \leq 1$
 to show that
\begin{equation*}
\begin{split}
|B_{\infty}(\omega)-B_n(\omega)| 
\leq    &\hspace{0.1cm} d^p\Vert 
\boldsymbol{M}_{\infty}(\omega)-\boldsymbol{M}_{N_k}(\omega)\big\Vert +d^p\Vert\boldsymbol{U}_{\infty}(\omega)^{\dagger}-\boldsymbol{U}_{N_k}^{\dagger}(\omega)\Vert \\
& +  d^p  \Vert\boldsymbol{M}_{\infty}(\omega)-\boldsymbol{M}_{N_k}(\omega)\Vert +d^p\Vert\boldsymbol{U}_{\infty}(\omega)-\boldsymbol{U}_{N_k}(\omega)\Vert  \\
& + d^p\Vert\boldsymbol{M}^2_{\infty}(\omega)-\boldsymbol{M}^2_{N_k}(\omega)\Vert
+d^p\Vert\boldsymbol{U}_{\infty}(\omega)-\boldsymbol{U}_{N_k}(\omega)\Vert_1\\
& + d^p  \Vert \boldsymbol{M}_{\infty}(\omega)-\boldsymbol{M}_{N_k}(\omega)\Vert_1 +d^p\Vert\boldsymbol{U}_{\infty}(\omega)^{\dagger}-\boldsymbol{U}_{N_k}(\omega)^{\dagger}\Vert_1,\\
=: & C_k(\omega),
\end{split}
\end{equation*}
and thus  $C_k(\omega)\rightarrow 0$.
We can conclude that $B_{\infty}(\omega)-B_k(\omega) \leq |B_{\infty}(\omega)-B_k(\omega)|\leq C_k(\omega)$, which implies $0\leq B_{\infty}(\omega) \leq B_k(\omega) + C_k(\omega)$. By combining this observation with  $B_k(\omega) \rightarrow 0$ and $C_k(\omega)\rightarrow 0$, as well as with the definition of $B_{\infty}(\omega)$ in (\ref{hmdghmg}), we can conclude that
\begin{equation}
\begin{split}
B_{\infty}(\omega)  & = \sum_{x_{p},\ldots,x_{1}} \bigg\Vert 
\boldsymbol{M}_{\infty}(\omega)\boldsymbol{U}_{\infty}(\omega)^{\dagger}A_{x_{1}}^{\dagger}\cdots A_{x_{p}}^{\dagger}A_{x_{p}} \cdots A_{x_{1}}\boldsymbol{U}_{\infty}(\omega)\boldsymbol{M}_{\infty}(\omega) \\
 & \quad\quad\quad \quad \quad - \boldsymbol{M}^2_{\infty}(\omega) \Tr\big(  A_{x_{1}}^{\dagger} \cdots  A_{x_{p}}^{\dagger}A_{x_{p}} \cdots  A_{x_{1}}\boldsymbol{U}_{\infty}(\omega)\boldsymbol{M}_{\infty}(\omega)\boldsymbol{U}_{\infty}(\omega)^{\dagger}  \big)\bigg\Vert,\\
& = 0.
\end{split}
\end{equation}
This in turn implies
\begin{equation}
\label{kjgbangjb}
\begin{split}
& \boldsymbol{M}_{\infty}(\omega)\boldsymbol{U}_{\infty}^{\dagger}(\omega)A_{x_{1}}^{\dagger}\cdots A_{x_{p}}^{\dagger}  A_{x_{p}} \cdots A_{x_{1}}\boldsymbol{U}_{\infty}(\omega)\boldsymbol{M}_{\infty}(\omega)\propto\boldsymbol{M}_{\infty}(\omega)\boldsymbol{U}_{\infty}^{\dagger}(\omega)\boldsymbol{U}_{\infty}(\omega)\boldsymbol{M}_{\infty}(\omega).
 \end{split}
\end{equation}

Since we have assumed the purity condition, it follows by Lemma \ref{netrswethn}, with  $O:=\boldsymbol{U}_{\infty}(\omega)\boldsymbol{M}_{\infty}(\omega)$, that $\mathrm{rank}(\boldsymbol{M}_{\infty}(\omega)) =\mathrm{rank}(\boldsymbol{U}_{\infty}(\omega)\boldsymbol{M}_{\infty}(\omega)) = 1$. Since this holds for almost all elements $\omega$ in the sample space, we can conclude that  $\mathrm{rank}(\boldsymbol{M}_{\infty})  = 1\,\, a.s.$
\end{proof}

\subsection{\label{SecRankOneImpliesPur}$\mathrm{rank}(\boldsymbol{M}_{\infty}) = 1\,\, a.s.$ implies the purity condition}

While we in Appendix \ref{SecPurImpliesRankOne} demonstrated that the purity condition is  sufficient for $\boldsymbol{M}_{\infty}$ being a rank-one operator, we here show that it also is a necessary condition. The idea is to assume that $\boldsymbol{M}_{\infty}$ has rank one, but that the purity condition does not hold. The latter means that there exists a projector $P$ with $\mathrm{rank}(P)>1$, while still $P A_{x_1}^{\dagger}\cdots A_{x_N}^{\dagger}A_{x_N}\cdots A_{x_1}P \propto P$. The latter is then showed to imply  $P\boldsymbol{M}_{\infty}(\omega)P  \propto P$. However, since $\mathrm{rank}(\boldsymbol{M}_{\infty}) = 1$, the only possibility is that $P\boldsymbol{M}_{\infty}(\omega)P  =0$. This turns out to be in contradiction with $\sum_{x=0}^{d-1}A_x^{\dagger}A_x = \1$.

\begin{lemma}
\label{ghdnghnhn}
Let $\{A_x\}_{x=0}^{d-1}$ be operators on a finite-dimensional complex Hilbert space, $\mathcal{H}$, such that $\sum_{x=0}^{d-1}A_x^{\dagger}A_x = \1$.  Let $(\boldsymbol{M}_N)_{N\in\mathbb{N}}$  be as defined in (\ref{gndghnghn}) with respect to $(\boldsymbol{x}_N)_{N\in\mathbb{N}}$ and distributed as in (\ref{dgnfsgn}). 
 Let $\boldsymbol{M}_{\infty} :=  \lim_{N\rightarrow \infty} \boldsymbol{M}_N\,\, a.s.$, as guaranteed by Lemma  \ref{fnjjnfddanj}.
If  $\mathrm{rank}(\boldsymbol{M}_{\infty}) = 1\,\, a.s.$, then $\{A_x\}_{x=0}^{d-1}$ satisfies the purity condition in Definition \ref{DefPur}.
\end{lemma}

\begin{proof}
We proceed via a proof by contradiction, and thus assume that $\{A_x\}_{x=0}^{d-1}$ is such that $\mathrm{rank}(\boldsymbol{M}_{\infty}) = 1\,\, a.s$, but that the purity condition does  \emph{not} hold. 
The latter means that there exists a projector $P$ such that
\begin{equation}
\label{hgmhmm}
 P A_{x_1}^{\dagger}\cdots A_{x_N}^{\dagger}A_{x_N}\cdots A_{x_1}P \propto P,\hspace{0.5cm}\forall N\in\mathbb{N},\quad \forall (x_1,\ldots,x_N) \in \{0,\ldots,d-1\}^{\times N},
\end{equation}
but $\mathrm{rank}(P) >1$.
Recall that
\begin{equation}
\begin{split}
\boldsymbol{M}_{N}= & \left\{ \begin{matrix} \dfrac{A_{\boldsymbol{x}_1}^{\dagger}\cdots A_{\boldsymbol{x}_N}^{\dagger}A_{\boldsymbol{x}_N}\cdots  A_{\boldsymbol{x}_1}}{\Tr(A_{\boldsymbol{x}_1}^{\dagger}\cdots A_{\boldsymbol{x}_N}^{\dagger}A_{\boldsymbol{x}_N}\cdots  A_{\boldsymbol{x}_1})} & \textrm{if} & \Tr(A_{\boldsymbol{x}_1}^{\dagger}\cdots A_{\boldsymbol{x}_N}^{\dagger}A_{\boldsymbol{x}_N}\cdots  A_{\boldsymbol{x}_1}) \neq 0,\\
0  & \textrm{if} &  \Tr(A_{\boldsymbol{x}_1}^{\dagger}\cdots A_{\boldsymbol{x}_N}^{\dagger}A_{\boldsymbol{x}_N}\cdots  A_{\boldsymbol{x}_1}) = 0,
\end{matrix}\right.\\
\end{split}
\end{equation}
where we note that 
 $\Tr(A_{\boldsymbol{x}_1}^{\dagger}\cdots A_{\boldsymbol{x}_N}^{\dagger}A_{\boldsymbol{x}_N}\cdots  A_{\boldsymbol{x}_1}) = 0$ if and only if $A_{\boldsymbol{x}_1}^{\dagger}\cdots A_{\boldsymbol{x}_N}^{\dagger}A_{\boldsymbol{x}_N}\cdots  A_{\boldsymbol{x}_1} = 0$.
By (\ref{hgmhmm}), it thus follows that $P \boldsymbol{M}_NP \propto P$.
Let $\omega\in\Omega$ be such that $\lim_{N\rightarrow\infty} \boldsymbol{M}_N(\omega)  = \boldsymbol{M}_{\infty}(\omega)$.
Consequently,
\begin{equation}
\label{qadqgrna}
\lim_{N\rightarrow\infty} \Vert P\boldsymbol{M}_N(\omega)P  - P\boldsymbol{M}_{\infty}(\omega)P\Vert = 0.
\end{equation}
By  $P \boldsymbol{M}_NP \propto P$, we know that there exists a proportionality constant, $a_{N}(\omega)$, for each $N$ and $\omega$, such that 
\begin{equation}
\label{ziolhg}
P \boldsymbol{M}_N(\omega)P = a_N(\omega) P.
\end{equation}
Next we use the general relation $|\Tr(AB)|\leq \Vert A\Vert_1\Vert B\Vert$ to show
\begin{equation}
\begin{split}
|a_N(\omega) \Tr(P) -\Tr(P \boldsymbol{M}_{\infty}(\omega)P)| = & \Big|\Tr\Big(\1\big(a_N(\omega)P -P \boldsymbol{M}_{\infty}(\omega) P\big)\Big)\Big|,\\
\leq & D\Vert P \boldsymbol{M}_N(\omega)P  -P \boldsymbol{M}_{\infty}(\omega) P\Vert \rightarrow  0,
\end{split}
\end{equation}
where we have used $\Vert \1\Vert_1 = D$, (\ref{ziolhg}) and (\ref{qadqgrna}). With $a_{\infty}(\omega) := \Tr(P \boldsymbol{M}_{\infty}(\omega))/\Tr(P)$, we can thus conclude that $\lim_{N\rightarrow\infty}|a_N(\omega)-a_{\infty}(\omega)| = 0$.
Hence,
\begin{equation*}
\begin{split}
\Vert a_{\infty}(\omega)P -P\boldsymbol{M}_{\infty}(\omega)P\Vert = & \Vert a_{\infty}(\omega)P - a_N(\omega)P + a_N(\omega)P -P\boldsymbol{M}_{\infty}(\omega)P\Vert,\\
& \leq  | a_{\infty}(\omega)- a_N(\omega)| + \Vert P\boldsymbol{M}_N(\omega)P -P\boldsymbol{M}_{\infty}(\omega)P\Vert\rightarrow 0.
\end{split}
\end{equation*}

We can thus conclude that $P\boldsymbol{M}_{\infty}(\omega)P  \propto P$, and hence $P\boldsymbol{M}_\infty P \propto P \,\, a.s.$
However, since $\boldsymbol{M}_{\infty}$ by assumption is rank-one $a.s.$, and $\mathrm{rank}(P)>1$, the only possibility is that the proportionality constant is zero, i.e., that $P\boldsymbol{M}_\infty P  = 0\,\, a.s.$ Next, we note that  $E(\boldsymbol{M}_N) =\1/D$. By (\ref{nghngfhm}) in Lemma \ref{fnjjnfddanj}, we know that $E(\boldsymbol{M}_N)\rightarrow E(\boldsymbol{M}_{\infty})$, and thus $E(\boldsymbol{M}_{\infty})  =\1/D$. However, this is in contradiction with $P\boldsymbol{M}_\infty P  = 0\quad a.s.$
\end{proof}

\subsection{$w(N)$ goes to zero exponentially if and only if the purity condition holds}

In Appendices \ref{SecPurImpliesRankOne} and \ref{SecRankOneImpliesPur}, we have shown that $\{A_{x}\}_{x=0}^{d-1}$ satisfies the purity condition if and only if $\mathrm{rank}(\boldsymbol{M}_{\infty}) = 1$. Here we show that the latter in turn is equivalent to $\lim_{N\rightarrow\infty}w(N) = 0$, and that this in turn is equivalent to $w(N)$ converging exponentially fast to zero.

If $\boldsymbol{M}_{\infty}$ has rank one, i.e., $\mathrm{rank}(\boldsymbol{M}_{\infty}) = 1$, then it follows that $\Vert\boldsymbol{M}_{\infty}\Vert = 1$ and we can relate $\Vert\boldsymbol{M}_{\infty}\Vert -\Vert \boldsymbol{M}_N\Vert  = 1-\Vert \boldsymbol{M}_N\Vert$ to the eigenvalues of $\boldsymbol{M}_N$ and the singular values of $\boldsymbol{W}_N$. The latter directly connects to the definition of $w(N)$ in (\ref{bfgndfgnqsygn}). To this end, we introduce the following notation.
For a general operator, $O$, on a space of finite dimension, $D$, let 
$\nu_1^{\downarrow}(O)\geq \cdots \geq \nu_D^{\downarrow}(O)$ be the ordered singular values of $O$. Similarly, for a Hermitian operator, $J$, let $\lambda_1^{\downarrow}(J)\geq \cdots \geq \lambda_D^{\downarrow}(J)$ be the ordered eigenvalues of $J$.

The fact that $w(N)$ converges to zero does, of course, not guarantee that $w(N)$ converges exponentially fast to zero. The latter we obtain by first showing that $w(N)$ is submultiplicative, i.e., $w(N+M)\leq w(N)w(M)$, which implies that $\log w(N)$ is subadditive.

We obtain the submultiplicativity  by rewriting $w(N)$ in terms of the norm of the second order exterior power of $A_{x_N}\cdots A_{x_1}$. In order to introduce the exterior power of an operator, consider a  Hilbert space, $\mathcal{H}$, with an orthonormal basis, $\{|j\rangle\}_{j=1}^{D}$, and $D =\dim\mathcal{H}$.
On the product space $\mathcal{H}\otimes\mathcal{H}$, we construct the swap-operator, $S := \sum_{j,k=1}^{D}|j\rangle\langle k|\otimes|k\rangle\langle j|$, where one may note that $S^2 =\1\otimes\1$ and $S^{\dagger} = S$. We also define the projector $P_A := (\1\otimes\1- S)/2$ onto the the anti-symmetric subspace of  $\mathcal{H}\otimes\mathcal{H}$. 
For an operator $O$ on $\mathcal{H}$, we define the exterior power (of degree two) of $O$ as $\mywedge(O) := P_A [O\otimes O]P_A$. A consequence of this definition is that $\Vert \mywedge(O)\Vert = \nu_1^{\downarrow}(O)\nu_2^{\downarrow}(O)$. 
 By comparing these definitions with (\ref{bfgndfgnqsygn}) below, we can conclude that $w(N) =  \sum_{x_1,\ldots,x_N=0}^{d-1}\Vert\mywedge(A_{x_N}\cdots A_{x_1})\Vert$.

One may note that the above construction presumes that $O$ is an operator from one space to itself. However, due to the operator $F$ [see e.g.~(\ref{mainFdef})], we need a generalization to mappings from one space to another, $O:\mathcal{H}_1\rightarrow \mathcal{H}_2$. However, analogous to $P_A$, we can for these two spaces let $P_{A}^{(1)}$ and $P^{(2)}_A$ be the projectors onto the anti-symmetric subspaces of $\mathcal{H}_1\otimes\mathcal{H}_1$ and   $\mathcal{H}_2\otimes\mathcal{H}_2$, respectively.
Based on these, we define the generalization $\mywedge(O) := P^{(2)}_A [O\otimes O]P_A^{(1)}$. For this generalization it remains true that $\Vert \mywedge(O)\Vert = \nu_1^{\downarrow}(O)\nu_2^{\downarrow}(O)$. By comparing with (\ref{bofboneabona})
we see that $f(N) = \sum_{x_N,\ldots,x_1= 1}^{d-1}\Vert \mywedge (FA_{x_N}\cdots A_{x_1}\sqrt{\sigma})\Vert$.
A further observation that also holds for the generalization is 
\begin{equation}
\label{eqBound}
\Vert \mywedge(O)\Vert \leq \Vert O\Vert^2.
\end{equation}
Moreover, if $O_A:\mathcal{H}_1\rightarrow\mathcal{H}_2$ and  $O_B:\mathcal{H}_2\rightarrow\mathcal{H}_3$, then  
 $\mywedge(O_BO_A) = \mywedge(O_B)\mywedge(O_A)$, and consequently $\Vert \mywedge(O_BO_A)\Vert \leq \Vert\mywedge(O_B)\Vert \Vert\mywedge(O_A)\Vert$. 

The exponential decay of $w(N)$ is obtained by combining $\lim_{N\rightarrow \infty}w(N) = 0$ with the submultiplicativity of $\log w(N)$ and Fekete's subadditivity lemma. Fekete's Lemma is commonly attributed to Ref.~\onlinecite{Fekete}. For a proof, see Lemma 1.2.1 in Ref.~\onlinecite{Steele}, and for a historical overview, see Section 1.10 in Ref.~\onlinecite{Steele}.

\begin{lemma}[Fekete's subadditive lemma]
\label{FeketesLemma}
Let $(a_N)_{N\in\mathbb{N}}$ be a subadditive sequence of real numbers, i.e., $a_{N+M}\leq a_N + a_M$. Then the limit $\lim_{N\rightarrow\infty} a_N/N$ is well defined (but may be $-\infty$) and 
\begin{equation}
\lim_{N\rightarrow\infty}\frac{a_N}{N} = \inf_{N\in\mathbb{N}}\frac{a_N}{N}.
\end{equation}
\end{lemma}

\begin{prop}
\label{wMain}
Let $\{A_x\}_{x=0}^{d-1}$ be linear operators on a finite-dimensional Hilbert space, such that $\sum_{x=0}^{d-1}A_x^{\dagger}A_x = \1$. 
Define
\begin{equation}
\label{bfgndfgnqsygn}
w(N) := \sum_{x_1,\ldots,x_N=0}^{d-1}\nu_1^{\downarrow}(A_{x_N}\cdots A_{x_1})\nu_2^{\downarrow}(A_{x_N}\cdots A_{x_1}).
\end{equation}
The following statements are equivalent:
\begin{enumerate}
\item $\{A_x\}_{x=0}^{d-1}$ satisfies the purity condition in Definition \ref{DefPur}.
\item $\lim_{N\rightarrow\infty}w(N) = 0$.
\item There exist real constants $C'\geq 0$ and $0<\gamma <1$ such that 
\begin{equation}
\label{dfbdfb}
w(N)\leq C'\gamma^N,\quad \forall N\in\mathbb{N}.
\end{equation}
\end{enumerate}
\end{prop}

\begin{proof}
{\bf 1 $\boldsymbol{\Rightarrow}$ 2:} 
Let  $\boldsymbol{M}_N$ be as defined in (\ref{gndghnghn}).  We first distinguish the two cases that $\boldsymbol{M}_N$  is a density operator, or that it is the zero operator. In the case that $\boldsymbol{M}_N$  is a density operator, it follows that $1 = \Tr(\boldsymbol{M}_N)  \geq  \lambda_1^{\downarrow}(\boldsymbol{M}_N) + \lambda_2^{\downarrow}(\boldsymbol{M}_N)$, and thus $1 \geq 1 -\lambda_1^{\downarrow}(\boldsymbol{M}_N)  \geq  \lambda_2^{\downarrow}(\boldsymbol{M}_N) \geq 0$. By noting that $\Vert\boldsymbol{M}_N\Vert = \lambda^{\downarrow}_1(\boldsymbol{M}_N)$, we thus get $\sqrt{\Vert\boldsymbol{M}_N\Vert (1 -\Vert\boldsymbol{M}_N\Vert)}  \geq \sqrt{\lambda_1^{\downarrow}(\boldsymbol{M}_N) \lambda_2^{\downarrow}(\boldsymbol{M}_N)}$. Since $\boldsymbol{M}_N$ is assumed to be a density operator, it moreover follows that $\Vert\boldsymbol{M}_N\Vert\leq 1$, and thus 
\begin{equation}
\label{dghmdghm}
\begin{split}
\sqrt{|1 -\Vert\boldsymbol{M}_N\Vert|}   \geq &\sqrt{\lambda_1^{\downarrow}(\boldsymbol{M}_N) \lambda_2^{\downarrow}(\boldsymbol{M}_N)}.
 \end{split}
\end{equation}
In the case that $\boldsymbol{M}_N$ is the zero operator, then (\ref{dghmdghm}) is trivially true.

By Lemma \ref{fnjjnfddanj}, we know that $\boldsymbol{M}_{\infty}$ almost surely is a density operator. By Lemma \ref{adfbafdba}, we also know that $\boldsymbol{M}_{\infty}$ almost surely is a rank-one operator. Hence, $\boldsymbol{M}_{\infty}$ almost surely corresponds to a pure state. Consequently, $\Vert \boldsymbol{M}_{\infty}\Vert =1\,\, a.s.$ Combining this observation with the inverted triangle inequality yields 
\begin{equation}
\label{bfgnsfgn}
    \sqrt{\Vert \boldsymbol{M}_{\infty} -\boldsymbol{M}_N\Vert}\geq \sqrt{|\Vert \boldsymbol{M}_{\infty}\Vert -\Vert\boldsymbol{M}_N\Vert|} = \sqrt{|1 -\Vert\boldsymbol{M}_N\Vert|}\quad a.s.
\end{equation}
Combining (\ref{dghmdghm}) with (\ref{bfgnsfgn}) yields
\begin{equation}
\label{fgnsfgmsf}
\sqrt{\Vert \boldsymbol{M}_{\infty} -\boldsymbol{M}_N\Vert}  \geq  \sqrt{\lambda_1^{\downarrow}(\boldsymbol{M}_N) \lambda_2^{\downarrow}(\boldsymbol{M}_N)}\quad a.s.
\end{equation}
We next observe that
\begin{equation}
\label{fgndfgndgf}
\nu_k^{\downarrow}\left(\frac{\boldsymbol{W}_N}{\sqrt{\Tr(\boldsymbol{W}_N^{\dagger}\boldsymbol{W}_N)} }\right) = \sqrt{\lambda_k^{\downarrow}
\left(\frac{\boldsymbol{W}_N^{\dagger}\boldsymbol{W}_N}{\Tr(\boldsymbol{W}_N^{\dagger}\boldsymbol{W}_N)}\right)}= \sqrt{\lambda_k^{\downarrow}\left(\boldsymbol{M}_N\right)}.
\end{equation}
Thus, (\ref{fgnsfgmsf}) and (\ref{fgndfgndgf}) yields
\begin{equation}
\label{ngfdgfng}
\sqrt{\Vert \boldsymbol{M}_{\infty} -\boldsymbol{M}_N\Vert}  \geq \frac{   
\nu_1^{\downarrow}\left(\boldsymbol{W}_N\right) \nu_2^{\downarrow}\left(\boldsymbol{W}_N\right)   
}
{
\Tr(\boldsymbol{W}_N^{\dagger}\boldsymbol{W}_N)
}\quad a.s.,
\end{equation}
which, by Lemma \ref{htdhgtmdg}, results in
\begin{equation}
\label{fdgnfdgnd}
E\big(\sqrt{\Vert \boldsymbol{M}_{\infty} -\boldsymbol{M}_N\Vert} \big)D \geq 
 E\left(  \frac{   
\nu_1^{\downarrow}\left(\boldsymbol{W}_N\right) \nu_2^{\downarrow}\left(\boldsymbol{W}_N\right)   
}
{
\Tr(\boldsymbol{W}_N^{\dagger}\boldsymbol{W}_N)
}
 \right)D =  w(N).  
\end{equation}
By Lemma \ref{fnjjnfddanj}, we know that  $\boldsymbol{M}_N \rightarrow \boldsymbol{M}_{\infty}$ almost surely. Since the underlying Hilbert space is finite-dimensional, we then have $\Vert \boldsymbol{M}_{\infty} - \boldsymbol{M}_N \Vert \rightarrow 0\,\, a.s.$, and consequently
\begin{equation}
\label{fgnfdgndfg}
\boldsymbol{x}_N :=\sqrt{\Vert\boldsymbol{M}_{\infty}-\boldsymbol{M}_N \Vert} \rightarrow 0\quad a.s.
\end{equation}
We next observe that $\boldsymbol{M}_N$ is a density operator, or the zero operator, and thus  $\Vert\boldsymbol{M}_N\Vert\leq 1$. Hence, $\Vert\boldsymbol{M}_{\infty} - \boldsymbol{M}_N \Vert \leq  \Vert\boldsymbol{M}_{\infty}\Vert + \Vert\boldsymbol{M}_N\Vert \leq  1 +\Vert\boldsymbol{M}_{\infty}\Vert$, which yields
\begin{equation}
\label{dfgndg}
    \boldsymbol{x}_N = \sqrt{\Vert\boldsymbol{M}_{\infty} - \boldsymbol{M}_{N}\Vert}\leq \sqrt{1 +\Vert\boldsymbol{M}_{\infty}\Vert}\leq 1 +\Vert\boldsymbol{M}_{\infty}\Vert=: \boldsymbol{y}.
\end{equation}
By Lemma \ref{fnjjnfddanj} we know that $E(\Vert\boldsymbol{M}_{\infty}\Vert) < +\infty$, and thus $E(\boldsymbol{y}) = 1 + E(\Vert\boldsymbol{M}_{\infty}\Vert) < +\infty$. 
By using this observation and Eqns.~(\ref{fgnfdgndfg}) and (\ref{dfgndg}) into Proposition \ref{Lebesgues}, we can conclude that  
$\lim_{N\rightarrow \infty}E(\sqrt{\Vert\boldsymbol{M}_{\infty}-\boldsymbol{M}_N \Vert})  = 0$. By combining this with  (\ref{fdgnfdgnd}), it follows that $\lim_{N\rightarrow\infty}w(N) = 0$. Hence, we can conclude that statement 1 implies statement 2.

{\bf 2 $\boldsymbol{\Rightarrow}$ 3:} We first make the observation that 
\begin{equation}
\begin{split}
    \Vert \mywedge (A_{x_{N+M}}\cdots A_{x_1})\Vert  \leq \Vert \mywedge&(A_{x_{N+M}}\cdots A_{x_{N+1}})\Vert\Vert\mywedge(A_{x_N}\cdots A_{x_1})\Vert,
\end{split}
\end{equation}
which in turn yields $w(N+M) \leq w(M)w(N)$.
Hence, $w$ is submultiplicative, and thus $\log w(N)$ is subadditive. 
By statement 2 we know that $\lim_{N\rightarrow\infty}w(N) = 0$. It follows that there exists an $N_0\in\mathbb{N}$ such that $\log w(N_0) <0$.
Hence, since $\log w(N)$ is subadditive, it follows by Lemma \ref{FeketesLemma} that 
\begin{equation}
\label{dfbdf}
\begin{split}
0>  \frac{\log w(N_0)}{N_0} \geq  \inf_{N}\frac{\log w(N)}{N} =  \lim_{N\rightarrow\infty}\frac{\log w(N)}{N}. 
\end{split}
\end{equation} 
In the case that the limit is finite, let $l:= \lim_{N\rightarrow\infty}1/N\log w(N)$. By definition of the limit, we know that for any $\epsilon>0$, there exists an $N_{\epsilon}$ such that $[\log w(N) ]/N-l\leq \epsilon$ for all  $N\geq N_{\epsilon}$. We choose an arbitrary but fixed $\epsilon>0$, and thus $w(N)\leq \gamma^N$ for all $N\geq N_{\epsilon}$, where $\gamma := e^{l+\epsilon}$. Define $C':= \max\big\{1, \max_{N=1,\ldots,N_{\epsilon}}w(N)/N\big\}$, and thus (\ref{dfbdfb}) holds.

Finally consider the case that $\lim_{N\rightarrow\infty}1/N\log w(N) = -\infty$. This means that for every $a>0$ there exists an $N_{a}$ such that $[\log w(N)]/N \leq -a$ for all $N\geq N_{a}$, which we can easily rewrite as $w(N) \leq e^{-aN}$. Hence, with $\gamma := e^{-a}$ and $C':= \max\{1, \max_{N=1,\ldots,N_{a}}w(N)/N\}$ we again obtain (\ref{dfbdfb}). We can conclude that statement 2 implies statement 3.

{\bf 3 $\boldsymbol{\Rightarrow}$ 2:} This implication is trivial.

{\bf 2 $\boldsymbol{\Rightarrow}$ 1:}  In our first step, we show that  $\lim_{N\rightarrow\infty}w(N) = 0$ implies that $\Vert\boldsymbol{M}_{\infty}\Vert = 1\,\, a.s.$ 
We first observe that if $\eta$ is a density operator on a complex Hilbert space with finite dimension $D$, then $1-\Vert\eta\Vert \leq  \sqrt{(D-1)D}  \sqrt{\lambda^{\downarrow}_{1}(\eta)\lambda^{\downarrow}_{2}(\eta)}$. We know that $\boldsymbol{M}_N$ is either a density operator, or the zero operator, and thus $1-\Vert\boldsymbol{M}_N\Vert\geq 0$. We moreover know  that $\boldsymbol{M}_N$ almost surely is  a density operator. With $\boldsymbol{x} := 1-\Vert\boldsymbol{M}_N\Vert $ and $\boldsymbol{y} := \sqrt{D(D-1)}\sqrt{\lambda^{\downarrow}_{1}(\boldsymbol{M}_N)\lambda^{\downarrow}_{2}(\boldsymbol{M}_N)} $, we can use the above observations to conclude that $0\leq \boldsymbol{x} \leq \boldsymbol{y}\,\, a.s.$ Moreover, by Lemma \ref{htdhgtmdg}, we obtain
\begin{equation}
\label{hjmdghmjd}
    1 -E(\Vert\boldsymbol{M}_N\Vert)=E(1-\Vert\boldsymbol{M}_N\Vert)\leq D\sqrt{\frac{D-1}{D}}  E\left(\sqrt{\lambda^{\downarrow}_{1}(\boldsymbol{M}_N)\lambda^{\downarrow}_{2}(\boldsymbol{M}_N)}\right).
\end{equation}
Next we note that the observation in (\ref{fgndfgndgf}) yields
\begin{equation}
\label{hgngfhn}
E\left( \sqrt{\lambda_1^{\downarrow}(\boldsymbol{M}_N) \lambda_2^{\downarrow}(\boldsymbol{M}_N)} \right)
D= E\left(  \frac{   
\nu_1^{\downarrow}\left(\boldsymbol{W}_N\right) \nu_2^{\downarrow}\left(\boldsymbol{W}_N\right)   
}
{
\Tr(\boldsymbol{W}_N^{\dagger}\boldsymbol{W}_N)
}
 \right)D= w(N). 
\end{equation}
By combining (\ref{hjmdghmjd}) and (\ref{hgngfhn}), one obtains $1 -E(\Vert\boldsymbol{M}_N\Vert) \leq  w(N) \sqrt{(D-1)/D}$. 
By the assumption that $\lim_{N\rightarrow\infty}w(N) = 0$, it follows that $\lim_{N\rightarrow\infty}E(\Vert\boldsymbol{M}_N\Vert)  = 1$. By (\ref{dghndghnh}) in Lemma \ref{fnjjnfddanj}, we know that  $E(\Vert \boldsymbol{M}_N\Vert) \rightarrow E(\Vert \boldsymbol{M}_{\infty}\Vert)$. We can thus conclude that $ E(\Vert \boldsymbol{M}_{\infty}\Vert) = 1$. With $\boldsymbol{x}:= 1-\Vert\boldsymbol{M}_{\infty}\Vert$, it follows that $E(\boldsymbol{x}) = 0$. Since $\boldsymbol{M}_{\infty}$ is almost surely a density operator, it follows that $1\geq\Vert \boldsymbol{M}_{\infty}\Vert$ almost surely. Hence, $\boldsymbol{x} = 1-\Vert\boldsymbol{M}_{\infty}\Vert\geq 0\,\, a.s.$ By combining this observation and $E(\boldsymbol{x}) = 0$ with Lemma \ref{ljhvxrushkl}, we obtain $\boldsymbol{x} = 0\,\, a.s.$, and thus  $\Vert\boldsymbol{M}_{\infty}\Vert = 1\hspace{0.15cm}a.s.$ By Lemma \ref{fnjjnfddanj}, we know that $\boldsymbol{M}_{\infty}$ almost surely is a density operator. If $\boldsymbol{M}_{\infty}$ is a density operator, then $\boldsymbol{M}_{\infty}$ is a rank one operator if and only if $\Vert\boldsymbol{M}_{\infty}\Vert = 1$. We can thus conclude that $\mathrm{rank}(\boldsymbol{M}_{\infty}) = 1\,\, a.s.$ By Lemma \ref{ghdnghnhn}, this implies that  $\{A_x\}_{x=0}^{d-1}$ satisfies the purity condition in Definition \ref{DefPur}. Hence, statement 2 implies statement 1.

\end{proof}

\subsection{\label{SecTransition} Generalization to $F$ and $\sigma$}

The entire proof has up to this point concerned the exponential decay of $w(N)$, while we actually wish to find conditions for the exponential decay of $f(N)$. Here, we find necessary as well as sufficient conditions  for  the exponential decay of $f(N)$. We state a slightly more elaborate version of Proposition \ref{PropMain}.

\begin{prop}
\label{bofboneabona} 
Let $\{A_x\}_{x=1}^{d-1}$ be linear operators on the finite-dimensional complex Hilbert space $\mathcal{H}$, such that $\sum_{x=1}^{d-1}A_x^{\dagger}A_x = \1$. For an operator $\sigma$ on $\mathcal{H}$, and an operator $F:\mathcal{H}\rightarrow\mathcal{H}'$ for a  finite-dimensional complex Hilbert space $\mathcal{H}'$, define

\begin{equation}
f(N) := \sum_{x_N,\ldots,x_1= 1}^{d-1}\nu_1^{\downarrow}(FA_{x_N}\cdots A_{x_1}\sqrt{\sigma})\nu_2^{\downarrow}(FA_{x_N}\cdots A_{x_1}\sqrt{\sigma}).
\end{equation}
 If $\{A_x\}_{x=1}^{d-1}$ satisfies the purity condition in Definition \ref{DefPur}, then there exist real constants $ 0\leq \overline{c}$ and $0<\gamma <1$, which satisfy
\begin{equation}
\label{hgmghmg2}
f(N)\leq \overline{c}\gamma^N,\quad \forall N\in\mathbb{N}, 
\end{equation}
for all density operators $\sigma$, and all $F$ such that $F^{\dagger}F\leq \1$.
Conversely, if $f$ is defined with respect to some full-rank operators $\sigma$ and $F^{\dagger}F$, such that there exist constants $0\leq \overline{c}_{\sigma,F}$ and $0<\gamma <1$ which fulfill
\begin{equation}
\label{bdfndnjfkbjlndf}
f(N)\leq \overline{c}_{\sigma,F}\gamma^N,\quad \forall N\in\mathbb{N}, 
\end{equation}
then $\{A_x\}_{x=1}^{d-1}$ satisfies the purity condition.
\end{prop}

\begin{proof}
We begin by proving the first claim of the proposition. 
We first note that 
\begin{equation}
\label{fdbadfbadfb}
\begin{split}
f(N) = & \sum_{x_1,\ldots,x_N}\Vert \mywedge (FA_{x_N}\cdots A_{x_1}\sqrt{\sigma})\Vert, \\
\leq &  \Vert \mywedge (F)\Vert \Vert\mywedge (\sqrt{\sigma})\Vert  w(N),\\
\leq  &  \Vert F\Vert^2 \Vert\sqrt{\sigma}\Vert^2  w(N),\\
\leq & w(N),
\end{split}
\end{equation}
where $w$ is as defined in (\ref{bfgndfgnqsygn}), and where the next to last inequality follows from (\ref{eqBound}). The last inequality follows  since $\sigma$ is assumed to be a density operator, and thus $\Vert \sqrt{\sigma}\Vert \leq 1$, and similarly $F^{\dagger}F\leq \1$ implies $\Vert F\Vert\leq 1$. If $\{A_x\}_{x=0}^{d-1}$ satisfies the purity condition, then it follows by Proposition \ref{wMain} that $w(N)\leq C'\gamma^N$. By combining this observation with (\ref{fdbadfbadfb}), we obtain (\ref{hgmghmg2}) with $\overline{c} := C'$. Note that Proposition \ref{wMain} makes no reference to $F$ or $\sigma$, and thus $C'$, and consequently $\overline{c}$, is independent of these.

Next, we turn to the second claim of the proposition. For this purpose, we first note that since $F^{\dagger}F$ and $\sigma$ (and thus $\sqrt{\sigma}$) are full-rank operators on a finite-dimensional space, it follows that $(F^{\dagger}F)^{-1}$ and $\sqrt{\sigma}^{-1}$ exist. With $w$ as defined in (\ref{bfgndfgnqsygn}), we thus find
\begin{equation}
\begin{split}
w(N) = \sum_{x_1,\ldots,x_N}\Vert \mywedge (A_{x_N}\cdots A_{x_1})\Vert\leq &\Vert \mywedge \big((F^{\dagger}F)^{-1}F^{\dagger}\big)\Vert\Vert\mywedge (\sqrt{\sigma}^{-1})\Vert f(N).
\end{split}
\end{equation}

Hence, with $\overline{c}'_{\sigma,F} :=  \Vert \mywedge \big((F^{\dagger}F)^{-1}F^{\dagger}\big)\Vert\Vert\mywedge (\sqrt{\sigma}^{-1})\Vert$, we get $w(N) \leq \overline{c}'_{\sigma,F} f(N)$. 
Combined with the assumption (\ref{bdfndnjfkbjlndf}),  it follows
\begin{equation}
w(N)\leq \overline{c}_{\sigma,F}\overline{c}'_{\sigma,F}\gamma^N,\quad \forall N\in\mathbb{N},
\end{equation}
where $w$ is as defined in Proposition \ref{wMain}, and $0<\gamma<1$. With $C':= \overline{c}_{\sigma,F}\overline{c}'_{\sigma,F}$ in Proposition \ref{wMain}, it follows that  $\{A_x\}_{x=1}^{d-1}$ satisfies the purity condition in Definition \ref{DefPur}.
\end{proof}

\nocite{*}
\bibliography{bib}

\end{document}